\documentclass[11pt]{article}

\usepackage{amsfonts}
\usepackage{amsmath, amstext, amssymb}
\usepackage{enumerate}
\usepackage{algorithm}
\usepackage{algorithmic}

\usepackage[normalem]{ulem}
\usepackage{graphicx}

\usepackage{bbm}

\usepackage{amsthm}

\usepackage{thmtools}

\usepackage{verbatim} 

\usepackage[usenames,dvipsnames,svgnames,table]{xcolor}
\definecolor{dartmouthgreen}{rgb}{0.05, 0.5, 0.06}
\definecolor{ceruleanblue}{rgb}{0.16, 0.32, 0.75}

\usepackage[pagebackref,letterpaper=true,colorlinks=true,pdfpagemode=UseNone,citecolor=dartmouthgreen,linkcolor=ceruleanblue,urlcolor=BrickRed,pdfstartview=FitH]{hyperref}

\usepackage{framed}
\usepackage{enumitem}


\newtheorem{theorem}{Theorem}[section]
\newtheorem{fact}[theorem]{Fact}
\newtheorem{lemma}[theorem]{Lemma}
\newtheorem{definition}[theorem]{Definition}

\newtheorem{corollary}[theorem]{Corollary}
\newtheorem{proposition}[theorem]{Proposition}

\newtheorem*{problem*}{Problem}
\newtheorem{remark}[theorem]{Remark}
\newtheorem*{remark*}{Remark}

\newcommand{\st}{\mbox{\rm subject to }}

\numberwithin{equation}{section}
\numberwithin{table}{section}

\renewcommand{\preceq}{\preccurlyeq}
\renewcommand{\succeq}{\succcurlyeq}

\newcommand{\R}{\ensuremath{\mathbb R}}

\newcommand{\N}{\ensuremath{\mathbb N}}

\renewcommand{\S}{\ensuremath{\mathbb S}}

\newcommand{\E}[1]{{\mathbb{E}}\left[#1\right]}

\newcommand{\floor}[1]{\ensuremath{\left\lfloor#1\right\rfloor}}

\newcommand{\junk}[1]{}

\renewcommand{\l}{\lambda}

\newcommand{\vol}{{\rm vol}}

\newcommand{\norm}[1]{\left\lVert#1\right\rVert}

\newcommand{\vertiii}[1]{{\left\vert\kern-0.25ex\left\vert\kern-0.25ex\left\vert #1 \right\vert\kern-0.25ex\right\vert\kern-0.25ex\right\vert}}

\newcommand{\one}{\ensuremath{\mathbbm{1}}}

\newenvironment{proofof}[1]{{\medbreak\noindent \em Proof of #1.  }}{\hfill\qed\medbreak}

\def\b1{{\bf 1}}
\def\eps{{\epsilon}}

\def\R{\mathbb{R}}

\newcommand{\directed}{\ensuremath{\overrightarrow}}

\def\vol{\operatorname{vol}} 
\def\diag{\operatorname{diag}} 

\def\tr{\operatorname{tr}}

\global\long\def\E{\mathbb{E}}

\global\long\def\R{\mathbb{R}}

\newcommand{\inner}[2]{\langle #1, #2 \rangle} 
\newcommand{\biginner}[2]{\Big\langle #1, #2 \Big\rangle} 

\DeclareMathOperator{\argmin}{argmin}

\DeclareMathOperator{\supp}{supp}
\DeclareMathOperator{\Var}{Var}

\newif\ifpaper
\papertrue

\title{Cheeger Inequalities for Vertex Expansion and\\ Reweighted Eigenvalues}

\ifpaper
\author{
Tsz Chiu Kwok\footnote{Institute for Theoretical Computer Science, Shanghai University of Finance and Economics. Supported by Science and Technology Innovation 2030 - ``New Generation of Artificial Intelligence'' Major Project No.(2018AAA0100903), NSFC grant 61932002, Program for Innovative Research Team of Shanghai University of Finance and Economics and the Fundamental Research Funds for the Central Universities. 
},~~~~~
Lap Chi Lau\footnote{Cheriton School of Computer Science, University of Waterloo. Supported by NSERC Discovery Grant. 
},~~~~~
Kam Chuen Tung\footnote{Cheriton School of Computer Science, University of Waterloo. Supported by NSERC Discovery Grant. 
}}
\else
\author{}
\fi

\date{}

\begin{document}

\begin{titlepage}
\def\thepage{}
\thispagestyle{empty}

\maketitle

\begin{abstract}
The classical Cheeger's inequality relates the edge conductance $\phi$ of a graph and the second smallest eigenvalue $\lambda_2$ of the Laplacian matrix.
Recently, Olesker-Taylor and Zanetti discovered a
Cheeger-type inequality
$\psi^2 / \log |V| \lesssim \lambda_2^* \lesssim \psi$
connecting the vertex expansion $\psi$ of a graph $G=(V,E)$ and the maximum reweighted second smallest eigenvalue $\lambda_2^*$ of the Laplacian matrix.

\vspace{1mm}

In this work, we first improve their result to 
$\psi^2 / \log d \lesssim \lambda_2^* \lesssim \psi$
where $d$ is the maximum degree in $G$,
which is optimal up to a constant factor.
Also, the improved result holds for weighted vertex expansion, answering an open question by Olesker-Taylor and Zanetti.

\vspace{1mm}

Building on this connection, 
we then develop a new spectral theory for vertex expansion.
We discover that several interesting generalizations of Cheeger inequalities relating edge conductances and eigenvalues have a close analog in relating vertex expansions and reweighted eigenvalues.
These include:
\begin{itemize}
\item An analog of Trevisan's result that relates the bipartite vertex expansion $\psi_B$ of a graph and the maximum reweighted lower spectral gap $\zeta^*$ of the adjacency matrix.
This implies the first approximation algorithm for bipartite vertex expansion.

\item An analog of higher-order Cheeger's inequalities that relates the $k$-way vertex expansion $\psi_k$ of a graph 
and the maximum reweighted $k$-th smallest eigenvalue $\lambda_k^*$ of the Laplacian matrix.
This implies the first approximation algorithm for $k$-way vertex expansion.
\item An analog of improved Cheeger's inequality that relates the vertex expansion $\psi$ and the reweighted eigenvalues $\lambda_2^*$ and $\lambda_k^*$.
This provides an improved bound for $\psi$ using $\lambda_2^*$, when the $k$-way vertex expansion $\psi_k$ is large for a small $k$.
\end{itemize}

Finally, inspired by this connection, we present negative evidence to the $0/1$-polytope edge expansion conjecture by Mihail and Vazirani.  
We construct $0/1$-polytopes whose graphs have very poor vertex expansion.
This implies that the fastest mixing time to the uniform distribution on the vertices of these $0/1$-polytopes is almost linear in the graph size.
This does not provide a counterexample to the conjecture, but this is in contrast with known positive results which proved poly-logarithmic mixing time to the uniform distribution on the vertices of subclasses of $0/1$-polytopes.

\end{abstract}

\end{titlepage}

\thispagestyle{empty}

\newpage

\section{Introduction}

The connection between vertex expansion and reweighted eigenvalue is discovered through the study of the fastest mixing time problem introduced by Boyd, Diaconis and Xiao~\cite{BDX04}.
In the fastest mixing time problem, we are given an undirected graph $G=(V,E)$ and a target probability distribution $\pi: V \to \R$.
The task is to find a time-reversible transition matrix $P \in \R^{|V| \times |V|}$ supported on the edges of $G$,
so that the stationary distribution of random walks with transition matrix $P$ is $\pi$.
The objective is to find such a transition matrix that minimizes the mixing time to the stationary distribution $\pi$.
It is well-known that the mixing time to the stationary distribution is approximately inversely proportional to the spectral gap $1-\alpha_2(P)$ of the time-reversible transition matrix $P$, where $1 = \alpha_1(P) \geq \alpha_2(P) \geq \cdots \geq \alpha_{|V|}(P) \geq -1$ are the eigenvalues of $P$.
The fastest mixing time problem is thus formulated as follows in~\cite{BDX04} by the maximum spectral gap achievable through such a ``reweighting'' $P$ of the input graph $G$.

\begin{definition}[Maximum Reweighted Spectral Gap~\cite{BDX04}] \label{def:BDX-primal}
Given an undirected graph $G=(V,E)$ and a probability distribution $\pi$ on $V$,
the maximum reweighted spectral gap is defined as
\begin{align*}
\lambda_2^*(G) ~:=~ \max_{P \geq 0} &~~~ 1-\alpha_2(P) & 
\\
\st &~~~ P(u,v) = P(v,u) = 0 & & \forall uv \notin E
\\
&~~~ \sum_{v \in V} P(u,v) = 1 & & \forall u \in V
\\
&~~~ \pi(u) P(u,v) = \pi(v) P(v,u) & & \forall uv \in E.
\end{align*}
The graph is assumed to have a self-loop on each vertex, to ensure that the optimization problem for $\lambda_2^*(G)$ is always feasible.
In the context of Markov chains, this corresponds to allowing a non-negative holding probability on each vertex.

The last constraint is the time reversible condition to ensure that the transition matrix $P$ corresponds to random walks on an undirected graph (where the edge weight of $uv$ is $\pi(u) P(u,v)$) and that the stationary distribution of $P$ is $\pi$.
Note that $\lambda_2^*(G) = \max_{P \geq 0} (1 - \alpha_2(P)) = \max_{P \geq 0} \lambda_2(I-P)$, which is the maximum reweighted second smallest eigenvalue of the normalized Laplacian matrix of $G$ (where the edge weight of $uv$ is $\pi(u) P(u,v)$) subject to the above constraints.
\end{definition}

Boyd, Diaconis and Xiao showed that this optimization problem can be written as a semidefinite program and thus $\lambda_2^*(G)$ can be computed in polynomial time.
Subsequently, the fastest mixing time problem has been studied in various work (see~\cite{Roc05,BDSX06,BDPX09,FK13,CA15} and more references in~\cite{OTZ22}), but no general characterization was known.
Roch~\cite{Roc05} showed that the vertex expansion $\psi(G)$ is an upper bound on the optimal spectral gap $\lambda_2^*(G)$.

\begin{definition}[Weighted Vertex Expansion] \label{def:vertex-expansion}
Let $G=(V,E)$ be an undirected graph and $\pi$ be a probability distribution on $V$.
For a subset $S \subseteq V$, let $\partial S := \{v \notin S \mid \exists u \in S \textrm{~with~} uv \in E \}$ be the vertex boundary of $S$, and $\pi(S) := \sum_{v \in S} \pi(v)$ be the weight of $S$.
The weighted vertex expansion of a set $S \subseteq V$ and of a graph $G$ are defined as
\[
\psi(S) := \frac{\pi(\partial S)}{\pi(S)}
\quad \textrm{and} \quad
\psi(G) := \min\Big\{ 1, \min_{S \subseteq V:0 < \pi(S) \leq 1/2} \psi(S)\Big\}.
\footnote{When $\pi$ is the uniform distribution, $\min_{S \subseteq V: 0 < \pi(S) \leq 1/2} \psi(S)$ is always at most $1$.
For general $\pi$, however, this could be arbitrarily large.
Consider for example a star graph where the center has most of the $\pi$-weight.
Therefore, we need to put an upper bound of $1$ on $\psi(G)$, 
as otherwise $\psi(G)$ cannot be bounded by eigenvalues of the normalized Laplacian matrix which are always upper bounded by $2$.}\]
When $\pi$ is the uniform distribution, $\psi(S)$ is the usual vertex expansion $|\partial S|/|S|$.
\end{definition}

Recently, Olesker-Taylor and Zanetti~\cite{OTZ22} discovered an elegant Cheeger-type inequality for vertex expansion and the maximum reweighted spectral gap, showing that small vertex expansion is qualitatively the only obstruction for the fastest mixing time to be small.
Note that their result only holds when $\pi$ is the uniform distribution.

\begin{theorem}[Cheeger Inequality for Vertex Expansion~\cite{OTZ22}] \label{thm:OTZ22}
For any undirected graph $G = (V,E)$ and the uniform distribution $\pi = \vec{1}/|V|$,
\[
\frac{\psi(G)^2}{\log |V|} \lesssim \lambda_2^*(G) \lesssim \psi(G).
\]
In terms of the fastest mixing time $\tau^*(G)$ to the uniform distribution,
$\frac{1}{\psi(G)} \lesssim \tau^*(G) \lesssim \frac{\log^2 |V|}{\psi^2(G)}.$
(See \autoref{sec:prelim} for definitions for random walks and mixing time.)
\end{theorem}

Unlike Cheeger's inequality for edge conductance where $\phi(G)^2 \lesssim \lambda_2(G) \lesssim \phi(G)$, 
it is noted in~\cite{OTZ22} that the $\log |V|$ term might not be completely removed:
Louis, Raghavendra and Vempala~\cite{LRV13} proved that it is NP-hard to distinguish between $\psi(G) \leq \eps$ and $\psi(G) \gtrsim \sqrt{\eps \log d}$ for every $\eps > 0$ where $d$ is the maximum degree of the graph $G$,
assuming the small-set expansion conjecture of Raghavendra and Steurer~\cite{RS10}.

Besides the fastest mixing time problem,
we note that these ``reweighting problems'' relating vertex expansion and reweighted eigenvalues are also well motivated in the study of approximation algorithms.
One example is a conjecture of Arora and Ge~\cite[Conjecture 12]{AG11}, which roughly states that, if a graph $G$ has almost perfect vertex expansion for every set, then there exists a reweighted doubly stochastic matrix $P$ of the adjacency matrix of $G$ so that $P$ has few eigenvalues less than $-\frac{1}{16}$.
They proved that if the conjecture was true, then there is an improved subexponential time algorithm for coloring $3$-colorable graphs.
Another example is a conjecture of Steurer~\cite[Conjecture 9.2]{Ste10}, which is also known to be related to a reweighting problem between vertex expansion and the graph spectrum, that if true would imply an improved subexponential time approximation algorithm for the sparsest cut problem.

\subsection{Our Results}

First we improve and generalize the result of Olesker-Taylor and Zanetti.
Then we build on this new connection to develop a spectral theory for vertex expansion.
Finally we present $0/1$-polytopes with poor vertex expansion and discuss the implications to the $0/1$-polytope expansion conjecture.

\subsubsection{Optimal Cheeger Inequality for Vertex Expansion}

Olesker-Taylor and Zanetti~\cite{OTZ22} posed the problem of reducing the $\log |V|$ factor in \autoref{thm:OTZ22} to $\log d$, and also the problem of generalizing their result to weighted vertex expansion.
Our first result provides a positive answer to these two questions.

\begin{theorem}[Cheeger Inequality for Weighted Vertex Expansion] \label{thm:Cheeger-vertex}
For any undirected graph $G = (V,E)$ with maximum degree $d$ and any probability distribution $\pi$ on $V$,
\[
\frac{\psi(G)^2}{\log d} \lesssim \lambda_2^*(G) \lesssim \psi(G).
\]
In terms of the fastest mixing time $\tau^*(G)$ to the stationary distribution,
$\frac{1}{\psi(G)} \lesssim \tau^*(G) \lesssim \frac{\log d \cdot \log \pi_{\min}^{-1}}{\psi^2(G)}$.
\end{theorem}

In \autoref{sec:tight-example}, we show that the $\log d$ factor in  \autoref{thm:Cheeger-vertex} is optimal, by exhibiting graphs $G$ with $\lambda_2^*(G) \asymp \frac{\psi(G)^2}{\log d}$.
Note that the tightness result does not rely on the small-set expansion hypothesis.

We note that Louis, Raghavendra and Vempala~\cite{LRV13} gave an SDP approximation algorithm for vertex expansion with the same approximation guarantee,
but their SDP is different from and stronger than that in \autoref{def:BDX-primal} (see \autoref{lem:lambda2-vs-sdp-infty}),
and so it does not have the natural interpretation as the reweighted second eigenvalue and does not imply the result on fastest mixing time.
The proof of \autoref{thm:Cheeger-vertex} is based on the techniques in~\cite{LRV13,BHT00}, which we will discuss in detail in \autoref{sec:background-LRV}.

\subsubsection{Maximum Reweighted Lower Spectral Gap and Bipartite Vertex Expansion}
Trevisan~\cite{Tre09} proved that the lower spectral gap $1 + \alpha_{\min}(G)$ of the normalized adjacency matrix of $G=(V,E)$ is small if and only if
there is a subset $S \subseteq V$ which is an almost bipartite component in $G$ with small edge conductance $\phi(S)$.
We define the analogous notions for vertex expansion and for reweighted lower spectral gap.

\begin{definition}[Bipartite Vertex Expansion] 
\label{def:bipartite-vertex-expansion}
Given an undirected graph $G=(V,E)$,
the bipartite vertex expansion of $G$ is defined as
\[
\psi_B(G) := \min\Big\{1, \min_{\emptyset \neq S \subseteq V} \big\{ \psi(S) \mid G[S] {\rm~is~an~induced~bipartite~graph} \big\} \Big\}.
\]
\end{definition}

\begin{definition}[Maximum Reweighted Lower Spectral Gap] \label{def:lambda-max-primal}
Given an undirected graph $G=(V,E)$ and a probability distribution $\pi$ on $V$,
the maximum reweighted lower spectral gap is defined as
\begin{align*}
\zeta^*(G) ~:=~ \max_{P \geq 0} &~~~ \lambda_{\min}(D_P + P) & 
\\
\st &~~~ P(u,v) = P(v,u) = 0 & & \forall uv \notin E
\\
&~~~ \sum_{v \in V} P(u,v) \le 1 & & \forall u \in V
\\
&~~~ \pi(u) P(u,v) = \pi(v) P(v,u) & & \forall uv \in E.
\end{align*}
where $D_P$ is the diagonal matrix of row sums of $P$ such that $D_P(u, u) = \sum_{v \in V} P(u, v)$ for $u \in V$.
We note that this program is slightly different from that in~\autoref{def:BDX-primal}, and the main reason is that self-loops should not be allowed in this problem.
We will explain more about this in \autoref{sec:Cheeger-bipartite}.
\end{definition}

We prove an analog of Trevisan's result that the maximum reweighted lower spectral gap is small if and only if there is an induced bipartite subgraph on $S$ with small vertex expansion $\psi(S)$.

\begin{theorem}[Cheeger Inequality for Bipartite Vertex Expansion] \label{thm:Cheeger-bipartite}
For any undirected graph $G = (V,E)$ with maximum degree $d$ and any probability distribution $\pi$ on $V$,
\[
\frac{\psi_B(G)^2}{\log d} \lesssim \zeta^*(G) \lesssim \psi_B(G).
\]
\end{theorem}

This is the first approximation algorithm for bipartite vertex expansion to our knowledge.
Finding a two-colorable set with small vertex expansion is one of the three ways in Blum's coloring tools~\cite{Blu94} to make progress in designing approximation algorithms for coloring $3$-colorable graphs.
Indeed, it is in this context that Arora and Ge~\cite{AG11} made the reweighting conjecture mentioned in the introduction.
\autoref{thm:Cheeger-bipartite} does not imply anything new about approximating graph coloring, but we hope that it is a step towards answering Arora and Ge's conjecture.

\subsubsection{Higher-Order Cheeger Inequality for Vertex Expansion}

Lee, Oveis Gharan and Trevisan~\cite{LOT12} and Louis, Raghavendra, Tetali and Vempala~\cite{LRTV12} proved the higher-order Cheeger inequalities, which state that the $k$-th smallest eigenvalue $\lambda_k(G)$ of the normalized Laplacian matrix of $G=(V,E)$ is small if and only if the $k$-way edge conductance $\phi_k(G)$ is small.
More precisely, they proved that $\lambda_k(G) \lesssim \phi_k(G) \lesssim k^2 \sqrt{\lambda_k}$ and $\lambda_{\frac{k}{2}}(G) \lesssim \sqrt{\lambda_{k} \log k}$.
We consider the analogous notion of $k$-way vertex expansion.

\begin{definition}[$k$-Way Vertex Expansion] \label{def:k-way-vertex-expansion}
Given an undirected graph $G = (V, E)$ and a probability distribution $\pi$ on $V$, the $k$-way vertex expansion of $G$ is defined as
\[
\psi_k(G) := \min\Big\{1, \min_{S_1, \dots, S_k \subseteq V} \max_{1 \leq i \leq k} \psi(S_i)\Big\},
\]
where the minimum is taken over pairwise disjoint subsets $S_1, \dots, S_k$ of $V$.
\end{definition}

\begin{definition}[Maximum Reweighted $k$-th Smallest Eigenvalue] 
\label{def:reweighted-lambda-k}
Given an undirected graph $G=(V,E)$ and a probability distribution $\pi$ on $V$,
the maximum reweighted $k$-th smallest eigenvalue of the normalized Laplacian matrix of $G$ is defined as $\lambda_k^*(G) := \max_{P \geq 0} \lambda_k(I-P)$, 
where $P$ is subject to the same constraints stated in \autoref{def:BDX-primal}.
\end{definition} 

We prove an analog of higher-order Cheeger inequalities that the maximum reweighted $k$-th smallest eigenvalue is small if and only if the $k$-way vertex expansion is small.
As in previous work~\cite{LOT12,LRTV12}, there is a better approximation guarantee if we consider only $\frac{k}{2}$-way vertex expansion. 

\begin{theorem}[Higher-Order Cheeger Inequality for Vertex Expansion] \label{thm:higher-order-vertex}
For any undirected graph $G = (V,E)$ with maximum degree $d$ and any probability distribution $\pi$ on $V$,
\[
\lambda_k^*(G) \lesssim \psi_k(G) \lesssim k^{\frac{9}{2}} \log k \sqrt{\log d \cdot \lambda_k^*(G)}
\quad {\rm and} \quad
\psi_{\frac{k}{2}}(G) \lesssim \sqrt{k} \log k \sqrt{\log d \cdot \lambda_{k}^*(G)}.
\]
\end{theorem}

Chan, Louis, Tang and Zhang~\cite{CLTZ18} developed a spectral theory for hypergraphs and proved a higher-order Cheeger inequality for hypergraph (edge) expansion.
Through a reduction from vertex expansion to hypergraph expansion, they proved that $\psi_{\frac{k}{2}}(G) \lesssim k^{\frac{5}{2}} \log k \log \log k \cdot \log d \cdot \sqrt{\xi_{k}}$ for graphs with bounded ratio between the maximum degree and the minimum degree, where $\xi_{k} \lesssim \psi_{k}(G)$ is a relaxation for $k$-way vertex expansion.
Compared to their result, \autoref{thm:higher-order-vertex} does not require the assumption about the maximum degree and the minimum degree of $G$, and has a better approximation ratio for $\frac{k}{2}$-way vertex expansion.
Furthermore, \autoref{thm:higher-order-vertex} provides the first true approximation algorithm for $k$-way vertex expansion $\psi_k(G)$ to our knowledge.

\subsubsection{Improved Cheeger Inequality for Vertex Expansion}

Kwok, Lau, Lee, Oveis Gharan, and Trevisan~\cite{KLLOT13} proved an improved Cheeger inequality that $\phi(G) \lesssim k \lambda_2(G) / \sqrt{\lambda_k(G)}$ for any $k \geq 2$.
This shows that $\lambda_2(G)$ is a tighter approximation to $\phi(G)$ when $\lambda_k(G)$ is large for a small $k$.
The result provides an explanation for the good empirical performance of the spectral partitioning algorithm.

We prove an analogous result that if the $\lambda_k^*(G)$ is large for a small $k$, then $\lambda_2^*(G)$ is a tighter approximation to the vertex expansion $\psi(G)$.
The following result is close to the tight result in~\cite{KLLOT13} for edge conductance as we will elaborate in \autoref{rem:tight}.

\begin{theorem}[Improved Cheeger Inequality for Vertex Expansion] \label{thm:improved-Cheeger-vertex}
For any undirected graph $G = (V,E)$ with maximum degree $d$, and for any probability distribution $\pi$ on $V$ and any $k \geq 2$,
\[
\lambda_2^*(G) \lesssim \psi(G) 
\lesssim \frac{k^{\frac{3}{2}} \cdot \lambda_2^*(G) \cdot \log d}{\sqrt{\lambda_k^*(G)}}.
\] 
\end{theorem}

We remark that the reweighting used in $\lambda_2^*(G)$ and $\lambda_k^*(G)$ may be different.
Through \autoref{thm:higher-order-vertex}, we obtain the following corollary that only depends on the graph structure: If the $k$-way vertex expansion $\psi_k(G)$ is large for a small $k$, then $\lambda_2^*(G)$ is a tighter approximation to $\psi(G)$.

\subsubsection{Vertex Expansion of $0/1$-Polytopes} \label{sec:intro-polytope}

Mihail and Vazirani (see~\cite{FM92}) conjectured that the graph $G=(V,E)$ (i.e.~$1$-skeleton) of any $0/1$-polytope is an edge expander, such that $|\delta(S)|/|S| \geq 1$ for every subset $S \subseteq V$ with $|S| \leq |V|/2$, where $\delta(S)$ denotes the set of edges between $S$ and $V \setminus S$.
This conjecture would imply fast mixing time of random walks to the stationary distribution, with applications in designing fast sampling algorithms for many classes of combinatorial objects.
The conjecture is proved to be correct in several cases~\cite{FM92,Kai04,ALOV19}, most notably the recent resolution of the matroid expansion conjecture~\cite{ALOV19} by Anari, Liu, Oveis Gharan and Vinzant.

In all these positive results, the Markov chain can be set up so that the stationary distribution is the uniform distribution, with the mixing time to the stationary distribution poly-logarithmic in the graph size.
Then the fast sampling algorithms can also be used to obtain an approximate counting algorithm on the number of vertices in the given $0/1$-polytope, with poly-logarithmic runtime in the graph size.
Therefore, sampling from the uniform distribution is usually the setting of interest.

Inspired by the connection between fastest mixing time and vertex expansion, we consider a variant of Mihail and Vazirani's conjecture:
Is the graph of every $0/1$-polytope a vertex expander?
Perhaps surprisingly, we show that there are $0/1$-polytopes whose graphs are very poor vertex expanders.

\begin{theorem}[$0/1$-Polytopes with Poor Vertex Expansion] \label{thm:01-vertex-expansion}
Let $\pi$ be the uniform distribution.
For any $k > 2$ and any $n > 2k$ sufficiently large, there is a $0/1$-polytope $Q=Q_{n,k} \subseteq \{0,1\}^n$ with $O(n^k)$ vertices and  
\[
\psi(Q) \lesssim \frac{(4k)^k}{n^{k-2}}.
\]
\end{theorem}

\autoref{thm:01-vertex-expansion} and \autoref{thm:OTZ22} together imply that even the fastest mixing time of the reversible random walks on some $0/1$-polytopes is almost linear in the graph size.

\begin{corollary}[Torpid Mixing to Uniform Distribution] \label{cor:01-slow-uniform}
For any constant $k > 2$, 
there exists a $0/1$-polytope $Q$ such that any reversible Markov chain on its graph $G_Q=(V,E)$ with stationary distribution $\vec{1}/|V|$ has mixing time $\Omega\big(|V|^{1-\frac{2}{k}}\big)$.
\end{corollary}

While \autoref{thm:01-vertex-expansion} does not provide a counterexample to the conjecture of Mihail and Vazirani, it shows that even if the conjecture is true, there are $0/1$-polytopes for which random walks cannot be used for efficient uniform sampling and for efficient approximate counting.

\begin{remark} \label{rem:Gillmann}
  After posting the first version of this paper on arXiv, we recently found out that Gillmann \cite[Chapter 3.2]{Gil07} has already constructed examples of 0/1-polytopes whose graphs have poor vertex expansion. The polytopes $Q \subseteq \{0, 1\}^n$ constructed have $2^{(h(c) + o(1)) n}$ vertices and satisfy
  \[
    \psi(Q) \lesssim 2^{- (h(c) - 2c) n},
  \]
  where $h(x) := -x \log x - (1-x) \log(1-x)$ is the binary entropy function and $c := 1/5$ (correspondingly $h(c) = 0.7219...$). Applying $\autoref{thm:OTZ22}$, this would imply a fastest mixing time bound of $\Omega(|V|^{0.4459...})$.
  By choosing smaller values of $c$, an almost linear fastest mixing time bound can be obtained as in \autoref{cor:01-slow-uniform}.
\end{remark}

\subsection{Related Work}

In this subsection, we review previous spectral approaches for vertex expansion and compare them to the current approach using reweighted eigenvalues.
For previous results about Cheeger's inequalities for edge conductances mentioned in the introduction, they will be discussed in the corresponding technical sections.

{\bf Second Eigenvalue and Vertex Expansion:}
There are classical results in spectral graph theory relating vertex expansions and (ordinary) eigenvalues.
For any graph $G=(V,E)$ with maximum degree $d$,
let $\lambda'_2(G)$ be the second smallest eigenvalue of the (unnormalized) Laplacian matrix, it is known that
\[
\psi(G) \geq \frac{2\lambda'_2(G)}{d+2\lambda'_2(G)}
\quad {\rm and} \quad
\lambda'_2(G) \geq \frac{\psi(G)^2}{4+2\psi(G)^2},
\]
where the first inequality is the ``easy'' direction proved by Tanner~\cite{Tan84} and Alon and Milman~\cite{AM85}, and the second inequality is the ``hard'' direction proved by Alon~\cite{Alo86}.
These imply that $\lambda'_2(G)$ can be used to give an $O(\sqrt{d \cdot \psi(G)})$-approximation algorithm to $\psi(G)$.
Compared to Cheeger's inequality for edge conductance that $\phi(G)^2/2 \leq \lambda_2(G) \leq 2\phi(G)$ where $\lambda_2(G)$ is the second smallest eigenvalue of the normalized Laplacian matrix, there is an extra factor $d$ between the upper and lower bounds.

{\bf Spectral Formulation:}
Bobkov, Houdr\'e and Tetali~\cite{BHT00} defined an interesting ``spectral'' quantity called $\lambda_{\infty}$ (see \autoref{def:lambda-infty}), which satisfies an exact analog of Cheeger's inequality for symmetric vertex expansion:
\[
\frac{1}{2} \psi_{\text{sym}}(G)^2 \leq \lambda_{\infty} \leq 2 \psi_{\text{sym}}(G),
\]
where the symmetric vertex boundary of a set $S \subseteq V$ is defined as $\partial_{\text{sym}}(S) := \partial(S) \cup \partial (V-S)$ and the symmetric vertex expansion of $S$ is defined as $\psi_{\text{sym}}(S) := |\partial_{\text{sym}}(S)| / |S|$, and the symmetric vertex expansion of a graph $G$ is defined as $\psi_{\text{sym}}(G) := \min_{S:|S| \leq |V|/2} \psi_{\text{sym}}(S)$.
However, it is not known how to compute $\lambda_{\infty}$ efficiently, and it is recently shown to be NP-hard to compute $\lambda_{\infty}$ by Farhadi, Louis and Tetali~\cite{FLT20}.

{\bf Semidefinite Programming Relaxations:}
Louis, Raghavendra and Vempala~\cite{LRV13} gave a semidefinite programming relaxation ${\sf sdp}_{\infty}$ for $\lambda_{\infty}$, and proved that for any graph $G=(V,E)$ with maximum degree $d$,
\[
\frac{\psi_{\text{sym}}(G)^2}{\log d} \lesssim {\sf sdp}_{\infty} \lesssim \psi_{\text{sym}}(G).
\]

Then, by constructing a graph $H$ such that $\psi_{\text{sym}}(H) = \Theta(\psi(G))$, they reduce vertex expansion to symmetric vertex expansion and obtain a Cheeger's inequality for $\psi(G)$, one that is of the same form as in \autoref{thm:Cheeger-vertex} for $\lambda_2^*(G)$.
We will show in \autoref{lem:lambda2-vs-sdp-infty} that $\lambda_2^*(G)$ and ${\sf sdp}_{\infty}$ are different and ${\sf sdp}_{\infty}$ is a stronger relaxation such that $\lambda_2^*(G) \leq {\sf sdp}_{\infty}$.

The current best known approximation algorithm for vertex expansion $\psi(G)$ is an $O(\sqrt{\log |V|})$ SDP-based approximation algorithm by Feige, Hajiaghayi and Lee~\cite{FHL08}.
This is an extension of the $O(\sqrt{\log |V|})$ SDP-based approximation algorithm for edge conductance $\phi(G)$ by Arora, Rao, and Vazirani~\cite{ARV09}.
The SDP formulation of~\cite{ARV09} is known to be strictly more powerful than the spectral formulation by the second eigenvalue.

Even though $\lambda_2^*(G)$, ${\sf sdp}_{\infty}$ and the SDP in~\cite{FHL08} are all semidefinite programming relaxations for $\psi(G)$ and satisfy similar inequalities, we note that the approach of using reweighted eigenvalues has some additional features.
One important feature is that $\lambda_2^*(G)$ is closely related to fastest mixing time.
This allows one to develop a spectral theory for vertex expansion that relates (i) vertex expansion, (ii) reweighted eigenvalues and (iii) fastest mixing time,
which parallels the classical spectral graph theory that relates (i) edge conductance, (ii) eigenvalues and (iii) mixing time.
Another feature is that it allows one to extend known generalizations of Cheeger inequalities to the vertex expansion setting, and as a consequence to obtain approximation algorithms for bipartite vertex expansion and $k$-way vertex expansion.

{\bf Spectral Hypergraph Theory:}
Louis~\cite{Lou15} and Chan, Louis, Tang, Zhang~\cite{CLTZ18} developed a spectral theory for hypergraphs.
They defined a continuous time diffusion process on a hypergraph $H=(V,E)$ and used it to define the Laplacian operator and its eigenvalues $\gamma_1 \leq \gamma_2 \leq \ldots \leq \gamma_{|V|}$.
The formulation is similar to the one in~\cite{BHT00} for vertex expansion, and they proved that there is an exact analog of Cheeger's inequality for hypergraphs:
\[
\frac{1}{2} \phi(H) \leq \gamma_2 \leq \sqrt{2 \phi(H)},
\]
where $\phi(H)$ is the hypergraph edge conductance of $H$.
As in~\cite{BHT00}, the quantity $\gamma_2$ is not polynomial time computable, 
and a semidefinite programming relaxation similar to that in~\cite{LRV13} is used to design a $O(\sqrt{\phi(G) \log r})$-approximation algorithm for hypergraph edge conductance where $r$ is the maximum size of a hyperedge.
Using this spectral theory, they prove an analog of higher-order Cheeger inequality for hypergraph edge conductance, and also an approximation algorithm for small-set hypergraph edge conductance. 
Through a reduction from vertex expansion to hypergraph edge conductance,  they obtain an analog of higher-order Cheeger inequality for vertex expansion as mentioned earlier after \autoref{thm:higher-order-vertex} and also an approximation algorithm for small-set vertex expansion.
This theory also relates (i) expansion, (ii) eigenvalues and (iii) mixing time, and so the work in~\cite{Lou15,CLTZ18} is closest to the current work.

Compared to the theory in~\cite{Lou15,CLTZ18} for hypergraphs and for vertex expansion through reduction, we note that the current approach using reweighted eigenvalues is more direct and effective for vertex expansion.
The reduction in~\cite[Fact 3]{CLTZ18} from vertex expansion $\psi(G)$ of graph $G$ with maximum degree $d_{\max}$ and minimum degree $d_{\min}$ to edge conductance $\phi(H)$ only satisfies
\[
d_{\min} \cdot \phi(H) \leq \psi(G) \leq d_{\max} \cdot \phi(H),
\]
and so the approximation ratio depends on the ratio between the maximum degree and the minimum degree.
The current approach using reweighted eigenvalues does not have this dependency and also proves stronger bounds in $k$-way vertex expansion as discussed after Theorem~\ref{thm:higher-order-vertex}.
Also, the definitions of the hypergraph diffusion process and its eigenvalues are quite technically involved and require considerable effort to make rigorous~\cite{CTWZ17}.
We believe that the definitions of reweighted eigenvalues are more intuitive and more closely related to ordinary eigenvalues.
Also, reweighted eigenvalues have close connections to other important problems such as fastest mixing time and the reweighting conjectures in approximation algorithms.

\subsection{Techniques} \label{sec:techniques}

From a technical perspective, the advantage of relating reweighted eigenvalues to vertex expansions is that many ideas relating eigenvalues to edge conductances can be carried over to the new setting.
So, many steps in our proofs are natural extensions of previous arguments, and we focus our discussion here on the new elements.

{\bf Vertex Expansion:}
The proof of \autoref{thm:OTZ22} by Olesker-Taylor and Zanetti is based on the dual characterization of \autoref{def:BDX-primal} in \autoref{prop:Roch-dual}, due to Roch~\cite{Roc05}, and it has two main steps.
In the first step, they used the Johnson-Lindenstrass lemma to project the SDP solution into a $O(\log |V|)$-dimension solution, and then further reduce it to a $1$-dimensional ``spectral'' solution by taking the best coordinate.
This is the step where the $\log |V|$ factor is lost.
In the second step, they introduced an interesting new concept called the ``matching conductance'', and used some combinatorial arguments about greedy matchings for the analysis of Cheeger rounding on Roch's dual program.

In our proof of \autoref{thm:Cheeger-vertex}, we also use Roch's dual characterization and follow the same two steps.
In the first step, we use the Gaussian projection method in~\cite{LRV13} to reduce the SDP solution to a $1$-dimensional solution directly, and adapt their analysis to show that only a factor of $\log d$ is lost.
In the second step, we bypass the concept of matching conductance and do a more traditional analysis of Cheeger rounding as in Bobkov, Houdr\'e and Tetali~\cite{BHT00}.
It turns out that this analysis works smoothly for weighted vertex conductance, while the approach using matching conductance faced some difficulty as described in~\cite{OTZ22}.
A new element in our proof is the introduction of an intermediate dual program using graph orientation, which is important in the analysis of both steps.
In \autoref{sec:Cheeger-vertex}, we will review the background from~\cite{OTZ22, Roc05, LRV13, BHT00} and give a more detailed comparison and overview.

{\bf Bipartite Vertex Expansion:}
The proof of \autoref{thm:Cheeger-bipartite} for bipartite vertex expansion follows closely the proof of \autoref{thm:Cheeger-vertex} and Trevisan's result~\cite{Tre09}, once the correct formulation in \autoref{def:lambda-max-primal} is found.

{\bf Multiway Vertex Expansion:}
For the proof of higher-order Cheeger inequality for vertex expansion in \autoref{thm:higher-order-vertex}, one technical issue is that we do not know of a convex relaxation for the maximum reweighted $k$-th smallest eigenvalue in \autoref{def:reweighted-lambda-k}.
Instead, we define a related quantity $\sigma_k^*(G)$ called the maximum reweighted sum of the $k$ smallest eigenvalues in \autoref{def:reweighted-sum}, which can be written as a semidefinite program.
We show in \autoref{def:sigma-k-dual} that this quantity has a nice dual characterization that satisfies the sub-isotropy condition. 
This allows us to adapt the techniques in~\cite{LOT12} to decompose the SDP solution into $k$ disjointly supported SDP solutions with small objective values, so that we can apply \autoref{thm:Cheeger-vertex} to find $k$ disjoint sets with small vertex expansion.
We will review the background in~\cite{LOT12} needed for the proof in \autoref{sec:embedding}. 

{\bf Improved Cheeger Inequality:}
The proof of improved Cheeger inquality for vertex expansion is similar to that in~\cite{KLLOT13}, which has two main steps.
The first step is to prove that if the $1$-dimensional solution to Roch's dual program is close to a $k$-step function, then Cheeger rounding performs well.
The second step is to prove that if the $1$-dimensional solution to Roch's dual program is far from a $k$-step function, then we can construct an SDP solution to $\sigma_k^*$ with small objective value, which proves that $\lambda_k^*$ is small.
Therefore, if $\lambda_k^*$ is large, then the $1$-dimensional solution must be close to a $k$-step function, and hence Cheeger rounding performs well.
One interesting aspect in this proof is to relate the performance of a rounding algorithm of one SDP (in this case $\lambda_2^*(G)$) to the objective value of another SDP (in this case $\sigma_k^*(G)$).

{\bf Vertex Expansion of $0/1$-Polytopes:}
The examples in \autoref{thm:01-vertex-expansion} for $0/1$-polytope is obtained by a simple probabilistic construction.
The graph of a $0/1$-polytope is defined by the set of points chosen in $\{0,1\}^n$.
Let $L$ be the set of points with $k$ ones, and let $R$ be the set of points with $(n-k)$ ones.
We prove that if we choose a random set $M$ of points with $n/2$ ones and set $|M| \asymp 4^k n^2$, then with high probability there are no edges between $L$ and $R$ in the resulting polytope, and so $M$ is a small vertex separator of $L$ and $R$ where each has $\binom{n}{k}$ points.
The proof is by elementary geometric arguments about the edges of a polytope, and a simple result bounding the number of linear threshold functions in the boolean hypercube $\{0, 1\}^n$.

\subsection{Concurrent Work}

Jain, Pham, and Vuong \cite{JPV22} independently published a proof of \autoref{thm:Cheeger-vertex} for the uniform distribution case.
Their approach is based on a better analysis of dimension reduction for maximum matching, which is quite different from our approach as we bypassed the concept of matching conductance in~\cite{OTZ22}.
\section{Preliminaries} \label{sec:prelim}

Given two functions $f, g$, we use $f \lesssim g$ to denote the existence of a positive constant $c > 0$, such that $f \le c \cdot g$ always holds.
We use $f \asymp g$ to denote $f \lesssim g$ and $g \lesssim f$.
We use $f \approx_{k} g$ to denote $f / g = 1 + o_k(1)$.
For positive integers $k$, we use $[k]$ to denote the set $\{1, 2, \dots, k\}$.
For a function $f: X \rightarrow \R$, $\supp(f)$ denotes the domain subset on which $f$ is nonzero.
For an event $E$, $\one[E]$ denotes the indicator function that is $1$ when $E$ is true and $0$ otherwise.

\subsubsection*{Graphs}

Let $G = (V, E)$ be an undirected graph.
Throughout this paper, we use $n := |V|$ to denote the number of vertices and $m := |E|$ to denote the number of edges in the graph. 
If $uv$ is an edge in $G$, we either write $uv \in E$ or use the notation $u \sim v$.
The degree of a vertex $v$, denoted by $\deg(v)$, is the number of edges incident to $v$.
The maximum degree of a graph is defined as $\max_{v \in V} \deg(v)$.
We usually associate $G$ with a probability distribution $\pi: V \to \R$ on the set of vertices, and we write $\pi(S) := \sum_{v \in S} \pi(v)$ for a subset $S \subseteq V$. We assume without loss that $\pi(u) > 0$ for all $u \in V$.

Let $S \subseteq V$ be a subset of vertices.
The edge boundary of $S$ is defined as $\delta(S) := \{uv \in E \mid u \in S, v \notin S\}$.
The volume of $S$ is defined as $\vol(S) := \sum_{v \in S} \deg(v)$.
The edge conductance of $S$ is defined as $\phi(S) := |\delta(S)|/|S|$.
The vertex boundary of $S$ is defined as $\partial S := \{v \in V \setminus S \mid \exists u \in S \text{ with } uv \in E\}$.
The $\pi$-weighted vertex expansion of $S$ is defined as $\psi(S) := \pi(\partial S) / \pi(S)$, and when $\pi$ is the uniform distribution $\psi(S) = |\partial S|/|S|$ is the usual vertex expansion.
The induced edge set of $S$ is defined as $E[S] := \{uv \in E \mid u \in S \text{ and } v \in S\}$.

Let $G = (V,\directed{E})$ be a directed graph.
If $uv$ is a directed edge in $G$, we either write $uv \in \directed{E}$ or use the notation $u \to v$.
The indegree of a vertex $v$ is defined as $\deg^{\text{in}}(v) := |\{u \in V \mid u \to v\}|$.
We will define some directed analogs of vertex expansion in \autoref{sec:Cheeger-vertex} and \autoref{sec:Cheeger-bipartite}.  
They are not standard definitions and so we defer them to the respective sections.

\subsubsection*{Linear Algebra}

Let $M \in \R^{n \times n}$ be a matrix. 
When $M$ is symmetric, the spectral theorem states that $M$ admits an orthonormal eigendecomposition $M = UDU^{-1}$, where $D$ is a diagonal matrix and $U$ is a unitary matrix such that $U^{-1} U = I_n$ where $I_n$ is the $n \times n$ identity matrix.

Two matrices $M, N \in \R^{n \times n}$ are said to be cospectral if they are both diagonalizable, and their eigenvalues are the same. 
There are two cases of cospectral matrices that we will use.

\begin{fact} \label{fact:similar}
Let $M, N \in \R^{n \times n}$. Suppose that $M$ is diagonalizable and that $M$ and $N$ are similar (i.e. $M = X^{-1}NX$ for some invertible matrix $X \in \R^{n \times n}$). Then, $N$ is also diagonalizable, and $M$ and $N$ are cospectral.
\end{fact}

\begin{fact} \label{fact:AB=BA}
Let $M, N \in \R^{n \times n}$. Suppose that there exist $A, B \in \R^{n \times n}$ such that $M = AB$ and $N = BA$. 
If $M$ is diagonalizable, then $N$ is also diagonalizable, and $M$ and $N$ are cospectral.
\end{fact}

Given that $M$ is symmetric, we say that $M$ is positive semidefinite (PSD) if $v^T M v \ge 0$ for all $v \in \R^n$, and we write $M \succeq 0$. 
Equivalently, $M$ is PSD if all its eigenvalues are nonnegative.
Also equivalently, $M$ is PSD if there exists $X$ such that $M = X^T X$.
Let $x_i \in \R^n$ be the $i$-th column of $X$.
Then $M$ is called the Gram matrix of $x_1, \ldots, x_n \in \R^n$ as $M(i,j) = \inner{x_i}{x_j}$ for all $i,j \in [n]$.

The trace of a matrix $M \in \R^{n \times n}$ is defined as $\tr(M) := \sum_{i=1}^n M(i,i)$.
We will often use the fact that $\tr(AB) = \tr(BA)$ for two matrices of compatible dimensions.

\subsubsection*{Random Walks}

Given a finite state space $X$, a Markov chain on $X$ is represented by a matrix $P \in \R^{X \times X}$, where $P(u, v)$ is the probability of traversing from state $u$ to state $v$ in one step. 
Thus, $P$ has nonnegative entries and satisfies $\sum_{v \in X} P(u, v) = 1$ for all $u \in X$. 
A distribution $\pi: X \rightarrow \R$ is said to be a stationary distribution of $P$ if $\pi^T P = \pi^T$. 

A transition matrix $P$ is said to be time-reversible with respect to $\pi$ if $\pi(u) P(u, v) = \pi(v) P(v, u)$ for any $u, v \in X$.
Note that this implies that $\pi$ is a stationary distribution of $P$.
The time reversibility condition can be written as $\Pi P = P^T \Pi$, where $\Pi := \diag(\pi)$. Thus, $\Pi^{1/2}P\Pi^{-1/2}$ is symmetric, hence diagonalizable with eigenvalues $1 = \alpha_1(P) \ge \alpha_2(P) \ge \dots \ge \alpha_n(P) \ge -1$. 
As $P$ is similar to $\Pi^{1/2} P \Pi^{-1/2}$, they have the same eigenvalues by \autoref{fact:similar}.
The spectral gap of $P$ is defined as $1 - \alpha_2(P)$. 

For $\eps \in (0, 1)$, we define the $\eps$-mixing time $\tau_{\rm{mix}}(P, \eps)$ of $P$ to be the smallest $t \in \N$ such that $d_{TV}(\pi, \rho) \le \eps$ for any initial distribution $\rho$. Here, $d_{TV}(\cdot, \cdot)$ is the total variation distance, defined as $d_{TV}(\rho_1, \rho_2) := \max_{S \subseteq V}|\rho_1(S) - \rho_2(S)|$ for any two distributions $\rho_1, \rho_2: X \rightarrow \R_{\ge 0}$.
The relaxation time $\tau_{\rm{rel}}(P)$ of $P$ is defined as the reciprocal of the spectral gap, so $\tau_{\rm{rel}}(P) := \frac{1}{1 - \alpha_2(P)}$.
Let $\pi_{\min} := \min_{u \in V} \pi(u)$. 
It is known that (see e.g. Chapter 12 of \cite{LP17})
\[
  \big( \tau_{\rm{rel}}(P) - 1 \big) \cdot \log \frac{1}{2 \eps} \le \tau_{\rm{mix}}(P, \eps) \le \tau_{\rm{rel}}(P) \cdot \log \frac{1}{\eps \cdot \pi_{\min}}.
\]
Because of this connection between the spectral gap and the mixing time of $P$, the optimization problem of maximizing the spectral gap of the random walk matrix is referred to as ``fastest mixing time'' in~\cite{BDX04}.

\subsubsection*{Spectral Graph Theory}

Given a graph $G = (V, E)$, its adjacency matrix $A = A(G)$ is a $n \times n$ matrix where the $(u, v)$-th entry is $\mathbbm{1}_{uv \in E}$. 
The Laplacian matrix is defined as $L := D - A$, where $D:=\diag(\{\deg(v)\}_{v \in V})$ is the diagonal degree matrix. 
For a vector $x \in \R^n$, the Laplacian matrix has a useful quadratic form $x^T L x = \sum_{uv \in E} \big(x(u)-x(v)\big)^2$.

The normalized adjacency matrix is defined as $\mathcal{A} = D^{-1/2} A D^{-1/2}$, and the normalized Laplacian matrix is defined as $\mathcal{L} := I - \mathcal{A}$. 
Observe that $\mathcal{A}$ is similar to the simple random walk matrix on $G$, so it is diagonalizable with eigenvalues $1 = \alpha_1({\cal A}) \geq \alpha_2({\cal A}) \geq \cdots \geq \alpha_n({\cal A}) \geq -1$.
Therefore, $\mathcal{L}$ is diagonalizable, and its eigenvalues are $0 = \lambda_1({\cal L}) \le \lambda_2({\cal L}) \le \dots \le \lambda_n({\cal L}) \le 2$.
Note that we use $\alpha_i$ to denote the eigenvalues of the normalized adjacency matrix ${\cal A}$ and random walk matrix $P$, and we use $\lambda_i$ to denote the eigenvalues of the normalized Laplacian matrix ${\cal L}$.

Let $\phi(G) := \min_{S \subseteq V: 0 < \pi(S) \le 1/2} \phi(S)$ be the edge conductance of the graph $G$.
Cheeger's inequality~\cite{Che70,AM85,Alo86} states that
\[
  \frac{\lambda_2}{2} \le \phi(G) \le \sqrt{2 \lambda_2}.
\]
This theorem is important because it connects (i) the spectral gap of the normalized Laplacian matrix, (ii) the edge conductance of the graph and (iii) the mixing time of random walks.

\subsubsection*{Convex Optimization}

Consider optimization programs in the standard form
\begin{eqnarray*}
  \mathcal{P} :=
  & \min_{x \in \Omega} & 
  f(x)
  \\
  & \st & g_i(x) \le 0 \quad \forall i \in [l] \\
  & & h_j(x) = 0 \quad \forall i \in [p]
\end{eqnarray*}
where $\Omega \subseteq \R^n$ and $f, g_i, h_j: \R^n \rightarrow \R \cup \{ \pm \infty \}$. To define its Lagrangian dual, consider
\[
  \Lambda(x, \lambda, \mu) := f(x) + \sum_{i \in [l]} \lambda_i g_i(x) + \sum_{j \in [p]} \mu_j h_j(x)
\]
defined on $\Omega \times \R^l \times \R^p$. The dual program is
\[
  \mathcal{D} := \max_{\lambda \ge 0, \mu} \inf_{x \in \Omega} \Lambda(x, \lambda, \mu).
\]
Weak duality always holds, that is, $\mathcal{D} \le \mathcal{P}$. 
We say that strong duality holds if $\mathcal{D} = \mathcal{P}$.

Linear programs are optimization programs where $\Omega = \R^n$ and $f, g_i, h_j$ are all affine functions. It is well-known that strong duality always holds for linear programs.

Semidefinite programs (SDP) are optimization programs where the ambient space is $\R^{n \times n}$, $\Omega := \{X \in \R^{n \times n}: X \succeq 0\}$, and $f, g_i, h_j$ are all affine functions. 
Unlike linear programs, there are SDP's where strong duality does not hold.
We will use Von Neumann's minimax theorem in \autoref{sec:higher-order-vertex} to establish strong duality for the SDP for multiway vertex expansion. 

\begin{theorem}[Von Neumann's Minimax Theorem (see \cite{Sim95})] \label{thm:von-Neumann}
Let $X,Y$ be compact convex sets.
If $f$ is a real-valued continuous function on $X \times Y$ with 
$f(x,\cdot)$ concave on $Y$ for all $x \in X$
and $f(\cdot,y)$ convex on $X$ for all $y \in Y$,
then 
\[
\min_{x \in X} \max_{y \in Y} f(x,y) = \max_{y \in Y} \min_{x \in X} f(x,y).
\]
\end{theorem}

Several eigenvalue optimization problems can be formulated as semidefinite programs.
Boyd, Diaconis and Xiao~\cite{BDX04} showed that the maximum reweighted spectral gap in \autoref{def:BDX-primal} can be written as a semidefinite program.
Roch~\cite{Roc05} showed that the maximum reweighted lower spectral gap problem in \autoref{def:lambda-max-primal} can be written as a semidefinite program; see \autoref{prop:Roch-dual-bipartite}.
For the higher-order Cheeger inequality for vertex expansion,
we will use the following proposition in writing the maximum reweighted sum of $k$ smallest eigenvalue problem in \autoref{def:reweighted-sum} as a semidefinite program.

\begin{proposition} \label{prop:sum-of-lambda-k}
Let $X \in \R^{n \times n}$ be a symmetric matrix and let $1 \le k \le n$. 
Suppose the eigenvalues of $X$ are $\lambda_1 \le \lambda_2 \le \dots \le \lambda_n$. Then, $\lambda_1 + \lambda_2 + \cdots + \lambda_k$ is the value of the following semidefinite program:
  \begin{eqnarray*}
    \min_{Y \in \R^{n \times n}} && \tr(XY) \\
    \st && 0 \preceq Y \preceq I_n \\
    && \tr(Y) = k.
  \end{eqnarray*}    
\end{proposition}

\subsubsection*{Polytopes}

A polytope $\emptyset \neq Q \subseteq \R^n$ is a solution set to affine inequalities, i.e. $Q = \{x \in \R^n: Ax \le b\}$ for some $A \in \R^{m \times n}$ and  $b \in \R^m$.
Given a point set $X \subseteq \R^n$, its convex hull ${\rm conv}(X) \subset \R^n$ is the set of all convex combinations of points in $X$. 
Equivalently, ${\rm conv}(X)$ is the smallest convex set containing $X$. 
A basic result is that the convex hull of a finite point set in $\R^n$ is a polytope.

Given a polytope $Q \subseteq \R^n$. A subset $F \subseteq Q$ is a face of $Q$ if, for any $x, y \in Q$, if $tx + (1-t)y \in F$ for some $t \in (0, 1)$ then $x, y \in F$. If $F$ is a face of $Q$, then it is the intersection of $Q$ and an affine subspace $Y \subseteq \R^n$. The dimension of a face $F$ is the dimension of the smallest (by inclusion) affine subsapce $Y \subseteq \R^n$ such that $F = Q \cap Y$.

Dimension-$0$ faces are called extreme points or vertices, whereas dimension-$1$ faces are called edges. 
The graph $G_Q=(V,E)$ of $Q$ has the dimension-$0$ faces as the vertices and the dimension-$1$ faces as the edges.
This is also called the $1$-skeleton of the polytope.
The following fact about the non-existence of edges between two vertices of a polytope will be useful.

\begin{proposition}[\cite{KR03}]
Let $Q \subseteq \R^n$ be a bounded polytope with vertex set $F_0 = F_0(Q)$. Let $x \neq y \in F_0$, and let $L(x, y)$ be the line segment with endpoints $x$ and $y$. Suppose that ${\rm conv}\big(F_0 \setminus \{x,y\}\big) \cap L(x, y)$ is nonempty. Then, $L(x, y)$ is not an edge of $Q$.
\end{proposition}

We will also use the following version of hyperplane separation theorem.

\begin{proposition} \label{prop:separation-thm}
  Let $M \subseteq \R^n$ be convex and let $x \in \R^n$ be such that $x \not\in M$. Then there exists an affine function $l: \R^n \rightarrow \R$ such that $l(x) = 0$ and $l(y) < 0$ for all $y \in M$.
\end{proposition}
\section{Optimal Cheeger Inequality for Vertex Expansion} \label{sec:Cheeger-vertex}

The goal of this section is to prove \autoref{thm:Cheeger-vertex}.
We will first review the proofs in~\cite{OTZ22,LRV13} in \autoref{sec:background-Cheeger},
and then present how to combine their proofs with a graph orientation idea to prove \autoref{thm:Cheeger-vertex} in \autoref{sec:proof-Cheeger}.

\subsection{Background} \label{sec:background-Cheeger}

We will first review the proofs by Olesker-Taylor and Zanetti~\cite{OTZ22} in \autoref{sec:background-OTZ}, and then the proofs by Louis, Raghavendra and Vempala in \cite{LRV13} in \autoref{sec:background-LRV}.

In this subsection, the stationary distribution $\pi$ is assumed to be the uniform distribution.
This will slightly simplify the presentation and was also the setting considered in previous works.

\subsubsection{Review of~\cite{OTZ22}} \label{sec:background-OTZ}

Recall the fastest mixing time problem formulated in \autoref{def:BDX-primal}.
When $\pi$ is the uniform distribution, the problem is to find a doubly stochastic reweighted matrix $P$ of $G$ that minimizes the second largest eigenvalue of $P$.

The starting point is the following dual characterization of the primal program in \autoref{def:BDX-primal} obtained by Roch~\cite{Roc05},
which is stated in the form for a general distribution $\pi$ that we will use.

\begin{proposition}[Dual Program for Fastest Mixing~\cite{Roc05,OTZ22}] \label{prop:Roch-dual}
Given an undirected graph $G=(V,E)$ and a probability distribution $\pi$ on $V$,
the following semidefinite program is dual to the primal program in \autoref{def:BDX-primal} with strong duality $\lambda_2^*(G) = \gamma(G)$ where
\begin{eqnarray*}
\gamma(G) := \min_{f: V \to \R^n,~g: V \to \R_{\geq 0}} & & \sum_{v \in V} \pi(v) g(v)
\\
\st & & {\sum_{v \in V} \pi(v) \norm{f(v)}^2} = 1
\\
& & \sum_{v \in V} \pi(v) f(v) = \vec{0}
\\
& & g(u) + g(v) \geq \norm{f(u) - f(v)}^2 \quad \quad \forall uv \in E.
\end{eqnarray*}
\end{proposition}

We note that this is equivalent to the dual program given in~\cite{BDX04}, but Roch's program is written in a vector program form that will be more convenient for rounding.
In \autoref{sec:higher-order-vertex}, we will use von Neumann's minimax theorem to derive a generalization of \autoref{prop:Roch-dual} for proving the higher-order Cheeger inequality for vertex expansion.

The proof of \autoref{thm:OTZ22} has two main steps.
The first step is to project the above dual program to the following one-dimensional ``spectral'' program.

\begin{definition}[One-Dimensional Dual Program for Fastest Mixing~\cite{OTZ22}] \label{def:Roch-dual-1D}
Given an undirected graph $G=(V,E)$ and a probability distribution $\pi$ on $V$, $\gamma^{(1)}(G)$ is defined to be the program:
\begin{eqnarray*}
  \gamma^{(1)}(G) := \min_{f: V \to \R,~g: V \to \R_{\geq 0}} & & \sum_{v \in V} \pi(v) g(v) 
\\
\st & & \sum_{v \in V} \pi(v) f(v)^2 = 1
\\
& & \sum_{v \in V} \pi(v) f(v) = 0
\\
  & & g(u) + g(v) \geq (f(u) - f(v))^2 \quad \quad \forall uv \in E.
\end{eqnarray*}
\end{definition}

Olesker-Taylor and Zanetti use the Johnson-Lindenstrauss lemma to first project the solution in \autoref{prop:Roch-dual} to $O(\log n)$ dimensions with constant distortion, and then take the best coordinate to obtain a $1$-dimensional solution with the following guarantee.
Note that this step works for any probability distribution $\pi$ on $V$.

\begin{proposition}[\cite{OTZ22}, Proposition 2.9] \label{prop:projection}
For any undirected graph $G=(V,E)$ and any probability distribution $\pi$ on $V$,
\[
\gamma(G) \leq \gamma^{(1)}(G) \lesssim \log |V| \cdot \gamma(G). 
\]
\end{proposition}

In the second step, Olesker-Taylor and Zanetti observed that the dual program in \autoref{def:Roch-dual-1D} is similar to the weighted vertex cover problem with edge weights $(f(u)-f(v))^2$ for each edge $uv \in E$, which is equivalent to the fractional matching problem by linear programming duality.
To analyze \autoref{def:Roch-dual-1D}, they introduced an interesting new concept called ``matching conductance'', and used some combinatorial arguments about greedy matching as well as some spectral arguments to prove the following Cheeger-type inequality.

\begin{theorem}[\cite{OTZ22}, Theorem 2.10] \label{thm:OTZ-Cheeger}
For any undirected graph $G=(V,E)$ and the uniform distribution $\pi = \vec{1}/|V|$,
\[
\psi(G)^2 \lesssim \gamma^{(1)}(G) \lesssim \psi(G).
\]
\end{theorem}

Combining \autoref{prop:Roch-dual} and \autoref{prop:projection} and \autoref{thm:OTZ-Cheeger} gives 
\[
\psi(G)^2 \lesssim \gamma^{(1)}(G) \lesssim \log |V| \cdot \gamma(G) = \log |V| \cdot \lambda_2^*(G)
\quad {\rm and} \quad
\lambda_2^*(G) = \gamma(G) \leq \gamma^{(1)}(G) \lesssim \psi(G),
\]
proving \autoref{thm:OTZ22}.

Note that the proof of the second step only works when $\pi$ is the uniform distribution.  
Olesker-Taylor and Zanetti discussed some difficulty in generalizing their combinatorial arguments to the weighted setting, and left it as an open question to prove \autoref{thm:OTZ-Cheeger} for any probability distribution $\pi$.

\subsubsection{Review of~\cite{LRV13}} \label{sec:background-LRV}

Our proof is based on the techniques in~\cite{LRV13} which we review here.
Their algorithm is based on the following ``spectral'' formulation $\lambda_{\infty}$ by Bobkov, Houdr\'e and Tetali~\cite{BHT00}, which is for the uniform distribution $\pi$.

\begin{definition}[$\lambda_{\infty}$ in~\cite{BHT00}] \label{def:lambda-infty}
Given an undirected graph $G = (V, E)$, 
\begin{eqnarray*}
\lambda_{\infty}(G) :=
\min_{x: V \to \R,~x \perp \vec{1}} 
\frac{\sum_{u \in V} \max_{v: (v, u) \in E}~(x(u) - x(v))^2}{\sum_{u \in V} x(u)^2}.
\end{eqnarray*}
\end{definition}

Bobkov, Houdr\'e and Tetali~\cite{BHT00} proved an exact analog of Cheeger's inequality for symmetric vertex expansion that $\frac12 \psi_{\text{sym}}(G)^2 \leq \lambda{\infty}(G) \leq 2 \psi_{\text{sym}}(G)$.
We will use some of their arguments to prove a similar statement in \autoref{thm:weighted-Cheeger} in \autoref{sec:proof-Cheeger}.

The issue is that $\lambda_{\infty}$ is not known to be efficiently computable, and indeed recently Farhadi, Louis and Tetali~\cite{FLT20} proved that it is NP-hard to compute $\lambda_{\infty}(G)$ exactly.
To design an approximation algorithm for $\psi(G)$, Louis, Raghavendra and Vempala~\cite{LRV13} defined the following semidefinite programming relaxation for $\lambda_{\infty}$, which we denote by ${\sf sdp}_{\infty}$.

\begin{definition}[${\sf sdp}_{\infty}$ in~\cite{LRV13}] \label{def:sdp-infty}
Given an undirected graph $G = (V, E)$, 
\begin{eqnarray*}
{\sf sdp}_{\infty}(G) :=
\min_{f: V \to \R^n,~g: V \to \R} & &
\sum_{v \in V} g(v)
\\
\st & & \sum_{v \in V} \norm{f(v)}^2 = 1 
\\
& & \sum_{v \in V} f(v) = \vec{0}
\\
& & g(v) \geq \norm{f(u) - f(v)}^2	\quad \quad \forall u \in V {\rm~with~} uv \in E.
\end{eqnarray*}
\end{definition}

The rounding algorithm in~\cite{LRV13} is to project the solution to ${\sf sdp}_{\infty}$ into a $1$-dimensional solution by setting $x(v) = \inner{f(v)}{h}$ where $h \sim N(0,1)^n$ is a random Gaussian vector.
They proved that the $1$-dimensional solution is a $O(\log d)$-approximation to ${\sf sdp}_{\infty}$ where $d$ is the maximum degree of the graph.

\begin{theorem}[\cite{LRV13}, Lemma 9.6] \label{thm:sdp-vs-lambda}
For any undirected graph $G=(V,E)$ with maximum degree $d$,
\[
{\sf sdp}_{\infty}(G) \leq \lambda_{\infty} \lesssim \log d \cdot {\sf sdp}_{\infty}(G).
\]
\end{theorem}

For the analysis, they used the following properties of Gaussian random variables, for which we will also use in our proofs and so we state them here.
The first fact is for the analysis of the numerator and the second fact is for the analysis of the denominator of $\lambda_{\infty}$.

\begin{fact}[\cite{LRV13}, Fact 9.7] \label{Gaussian-expected-max}
Let $Y_1, Y_2, \ldots, Y_d$ be $d$ Gaussian random variables with mean $0$ and variance at most $\sigma^2$.
Let $Y$ be the random variable defined as $Y := \max\{Y_i \mid i \in [d]\}$.
Then 
\[
\E[Y] \leq 2 \sigma \sqrt{\log d}.
\]
\end{fact}

\begin{fact}[\cite{LRV13}, Lemma 9.8] \label{Gaussian-denominator}
Suppose $z_1, \ldots, z_m$ are Gaussian random variables (not necessarily independent) such that $\E[\sum_{i=1}^m z_i^2] = 1$.
Then 
\[
\Pr\left[ \sum_{i=1}^m z_i^2 \geq \frac{1}{2} \right] \geq \frac{1}{12}.
\]
\end{fact}

\subsection{Proof of \autoref{thm:Cheeger-vertex}} \label{sec:proof-Cheeger}

We follow the same two-step plan as in~\cite{OTZ22}.
We will prove in \autoref{prop:improved-projection} in \autoref{sec:projection} that $\gamma^{(1)}(G) \lesssim \gamma(G) \cdot \log d$ for any probability distribution $\pi$.
Note that this already improves \autoref{thm:OTZ22} to the optimal bound,
when $\pi$ is the uniform distribution.
Then, we will prove in \autoref{thm:weighted-Cheeger} in \autoref{sec:Cheeger-rounding} that $\psi(G)^2 \lesssim \gamma^{(1)}(G) \lesssim \psi(G)$ for any probability distribution $\pi$ on $V$.
As in~\cite{OTZ22}, combining \autoref{prop:Roch-dual} and \autoref{prop:improved-projection} and \autoref{thm:weighted-Cheeger} gives \autoref{thm:Cheeger-vertex}.

\subsubsection{Dual Program on Graph Orientation} \label{sec:dual-orientation}

To extend the techniques in~\cite{LRV13,BHT00} to prove the two steps,
we will introduce a ``directed'' program $\directed{\gamma}(G)$ to bring $\gamma(G)$ in \autoref{prop:Roch-dual} closer to ${\sf sdp}_{\infty}(G)$ in \autoref{def:sdp-infty}.

Observe that the two SDP programs $\gamma(G)$ and ${\sf sdp}_{\infty}(G)$ have very similar form.
The only difference is that the last constraint in \autoref{prop:Roch-dual} only requires that $g(u) + g(v) \geq \norm{f(u) - f(v)}^2$ for $uv \in E$, while the last constraint in \autoref{def:sdp-infty} has a stronger requirement that $\min\{g(u),g(v)\} \geq \norm{f(u)-f(v)}^2$ for $uv \in E$.
So ${\sf sdp}_{\infty}$ is a stronger relaxation than $\gamma(G) = \lambda_2^*(G)$.

\begin{lemma}
\label{lem:lambda2-vs-sdp-infty}
For any undirected graph $G=(V,E)$ and any probability distribution $\pi$ on $V$,
\[\lambda_2^*(G) \leq {\sf sdp}_{\infty}(G).\]
\end{lemma}

For our analysis of $\lambda_2^*(G)$, we consider the following ``directed'' program $\directed{\gamma}(G)$ where the last constraint is $\max\{g(u),g(v)\} \geq \norm{f(u)-f(v)}^2$ for $uv \in E$.
We also state the corresponding $1$-dimensional version as in \autoref{def:Roch-dual-1D} in the following definition.

\begin{definition}[Directed Dual Programs for $\gamma(G)$] \label{def:directed-gamma}
Given an undirected graph $G=(V,E)$ and a probability distribution $\pi$ on $V$,
\begin{eqnarray*}
\directed{\gamma}(G) := \min_{f: V \to \R^n,~g: V \to \R_{\geq 0}} & & \sum_{v \in V} \pi(v) g(v)
\\
\st & & {\sum_{v \in V} \pi(v) \norm{f(v)}^2} = 1
\\
& & \sum_{v \in V} \pi(v) f(v) = \vec{0}
\\
& & \max\{g(u),g(v)\} \geq \norm{f(u) - f(v)}^2 \quad \quad \forall uv \in E.
\end{eqnarray*}
$\directed{\gamma}^{(1)}(G)$ is defined as the $1$-dimensional program of $\directed{\gamma}(G)$ where $f:V \to \R$ instead of $f: V \to \R^n$.
\end{definition}

Note that $\directed{\gamma}(G)$ is not a semidefinite program because of the max constraint, but $\gamma(G)$ and $\directed{\gamma}(G)$ are closely related and $\directed{\gamma}(G)$ is only used in the analysis as a proxy for $\gamma(G)$.

\begin{lemma} \label{lem:factor2}
For any undirected graph $G=(V,E)$ and any probability distribution $\pi$ on $V$,
\[
\gamma(G) \leq \directed{\gamma}(G) \leq 2\gamma(G)
\quad {\rm and} \quad
\gamma^{(1)}(G) \leq \directed{\gamma}^{(1)}(G) \leq 2\gamma^{(1)}(G).
\]
\end{lemma}
\begin{proof}
As $g \geq 0$, any feasible solution $f,g$ to $\directed{\gamma}(G)$ is a feasible solution to $\gamma(G)$ and so the first inequalities follow.
On the other hand, for any feasible solution $f,g$ to $\gamma(G)$, note that $f,2g$ is a feasible solution to $\directed{\gamma}(G)$ and so the second inequalities follow.
\end{proof}

The reason that we call $\directed{\gamma}(G)$ the ``directed'' program is as follows.
For each edge $uv \in E$, the constraint in ${\sf sdp}_{\infty}(G)$ requires both $g(u)$ and $g(v)$ to be at least $\norm{f(u)-f(v)}^2$, while the constraint in $\directed{\gamma}(G)$ only requires at least one of $g(u)$ or $g(v)$ to be at least $\norm{f(u)-f(v)}^2$. 
We think of $\directed{\gamma}(G)$ as assigning a direction to each edge and requiring that $g(v) \geq \norm{f(u)-f(v)}^2$ for each directed edge $u \to v$.
Then, we can rewrite the programs $\directed{\gamma}(G)$ and $\directed{\gamma}^{(1)}(G)$ by eliminating the variables $g(v)$ for $v \in V$, by minimizing over all possible orientations of the edge set $E$.
 
\begin{lemma}[Directed Dual Programs Using Orientation for $\gamma(G)$] \label{def:compact}
Let $G=(V,E)$ be an undirected graph and $\pi$ be a probability distribution on $V$.
Let $\directed{E}$ be an orientation of the undirected edges in $E$.
Then 
\begin{eqnarray*}
\directed{\gamma}(G) = \min_{f: V \to \R^n} \min_{\directed{E}} & & \sum_{v \in V} \pi(v) \max_{u: uv \in \directed{E}} \norm{f(u)-f(v)}^2
\\
\st & & {\sum_{v \in V} \pi(v) \norm{f(v)}^2} = 1
\\
& & \sum_{v \in V} \pi(v) f(v) = \vec{0}.
\end{eqnarray*}
Similarly, $\directed{\gamma}^{(1)}(G)$ can be written in the same form with $f:V \to \R$ instead of $f: V \to \R^n$.
\end{lemma}
\begin{proof}
In one direction, given an orientation $\directed{E}$, we can define $g(v) := \max_{u:uv \in \directed{E}} \norm{f(u)-f(v)}^2$, so that $f,g$ is a feasible solution to $\directed{\gamma}(G)$ as stated in \autoref{def:directed-gamma} with the same objective value.

In the other direction, given a solution $f,g$ in \autoref{def:directed-gamma}, we can define an orientation $\directed{E}$ of $E$ so that each directed edge $uv$ satisfies $g(v) \geq \norm{f(u)-f(v)}^2$.
Note that $g(v) \geq \max_{u:uv \in \directed{E}} \norm{f(u)-f(v)}^2$, and setting it to be an equality would satisfy all the constraints and not increase the objective value as $g \geq 0$.
\end{proof}

This formulation will be useful in both the Gaussian projection step for \autoref{prop:improved-projection} and the threshold rounding step for \autoref{thm:weighted-Cheeger}.

\subsubsection{Gaussian Projection} \label{sec:projection}

The following proposition is an improvement of \autoref{prop:projection} in~\cite{OTZ22}.
The formulation in \autoref{def:compact} allows us to use the expected maximum of Gaussian random variables in \autoref{Gaussian-expected-max} to analyze the projection as was done in~\cite{LRV13}.

\begin{proposition}[Gaussian Projection for $\gamma(G)$] \label{prop:improved-projection}
For any undirected graph $G=(V,E)$ with maximum degree $d$ and any probability distribution $\pi$ on $V$,
\[
\gamma(G) \leq \gamma^{(1)}(G) \lesssim \gamma(G) \cdot \log d. 
\]
\end{proposition}

\begin{proof}
We will prove that $\directed{\gamma}(G) \leq \directed{\gamma}^{(1)}(G) \lesssim \log d \cdot \directed{\gamma}(G)$, and the proposition will follow from \autoref{lem:factor2}.
The first inequality is immediate as $\directed{\gamma}^{(1)}(G)$ is a restriction of $\directed{\gamma}(G)$, so we focus on proving the second inequality.

Let $f:V \to \R^n$ and $\directed{E}$ be a solution to $\directed{\gamma}(G)$ as stated in \autoref{def:compact}.
As in~\cite{LRV13}, we construct a $1$-dimensional solution $y \in \R^n$ to $\directed{\gamma}^{(1)}(G)$ by setting $y(v) = \inner{f(v)}{h}$, where $h \sim N(0,1)^n$ is a Gaussian random vector with independent entries.

First, consider the expected objective value of $y$ to $\directed{\gamma}^{(1)}(G)$.
For each max term in the summand, 
\[
\E\bigg[ \max_{u: u \to v} \big(y(u)-y(v)\big)^2 \bigg]
=\E\bigg[ \max_{u: u \to v} \biginner{f(u)-f(v)}{h}^2 \bigg]
\leq 2 \max_{u: u \to v} \norm{f(u)-f(v)}^2 \cdot \log d,
\]
where the last inequality is by applying \autoref{Gaussian-expected-max} on normal random variable $\inner{f(u)-f(v)}{h}$ with variance $\norm{f(u)-f(v)}^2$ for each of the at most $d$ terms.
By linearity of expectation, the expected objective value of $\directed{\gamma}^{(1)}(G)$ is
\[
\E\Bigg[ \sum_{v \in V} \pi(v) \max_{u: u \to v} \big(y(u)-y(v)\big)^2 \Bigg]
\leq 2 \log d \cdot \sum_{v \in V} \pi(v) \max_{u: u \to v} \norm{f(u)-f(v)}^2
= 2 \log d \cdot \directed{\gamma}(G).
\]
Therefore, by Markov's inequality, 
\[
\Pr\Bigg[ \sum_{v \in V} \pi(v) \max_{u: u \to v} \big(y(u)-y(v)\big)^2 \geq 48 \log d \cdot \directed{\gamma}(G)\Bigg] \leq \frac{1}{24}.
\]
Next, by applying \autoref{Gaussian-denominator} with $z_v = \sqrt{\pi(v)} \cdot y(v)$, it follows that
\[
\E\Bigg[ \sum_{v \in V} \pi(v) y(v)^2 \Bigg]
= \sum_{v \in V} \pi(v) \norm{f(v)}^2 = 1 
\quad \implies \quad
\Pr\Bigg[ \sum_{v \in V} \pi(v) y(v)^2 \geq \frac{1}{2} \Bigg] \geq \frac{1}{12}.
\]
Finally, since $\sum_{v \in V} \pi(v) f(v) = \vec{0}$, it holds that 
\[
\sum_{v \in V} \pi(v) y(v) = \sum_{v \in V} \pi(v) \inner{f(v)}{h} = \biginner{\sum_{v \in V}\pi(v) f(v)}{h} = 0.
\]
Therefore, with probability at least $\frac{1}{24}$, all of these events hold simultaneously. The second event
\[
 \sum_{v \in V} \pi(v) y(v)^2 \geq \frac{1}{2}
\]
means that we can rescale $y$ by a factor of at most $\sqrt{2}$, so that the constraint $\sum_{v \in V} \pi(v) y(v)^2 = 1$ is satisfied and the objective value is at most $96 \log d \cdot \directed{\gamma}(G)$.
Hence we conclude that $\directed{\gamma}^{(1)}(G) \lesssim \directed{\gamma}(G)\cdot \log d$.
\end{proof}

\subsubsection{Cheeger Rounding for Vertex Expansion} \label{sec:Cheeger-rounding}

We generalize \autoref{thm:OTZ-Cheeger} to weighted vertex expansion.
Our proof does not use the concept of matching conductance in~\cite{OTZ22},
rather it is based on a more traditional analysis as in~\cite{BHT00} using the directed program $\directed{\gamma}^{(1)}(G)$ in \autoref{def:compact}.

\begin{theorem}[Cheeger Inequality for Weighted Vertex Expansion] \label{thm:weighted-Cheeger}
For any undirected graph $G=(V,E)$ and any probability distribution $\pi$ on $V$,
\[
\psi(G)^2 \lesssim \gamma^{(1)}(G) \lesssim \psi(G).
\]
\end{theorem}

The organization is as follows.  
We will prove the easy direction in \autoref{lem:Cheeger-easy} in \autoref{app:Cheeger-vertex}.
For the hard direction, we will work on $\directed{\gamma}^{(1)}(G)$ instead.
First we do the standard preprocessing step to truncate the solution to have $\pi$-weight at most $1/2$.
Then the main step is to define a modified vertex boundary condition for directed graphs and use it for the analysis of the standard threshold rounding.
Finally we clean up the solution obtained from threshold rounding to find a set with small vertex expansion in the underlying undirected graph.

\begin{lemma}[Easy Direction] \label{lem:Cheeger-easy}
For any undirected graph $G=(V,E)$ and any probability distribution $\pi$ on $V$, 
\[\gamma^{(1)}(G) \leq 2\psi(G).\]
\end{lemma}

We now turn to proving the hard direction.
Given a solution $y:V\to \R$ to $\directed{\gamma}^{(1)}(G)$ in \autoref{def:compact} satisfying $y \perp \pi$,
we do the standard preprocessing step to truncate $y$ to obtain a non-negative solution $x$ with $\pi(\supp(x)) \leq 1/2$ and comparable objective value.
Note that we no longer require that $x \perp \pi$.
The proof of the following lemma is standard and is deferred to \autoref{app:Cheeger-vertex}.

\begin{lemma}[Truncation] \label{lem:truncation}
Let $G=(V,E)$ be an undirected graph and $\pi$ be a probability distribution on $V$.
Given a solution $y$ and $\directed{E}$ to $\directed{\gamma}^{(1)}(G)$ as stated in \autoref{def:compact},
there is a solution $x$ and $\directed{E}$ with $x \geq 0$ and $\pi(\supp(x)) \leq 1/2$ and 
\[
\frac{\sum_{v \in V} \pi(v) \max_{u:uv \in \directed{E}} (x(u)-x(v))^2}{\sum_{v \in V} \pi(v) x(v)^2} \leq 4 \directed{\gamma}^{(1)}(G).
\]
\end{lemma}

For the standard threshold rounding, we define the appropriate vertex boundary $\directed{\partial} S$ for the analysis of the directed program $\directed{\gamma}^{(1)}(G)$.
Note that, unlike $\partial S$, $\directed{\partial} S$ may contain vertices in $S$.
A good interpretation is to think of $\directed{\partial} S$ as a vertex cover of the edge boundary $\delta(S)$ in the undirected sense.

\begin{definition}[Directed Vertex Boundary and Expansion] \label{def:vertex-cover}
Let $G=(V,\directed{E})$ be a directed graph.
For $S \subseteq V$, define the directed vertex boundary and the directed vertex expansion as
\[
\directed{\partial} S := \big\{ v \in S \mid \exists u \notin S {\rm~with~} uv \in \directed{E}\big\} \cup \big\{ v \notin S \mid \exists u \in S {\rm~with~} uv \in \directed{E}\big\}
\quad {\rm and} \quad
\directed{\psi}(S) := \frac{\pi(\directed{\partial}S)}{\pi(S)}.
\]
\end{definition}

The main step is to prove that the standard threshold rounding will find a set $S$ with small directed vertex expansion $\directed{\psi}(S)$.

\begin{proposition}[Threshold Rounding for $\gamma(G)$] \label{prop:threshold-rounding}
Let $G=(V,E)$ be an undirected graph and $\pi$ be a probability distribution on $V$.
Given a solution $x$ and $\directed{E}$ with $x \geq 0$ and 
\[
\frac{\sum_{v \in V} \pi(v) \max_{u:u \to v} (x(u)-x(v))^2}{\sum_{v \in V} \pi(v) x(v)^2} \leq \gamma(G),
\]
there is a set $S \subseteq \supp(x)$ with $\directed{\psi}(S) \lesssim \sqrt{\gamma(G)}$.
\end{proposition}

\begin{proof}
For any $t \geq 0$, define $S_t := \{v \in V \mid x(v)^2 > t\}$.
By a standard averaging argument,
\[
\min_t \directed{\psi}(S_t)
\leq \frac{\int_{0}^{\infty} \pi\big(\directed{\partial} S_t\big) \, dt}{\int_{0}^{\infty} \pi(S_t) \, dt}.
\]
The denominator is
\begin{equation*}
\int_{0}^{\infty} \pi(S_t) \, dt
= \sum_{v \in V} \pi(v) \int_{0}^{\infty} \mathbbm{1}[v \in S_t] \, dt
= \sum_{v \in V} \pi(v) x(v)^2.
\end{equation*}
For the numerator, note that a vertex $v$ is in $\directed{\partial}S_t$ if and only if $\min\{x(u)^2 \mid uv \in \directed{E} \} \leq t \leq \max\{x(u)^2 \mid uv \in \directed{E}\}$, where we recall the assumption that every vertex has a self loop, and so $vv \in \directed{E}$ and thus $\min\{x(u)^2 \mid uv \in \directed{E} \} \leq x(v)^2 \leq \max\{x(u)^2 \mid uv \in \directed{E}\}$.
Hence the numerator is
\begin{eqnarray*}
& & \int_{0}^{\infty} \pi\Big(\directed{\partial} S_t\Big) \, dt
\\
&=&
\sum_{v \in V} \pi(v) \cdot \int_{0}^{\infty} \mathbbm{1}\Big[v \in \directed{\partial} S_t\Big] \, dt
  \\
&=&
\sum_{v \in V} \pi(v) \cdot \int_{0}^{\infty} 
\mathbbm{1}\Big[ \min\big\{x(u)^2 \mid uv \in \directed{E} \big\} \leq t 
\leq \max\big\{x(u)^2 \mid uv \in \directed{E}\big\}  \Big]\, dt
  \\
&=&
  \sum_{v \in V} \pi(v) \bigg[
    \max_{\substack{u: u \to v \\ x(u) > x(v)}} \big\{x(u)^2 - x(v)^2\big\} +
    \max_{\substack{u: u \to v \\ x(u) < x(v)}} \big\{x(v)^2 - x(u)^2\big\}
  \bigg]
  \\
&\leq&
2 \sum_{v \in V} \pi(v) \bigg[ \max_{u:u \to v} \big\{ |x(u)^2-x(v)^2| \big\}\bigg]
\\
  &\leq&
  2\sum_{v \in V} \pi(v) \biggl[
    \max_{u: u \to v}
    \Big\{ 
      (x(u) - x(v))^2 + 2 x(v) \cdot |x(u) - x(v)|
    \Big\}
  \biggr]
  \\
  &\leq&
  2\sum_{v \in V} \pi(v) \max_{u: u \to v} (x(u) - x(v))^2
  + 4\sqrt{
   \sum_{v \in V} \pi(v) x(v)^2 \cdot \sum_{v \in V} \pi(v) \max_{u: u \to v} (x(u) - x(v))^2,
  }
\end{eqnarray*}
where the second-last inequality is by $|x(u)^2 - x(v)^2| \leq |x(u) - x(v)| \cdot \big( |x(u)-x(v)| + 2|x(v)| \big)$,
and the last inequality is by the Cauchy-Schwarz inequality.

Combining the numerator and the denominator bounds,
\begin{eqnarray*}
  \frac
    {\int_{0}^{\infty} \pi\Big(\directed{\partial} S_t\Big) \, dt}
    {\int_{0}^{\infty} \pi(S_t) \, dt}
  &\leq&
  \frac{\sum_{v \in V} \pi(v) \max_{u: u \to v} (x(u) - x(v))^2}{\sum_{v \in V} \pi(v) x(v)^2}
  +
  2\sqrt{
    \frac{\sum_{v \in V} \pi(v) \max_{u: u \to v} (x(u) - x(v))^2}{\sum_{v \in V} \pi(v) x(v)^2}
  }
  \\
  &=& \gamma + 2\sqrt{\gamma} \lesssim \sqrt{\gamma},
\end{eqnarray*}
where the last inequality is by $\gamma \leq 2$ as was shown in the proof of the easy direction in \autoref{lem:Cheeger-easy}.
Therefore, $\min_t \directed{\psi}(S_t) \lesssim \sqrt{\gamma}$ and $S_t \subseteq \supp(x)$ by construction.
\end{proof}

Finally, given a set $S$ with small directed vertex expansion $\directed{\psi}(S)$, we show how to find a set $S' \subseteq S$ with small vertex expansion $\psi(S')$.
This step is similar to the step in~\cite[Proposition 2.2]{OTZ22} from matching conductance to vertex expansion.

\begin{lemma}[Postprocessing for Vertex Expansion] \label{lem:cleanup}
Let $G=(V,\directed{E})$ be a directed graph.
Given a set $S$ with $\directed{\psi}(S) < 1/2$,
there is a set $S' \subseteq S$ with $\psi(S') \leq 2\directed{\psi}(S)$ in the underlying undirected graph of $G$. 
\end{lemma}
\begin{proof}
From \autoref{def:vertex-cover}, the observation is that all undirected edges in $\delta(S)$ are incident to at least one vertex in $\directed{\partial}S$.
Define $S' := S - \directed{\partial}S$.
Then observe that $\partial S' \subseteq \directed{\partial}S$, as there are no incoming edges to $S'$ from $V-(S' \cup \directed{\partial}S)$ and all outgoing edges from $S'$ go to $\directed{\partial}(S)$.
This implies that
\[
\pi(\partial S') 
\leq \pi\Big(\directed{\partial}S\Big) 
= \directed{\psi}(S) \cdot \pi(S)
\leq 2\directed{\psi}(S) \cdot \pi(S'),
\]
where the last inequality uses the assumption that 
$\directed{\psi}(S) = \pi\big(\directed{\partial} S\big) / \pi(S) < 1/2$ 
and so $\pi(S') \geq \pi(S) - \pi\big(\directed{\partial} S\big) \geq \pi(S)/2$.
We conclude that $\psi(S') \leq 2\directed{\psi}(S)$.
\end{proof}

We put together the steps and complete the proof of \autoref{thm:weighted-Cheeger} in \autoref{app:Cheeger-vertex}.
\section{Cheeger Inequality for Bipartite Vertex Expansion} \label{sec:Cheeger-bipartite}

The goal of this section is to prove \autoref{thm:Cheeger-bipartite}, which relates the maximum reweighted lower spectral gap $\zeta^*(G)$ in \autoref{def:lambda-max-primal} and the bipartite vertex expansion $\psi_B(G)$ in \autoref{def:bipartite-vertex-expansion}.
The proof follows closely the proof of \autoref{thm:Cheeger-vertex} in the previous section, so some steps will be stated without proofs, and the focus will be on the threshold rounding step.

\subsection{Primal and Dual Programs}
The primal program  $\zeta^*(G)$ in \autoref{def:lambda-max-primal} has the following dual which is similar to $\gamma(G)$ in \autoref{prop:Roch-dual}.

\begin{proposition}[Dual Program for Lower Spectral Gap~\cite{Roc05}] \label{prop:Roch-dual-bipartite}
Given an undirected graph $G=(V,E)$ and a probability distribution $\pi$ on $V$,
the following semidefinite program is dual to the primal program in \autoref{def:lambda-max-primal} with strong duality $\zeta^*(G) = \nu(G)$ where
\begin{eqnarray*}
\nu(G) := \min_{f: V \to \R^n,~g: V \to \R_{\geq 0}} & & \sum_{v \in V} \pi(v) g(v)
\\
\st & & {\sum_{v \in V} \pi(v) \norm{f(v)}^2} = 1
\\
& & g(u) + g(v) \geq \norm{f(u) + f(v)}^2 \quad \quad \forall uv \in E.
\end{eqnarray*}
\end{proposition}

There are two differences between $\gamma(G)$ in \autoref{prop:Roch-dual} and $\nu(G)$ in \autoref{prop:Roch-dual-bipartite}.
One is that the constraint $g(u)+g(v) \geq \norm{f(u)-f(v)}^2$ in $\gamma(G)$ is replaced by the constraint $g(u)+g(v) \geq \norm{f(u)+f(v)}^2$ in $\nu(G)$, which are handled in a similar way.
The other is that the constraint of $\sum_{v \in V} \pi(v) f(v) = \vec{0}$ in $\gamma(G)$ is not present in $\nu(G)$, and so it is slightly easier to work with $\nu(G)$, e.g. no truncation step needed.

The nice form of the dual program $\nu(G)$ is the main reason behind the definition of the primal program $\zeta^*(G)$. By the variational characterization of eigenvalues, the quadratic form of $D_P + P$, and the $\pi$-reversibility of $P$,
\[
  \lambda_{\min}(D_P + P) = \min_{x: V \rightarrow \R} \frac{x^T(D_P + P)x}{x^T x}
  = \min_{x: V \rightarrow \R} \frac{\sum_{uv \in E} \pi(u) P(u, v) (x(u) + x(v))^2}{\sum_{u \in V} \pi(u) x(u)^2},
\]

and this is the intermediate form we need to derive the dual, see~\cite{Roc05}.

Later, as in \autoref{sec:dual-orientation}, we will define a directed dual program $\directed{\nu}(G)$, and the dual constraint $g \ge 0$ crucially enables us to relate the two programs. The dual constraint $g \ge 0$ comes from the primal constraint $\sum_{v \in V} P(u, v) \le 1$, whereas if we use $\sum_{v \in V} P(u, v) = 1$ then $g$ will be unconstrained.

For $\lambda_2^*(G)$, we sidestep the issue by adding self loops to each vertex of $G$. The non-negativity of $g$ in the dual program $\gamma(G)$ follows indirectly from $g(u) + g(u) \ge \norm{f(u) - f(u)}^2$, as now $(u, u) \in E$.
Moreover, adding self loops does not change the vertex expansion of $G$. 
Therefore, $\lambda_2^*(G)$ can take the more natural form where $P$ does correspond to a transition matrix.
However, we cannot do the same for $\zeta^*(G)$, because the additional constraint on $g$ becomes $g(u) + g(u) \ge \norm{f(u) + f(u)}^2$ which changes the objective value, and also that adding self loops takes away the bipartiteness of subgraphs.

\subsection{Proof of \autoref{thm:Cheeger-bipartite}}

We use the same two-step plan as in \autoref{sec:proof-Cheeger}.
In the first step, we project the solution to the dual program in \autoref{prop:Roch-dual-bipartite} into a $1$-dimensional solution to the following program.

\begin{definition}[One-Dimensional Dual Program for Lower Spectral Gap] \label{def:Roch-dual-bipartite-1D}
Given an undirected graph $G=(V,E)$ and a probability distribution $\pi$ on $V$,
$\nu^{(1)}(G)$ is defined as the following program:
\begin{eqnarray*}
\nu^{(1)}(G) := \min_{f: V \to \R,~g: V \to \R_{\geq 0}} & & \sum_{v \in V} \pi(v) g(v)
\\
\st & & {\sum_{v \in V} \pi(v) \norm{f(v)}^2} = 1
\\
& & g(u) + g(v) \geq \norm{f(u) + f(v)}^2 \quad \quad \forall uv \in E.
\end{eqnarray*}
\end{definition}

As in \autoref{prop:improved-projection}, we use the Gaussian projection method in~\cite{LRV13} to prove the following guarantee.

\begin{proposition}[Gaussian Projection for $\nu(G)$] \label{prop:projection-nu}
For any undirected graph $G=(V,E)$ with maximum degree $d$ and any probability distribution $\pi$ on $V$,
\[
\nu(G) \leq \nu^{(1)}(G) \lesssim \nu(G) \cdot \log d. 
\]
\end{proposition}

In the second step, we prove a Cheeger-type inequality relating $\psi_B(G)$ and $\nu(G)$.

\begin{theorem} \label{thm:Cheeger-nu}
For any undirected graph $G=(V,E)$ and any probability distribution $\pi$ on $V$,
\[
\psi_B(G)^2 \lesssim \nu^{(1)}(G) \lesssim \psi_B(G).
\]
\end{theorem}

Combining \autoref{prop:Roch-dual-bipartite} and \autoref{prop:projection-nu} and \autoref{thm:Cheeger-nu} gives 
\[
\psi_B(G)^2 \lesssim \nu^{(1)}(G) \lesssim \nu(G) \cdot \log d = \zeta^*(G) \log d
\quad {\rm and} \quad
\zeta^*(G) = \nu(G) \leq \nu^{(1)}(G) \lesssim \psi_B(G),
\]
proving \autoref{thm:Cheeger-bipartite}.
We will prove \autoref{prop:projection-nu} and \autoref{thm:Cheeger-nu} in the following subsections.

\subsection{Dual Program on Graph Orientation} \label{sec:dual-orientation-bipartite}

As in \autoref{sec:dual-orientation}, we introduce a directed program for the analysis of both steps.

\begin{definition}[Directed Dual Programs for $\nu(G)$] 
\label{def:directed-nu}
Given an undirected graph $G=(V,E)$ and a probability distribution $\pi$ on $V$,
\begin{eqnarray*}
\directed{\nu}(G) := \min_{f: V \to \R^n,~g: V \to \R_{\geq 0}} & & \sum_{v \in V} \pi(v) g(v)
\\
\st & & {\sum_{v \in V} \pi(v) \norm{f(v)}^2} = 1
\\
& & \max\{g(u),g(v)\} \geq \norm{f(u) + f(v)}^2 \quad \quad \forall uv \in E.
\end{eqnarray*}
$\directed{\nu}^{(1)}(G)$ is defined as the $1$-dimensional program of $\directed{\nu}(G)$ where $f:V \to \R$ instead of $f: V \to \R^n$.
\end{definition}

As in \autoref{lem:factor2}, we show that $\nu(G)$ and $\directed{\nu}(G)$ are closely related.  
The proof is the same as in \autoref{lem:factor2} and is omitted,
but note that $g \geq 0$ is needed.

\begin{lemma} \label{lem:factor2-nu}
For any undirected graph $G=(V,E)$ and any probability distribution $\pi$ on $V$,
\[
\nu(G) \leq \directed{\nu}(G) \leq 2\nu(G)
\quad {\rm and} \quad
\nu^{(1)}(G) \leq \directed{\nu}^{(1)}(G) \leq 2\nu^{(1)}(G).
\]
\end{lemma}

As in \autoref{def:compact}, we use an orientation of the edges to eliminate the variables $g(v)$ for $v \in V$ in $\directed{\nu}(G)$.
The proof is the same as in \autoref{def:compact} and is omitted,
but note that $g \geq 0$ is needed.

\begin{lemma}[Directed Dual Programs Using Orientation for $\nu(G)$] \label{def:orientation-nu}
Let $G=(V,E)$ be an undirected graph and $\pi$ be a probability distribution on $V$.
Let $\directed{E}$ be an orientation of the undirected edges in $E$.
Then 
\begin{eqnarray*}
\directed{\nu}(G) = \min_{f: V \to \R^n} \min_{\directed{E}} & & \sum_{v \in V} \pi(v) \max_{u: uv \in \directed{E}} \norm{f(u)+f(v)}^2
\\
\st & & {\sum_{v \in V} \pi(v) \norm{f(v)}^2} = 1.
\end{eqnarray*}
Similarly, $\directed{\nu}^{(1)}(G)$ can be written in the same form with $f:V \to \R$ instead of $f: V \to \R^n$.
\end{lemma}

Once we have this formulation using orientation, we can use the same proof as in \autoref{prop:improved-projection} to show that $\directed{\nu}(G) \leq \directed{\nu}^{(1)}(G) \lesssim \log d \cdot \directed{\nu}(G)$, and thus \autoref{prop:projection-nu} follows from \autoref{lem:factor2-nu} and we omit the proof.
It remains to prove \autoref{thm:Cheeger-nu}, which will be done in the next subsection.

\subsection{Cheeger Rounding for Bipartite Vertex Expansion}

The goal of this subsection is to prove \autoref{thm:Cheeger-nu}.
We will prove the following easy direction in \autoref{app:Cheeger-bipartite}.

\begin{lemma}[Easy Direction] \label{lem:Cheeger-easy-bipartite}
For any undirected graph $G=(V,E)$ and any probability distribution $\pi$ on $V$, 
\[\nu^{(1)}(G) \leq 2\psi_B(G).\]
\end{lemma}

For the hard direction, we will work with $\directed{\nu}^{(1)}(G)$ instead.
There is no need to do the truncation step as in \autoref{lem:truncation}, as there is no constraint about $\pi(S)$ of the output set $S$.
The main step is to define a modified bipartite vertex expansion condition for directed graphs and use it for the analysis of the threshold rounding.

Let $S_1, S_2$ be two disjoint subsets of $V$.
In the edge conductance setting, rephrasing using our terminology,
Trevisan~\cite{Tre09} defined the ``bipartite edge boundary'' $\delta(S_1,S_2)$ as $E(S_1) \cup E(S_2) \cup \delta(S_1 \cup S_2)$ where $E(S_i)$ is the set of induced edges in $S_i$ for $i \in \{1,2\}$, and the ``bipartite edge conductance'' $\phi(S_1,S_2)$ as $|\delta(S_1,S_2)|/\vol(S_1 \cup S_2)$.
We define the appropriate bipartite vertex boundary $\directed{\partial}(S_1,S_2)$ for vertex expansion and for directed graphs in the following definition.
As in \autoref{def:vertex-cover}, note that $\directed{\partial}(S_1,S_2)$ could contain vertices in $V-(S_1 \cup S_2)$.
Again, a good intuition is to think of $\directed{\partial}(S_1,S_2)$ as a vertex cover of the edges in the bipartite edge boundary $\delta(S_1,S_2)$ in the undirected sense.

\begin{definition}[Directed Bipartite Vertex Boundary and Expansion] \label{def:vertex-cover-bipartite}
Let $G=(V,\directed{E})$ be a directed graph.
Let $S_1, S_2$ be two disjoint subsets of $V$.
The directed bipartite vertex boundary of $(S_1,S_2)$ is defined as
\begin{eqnarray*}
\directed{\partial}(S_1,S_2) & := &
\big\{ v \in S_1 \mid \exists u \in S_1 {\rm~with~} uv \in \directed{E}, {\rm~or~} \exists u \notin S_1 \cup S_2 {\rm~with~} uv \in \directed{E}\big\} \cup
\\ 
& & \big\{ v \in S_2 \mid \exists u \in S_2 {\rm~with~} uv \in \directed{E}, {\rm~or~} \exists u \notin S_1 \cup S_2 {\rm~with~} uv \in \directed{E}\big\} \cup
\\
& & \big\{ v \notin S_1 \cup S_2 \mid \exists u \in S_1 \cup S_2 {\rm~with~} uv \in \directed{E}\big\},
\end{eqnarray*}
and the directed bipartite vertex expansion as
\[
\directed{\psi}(S_1,S_2) := \frac{\pi\big(\directed{\partial}(S_1,S_2)\big)}{\pi(S_1 \cup S_2)}.
\]
\end{definition}

An example of directed bipartite vertex expansion is provided in Figure \ref{fig:bipartite-vb-example}. 

\begin{figure}[h]
  \centering
  \includegraphics[scale=0.55]{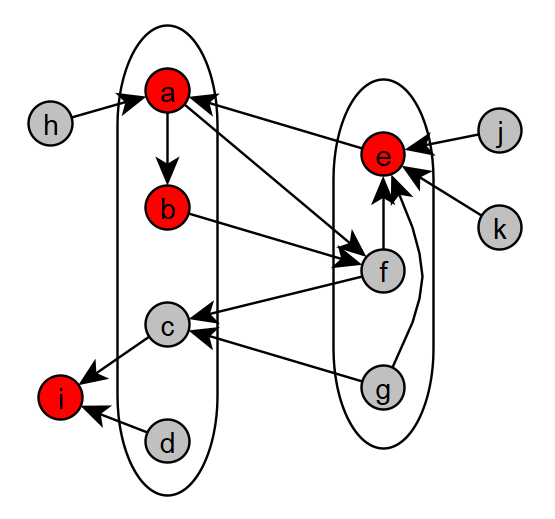}
  \caption{In the graph shown, the bipartition is $S_1 = \{a, b, c, d\}$ and $S_2 = \{e, f, g\}$. The vertices $a, b, e, i$, colored in red, are the vertices in $\protect \directed{\partial}(S_1, S_2)$.}
  \label{fig:bipartite-vb-example}
\end{figure}

We prove that the threshold rounding defined in~\cite{Tre09}, when applied on $\directed{\nu}^{(1)}(G)$, will give a set with small directed bipartite vertex expansion.

\begin{proposition}[Threshold Rounding for $\nu(G)$] \label{prop:threshold-rounding-nu}
Let $G=(V,E)$ be an undirected graph and $\pi$ be a probability distribution on $V$.
Given a solution $x$ and $\directed{E}$ to $\directed{\nu}^{(1)}(G)$, 
there is a polynomial time algorithm to find two disjoint subsets $S_1,S_2 \subseteq V$ with $\directed{\psi}(S_1,S_2) \lesssim \sqrt{\directed{\nu}^{(1)}(G)}$.
\end{proposition}

\begin{proof}
For any $t > 0$, define $S_t := \{v \in V \mid x(v) > \sqrt{t}\}$ and $S_{-t} := \{v \in V \mid x(v) < -\sqrt{t}\}$ as in~\cite{Tre09}.
By a standard averaging argument,
\[
\min_t \directed{\psi}(S_t,S_{-t})
\leq \frac{\int_{0}^{\infty} \pi\big(\directed{\partial}(S_t,S_{-t})\big) \, dt}{\int_{0}^{\infty} \pi(S_t \cup S_{-t}) \, dt}.
\]
The denominator is
\begin{equation*}
\int_{0}^{\infty} \pi(S_t \cup S_{-t}) \, dt
= \sum_{v \in V} \pi(v) \int_{0}^{\infty} \mathbbm{1}[v \in S_t \cup S_{-t}] \, dt
= \sum_{v \in V} \pi(v) x(v)^2 
= 1.
\end{equation*}
For the numerator, we consider when a vertex $v$ is in $\directed{\partial}(S_t,S_{-t})$.
Assume without loss that $x(v) \geq 0$; the other case is symmetric.
There are two scenarios where $v \in \directed{\partial}(S_t,S_{-t})$:
\begin{enumerate}
\item The first scenario is when $uv \in E[S_t]$ in the undirected sense for some directed edge $u \to v$.
For a fixed edge $u \to v$, 
this happens when $t < \max\{0, \min\{x(u),x(v)\}\}^2 \leq ((x(u)+x(v))/2)^2$,
where the last inequality can be verified by considering the cases $x(v) \geq 0$ and $x(v) \leq 0$ separately.
Therefore, the first scenario happens when
\[
t < \max_{u:u \to v} \max\{0, \min\{x(u),x(v)\}\}^2 
\leq \max_{u:u \to v} \Big(\frac{x(u)+x(v)}{2}\Big)^2.
\]
\item 
The second scenario is when $uv \in \delta(S_t \cup S_{-t})$ in the undirected sense for some directed edge $u \to v$.
For a fixed edge $u \to v$,
this happens when $x(u)^2 \leq t < x(v)^2$ (so $u \notin S_t \cup S_{-t}$ and $v \in S_t$), or when $x(v)^2 \leq t < x(u)^2$ (so $u \in S_t$ and $v \notin S_t \cup S_{-t}$).
Therefore, the second scenario happens when
\[
\min_{u:u \to v} x(u)^2 \leq t < x(v)^2
\quad \quad {\rm or} \quad \quad
x(v)^2 \leq t < \max_{u:u \to v} x(u)^2.
\]
\end{enumerate}
Hence the numerator is
\begin{eqnarray*}
  &&
  \int_0^{\infty} \pi\big(\directed{\partial}(S_t,S_{-t})\big)   \,dt
  \\
  &=&
  \sum_{v \in V} \pi(v) \cdot \int_{0}^{\infty} \mathbbm{1}\Big[v \in \directed{\partial} (S_t,S_{-t})\Big] \, dt
  \\
  &\leq&
  \sum_{v \in V} \pi(v) \Bigg[
    \max_{u:u \to v} \Big( \frac{x(u) + x(v)}{2} \Big)^2 +
    \bigg(  x(v)^2 - \min_{\substack{u:u \to v}} x(u)^2\bigg) + 
    \bigg( \max_{\substack{u:u \to v}} x(u)^2 - x(v)^2\bigg)
  \Bigg]
  \\
  &=&
  \frac{1}{4} \directed{\nu}^{(1)}(G) +
  \sum_{v \in V} \pi(v) \Bigg[
    \max_{\substack{u:u \to v \\ |x(u)| < |x(v)|}} \big\{x(v)^2 - x(u)^2\big\} + 
    \max_{\substack{u:u \to v \\ |x(u)| > |x(v)|}} \big\{x(u)^2 - x(v)^2\big\}
  \Bigg]
 \\
  &\leq&
  \frac{9}{4} \directed{\nu}^{(1)}(G) +
  2\sum_{v \in V} \pi(v) \Bigg[
    \max_{\substack{u:u \to v \\ |x(u)| < |x(v)|}} |x(u)| \cdot \big|x(u)+x(v)\big| + 
    \max_{\substack{u:u \to v \\ |x(u)| > |x(v)|}} |x(v)| \cdot \big|x(u)+x(v)\big|
  \Bigg]
 \\
  &\leq&
  \frac{9}{4} \directed{\nu}^{(1)}(G) +
  2\sum_{v \in V} \pi(v) |x(v)| 
    \max_{\substack{u:u \to v}} \big|x(u)+x(v)\big|
 \\
  &\leq&
  \frac{9}{4} \directed{\nu}^{(1)}(G) + 2\sqrt{
   \sum_{v \in V} \pi(v) x(v)^2 \cdot \sum_{v \in V} \pi(v) \max_{u: u \to v} \big(x(u) + x(v)\big)^2
  }
  \\
  &\lesssim& \directed{\nu}^{(1)}(G) + \sqrt{\directed{\nu}^{(1)}(G)}.
\end{eqnarray*}
We explain these steps one by one.
The first inequality is by the two scenarios explained in detail above.
The second equality uses the definition that $\directed{\nu}^{(1)}(G) = \sum_{v \in V} \pi(v) \max_{u:u \to v} (x(u)+x(v))^2$.
For the second inequality, in the first max we write $x(v)^2 - x(u)^2 = (x(u)+x(v))^2 - 2x(u)(x(v)+x(u)) \leq (x(u)+x(v))^2 + 2|x(u)||x(v)+x(u)|$, and then take out $(x(u)+x(v))^2$ from the summation by using again the definition that $\directed{\nu}^{(1)}(G) = \sum_{v \in V} \pi(v) \max_{u:u \to v} (x(u)+x(v))^2$, while the second max is handled similarly.
In the third inequality we replace $|x(u)|$ in the first max term by $|x(v)|$.
The fourth inequality is by an application of the Cauchy-Schwarz inequality.
The final inequality is by $\sum_{v \in V} \pi(v) x(v)^2=1$ in the constraint of $\directed{\nu}^{(1)}(G)$.

We showed in the easy direction in \autoref{lem:Cheeger-easy-bipartite} that $\directed{\nu}^{(1)}(G) \leq 2$ and thus $\directed{\nu}^{(1)}(G) \lesssim \sqrt{\directed{\nu}^{(1)}(G)}$.
We conclude that there exists $(S_t,S_{-t})$ with $\directed{\psi}(S_t,S_{-t}) \lesssim \sqrt{\directed{\nu}^{(1)}(G)}$.  
\end{proof}

Finally, as in \autoref{lem:cleanup}, given $(S_1,S_2)$ with small directed bipartite vertex expansion, we show how to extract an induced bipartite graph with small vertex expansion.

\begin{lemma}[Postprocessing for Bipartite Vertex Expansion] \label{lem:cleanup-bipartite}
Let $G=(V,\directed{E})$ be a directed graph.
Given two disjoint subsets $S_1,S_2 \subseteq V$ with $\directed{\psi}(S_1,S_2) < 1/2$,
there are $S_1' \subseteq S_1$ and $S_2' \subseteq S_2$ with $\psi(S_1',S_2') \leq 2\directed{\psi}(S_1,S_2)$ and $S_1' \cup S_2'$ is an induced bipartite graph in the underlying undirected graph of $G$. 
\end{lemma}

\begin{proof}
From \autoref{def:vertex-cover-bipartite}, the observation is that $\directed{\partial}(S_1,S_2)$ is a vertex cover of $E(S_1) \cup E(S_2) \cup \delta(S_1 \cup S_2)$. 
So, by setting $S_1' := S_1 - \directed{\partial}(S_1,S_2)$ and $S_2' := S_2 - \directed{\partial}(S_1,S_2)$, then $(S_1',S_2')$ is an induced bipartite graph as there could be no edges induced in $S_1'$ and no edges induced in $S_2'$.
Also, $\partial(S_1' \cup S_2') \subseteq \directed{\partial}(S_1,S_2)$, as there could be no edges between $S_1' \cup S_2'$ and $V - (S_1 \cup S_2 \cup \directed{\partial}(S_1,S_2))$. 
Therefore, 
\[
\pi(\partial(S_1' \cup S_2')) 
\leq \pi\big(\directed{\partial}(S_1,S_2)\big)
= \directed{\psi}(S_1,S_2) \cdot \pi(S_1 \cup S_2)
\leq 2 \directed{\psi}(S_1,S_2) \cdot \pi(S_1' \cup S_2'),
\]
where the last inequality uses the assumption that $\directed{\psi}(S_1,S_2) = \pi\big(\directed{\partial}(S_1,S_2)\big) / \pi(S_1 \cup S_2) < 1/2$ and so $\pi(S_1' \cup S_2') \geq \pi(S_1 \cup S_2) - \pi\big(\directed{\partial}(S_1,S_2)\big) \geq \pi(S_1 \cup S_2)/2$.
We thus conclude that $\psi(S_1',S_2') \leq 2\directed{\psi}(S_1,S_2)$. 
\end{proof}

We complete the proof of \autoref{thm:Cheeger-nu} in \autoref{app:Cheeger-bipartite}.
\section{Higher-Order Cheeger Inequality for Vertex Expansion} \label{sec:higher-order-vertex}

The goal of this section is to prove \autoref{thm:higher-order-vertex}.
There are four main steps in the proof.

The first step is to reformulate the problem as a semidefinite program using the maximum reweighted sum of the $k$ smallest eigenvalues $\sigma_k^*(G)$.
Using von Neumann's minimax theorem, we construct the dual program of $\sigma_k^*(G)$ and see that it satisfies the so-called sub-isotropy condition.
The main focus in this section will then be to relate $\sigma_k^*(G)$ and $\psi_k(G)$, rather than to relate $\lambda_k^*(G)$ and $\psi_k(G)$ directly.

The second step is to project the dual solution to $\sigma_k^*(G)$ into a low-dimensional solution.
In this step, we use a similar apporach as in \autoref{sec:proof-Cheeger}, by introducing an intermediate directed dual program and then using the Gaussian projection method.
Also, we use a theorem in~\cite{LOT12} that proves that Gaussian projection approximately preserves the sub-isotropy condition.

The third step is to partition the low-dimensional solution into $k$ disjointly supported functions each with small objective value.
In this step, we closely follow the techniques in~\cite{LOT12}, such as radial projection distance, smooth localization and random parititoning.
We will review the techniques in~\cite{LOT12} in \autoref{sec:embedding} before presenting our proofs.

The last step is to apply the Cheeger inequality for vertex expansion in \autoref{sec:Cheeger-vertex} on these $k$ functions to find disjoint sets with small vertex expansion.
We will prove the easy direction and put together the steps to prove \autoref{thm:higher-order-vertex} in \autoref{sec:wrapup-higher}.

\subsection{Primal and Dual Programs}

As mentioned in \autoref{sec:techniques}, the maximum reweighted $k$-th smallest eigenvalue $\lambda_k^*(G)$ as formulated in \autoref{def:reweighted-lambda-k} is not a convex program.
Instead, we will study the following related quantity.

\begin{definition}[Maximum Reweighted Sum of $k$ Smallest Eigenvalues] 
\label{def:reweighted-sum}
Given an undirected graph $G=(V,E)$ and a probability distribution $\pi$ on $V$,
the maximum reweighted sum of $k$ smallest eigenvalues of the normalized Laplacian matrix of $G$ is defined as 
$\sigma_{k}^*(G) := \max_{P \geq 0} \sum_{i=1}^k \lambda_{k}(I-P)$, 
where $P$ is subject to the same constraints stated in \autoref{def:BDX-primal}.
Note that 
\[
\lambda_k^*(G) \leq \sigma_k^*(G) \leq k \cdot \lambda_k^*(G).
\]
\end{definition}

We reformulate the primal program in \autoref{def:reweighted-sum} as the semidefinite program in \autoref{def:sigma-k-dual}.
The proof of \autoref{def:sigma-k-dual} has a few small steps.
First we rewrite the sum $\sum_{i=1}^k \lambda_k(I-P)$ as $\sum_{i=1}^k \lambda_k(I - {\cal Q})$ for a symmetric matrix ${\cal Q}$.
Then we use \autoref{prop:sum-of-lambda-k} to write $\sum_{i=1}^k \lambda_k(I - {\cal Q})$ as a minimization problem using semidefinite programming.
Next we apply von Neumann's minimax theorem to change the order of max-min to min-max.
Then we do a change of variable and rewrite the program into a vector program form.
Finally, we use linear programming duality to rewrite the inner maximization problem as a minimization problem as was done in~\cite{Roc05}.

\begin{proposition}[Dual Program for $\sigma_k^*(G)$] \label{def:sigma-k-dual}
For any undirected graph $G=(V,E)$ with a self loop at each vertex and any probability distribution $\pi$ on $V$,
the following semidefinite program is dual to the primal program in \autoref{def:reweighted-sum} with strong duality $\sigma_k^*(G) = \kappa(G)$ where
\begin{eqnarray*}
\kappa(G) :=  \min_{f: V \to \R^n,~g: V \to \R_{\geq 0}} & & \sum_{v \in V} \pi(v) g(v)
  \\
  \st
  & & g(u) + g(v) \geq \norm{f(u) - f(v)}^2 \quad \forall uv \in E
  \\
  & & \sum_{v \in V} \pi(v) f(v) f(v)^T \preceq I_n
  \\
  & & \sum_{v \in V} \pi(v) \norm{f(v)}^2 = k.
\end{eqnarray*}
\end{proposition}

\begin{proof}
By \autoref{def:reweighted-sum}, $\sigma_k^*(G) = \max_{P \geq 0} \sum_{i=1}^k \lambda_k(I-P)$, where the maximum is over all $P$ satisfying the constraints in \autoref{def:BDX-primal}.
Consider the sum of eigenvalues for a fixed $P$ that satisfies the constraints.
The time reversible constraint $\pi(u) P(u,v) = \pi(v) P(v,u)$ for all $uv \in E$ is equivalent to the matrix $Q := \Pi P$ being symmetric where $\Pi := \diag(\pi)$.
Let ${\cal Q} := \Pi^{-1/2} Q \Pi^{-1/2}$ be the normalized adjacency matrix of $Q$.
Note that $P$ and ${\cal Q}$ have the same spectrum,
as ${\cal Q} = \Pi^{-1/2} Q \Pi^{-1/2} = \Pi^{1/2} P \Pi^{-1/2}$.
Therefore $\sum_{i=1}^k \lambda_i(I - P) = \sum_{i=1}^k \lambda_i(I - \cal{Q})$
where $I - {\cal Q}$ is a symmetric matrix.

By \autoref{prop:sum-of-lambda-k}, the sum of the $k$ smallest eigenvalue of the symmetric matrix $I - {\cal Q}$ can be written as the following semidefinite program:
\begin{eqnarray*}
  \sum_{i=1}^k \lambda_i(I-{\cal Q}) = \min_{Y \in \R^{n \times n}} && \tr\big(Y \cdot (I - {\cal Q})\big)
  \\
  \st && 0 \preceq Y \preceq I
  \\
  && \tr(Y) = k.
\end{eqnarray*}

Note that $I-{\cal Q}$ is the normalized Laplacian matrix of $Q$.
We consider the change of variable $Y = \Pi^{1/2} Z \Pi^{1/2}$, so as to rewrite the objective function in terms of $\Pi - Q$ which is the Laplacian matirx of $Q$:
\begin{eqnarray*}
  \sum_{i=1}^k \lambda_i(I-{\cal Q}) = \min_{Z \in \R^{n \times n}} && 
  \tr\big(Z \cdot (\Pi - Q)\big)
  \\
  \st && 0 \preceq \Pi^{\frac12} Z \Pi^{\frac12} \preceq I
  \\
  && \tr(\Pi^{\frac12} Z \Pi^{\frac12}) = k.
\end{eqnarray*}

Therefore, the primal program for $\sigma_k^*(G)$ can be rewritten in terms of $Q$ as follows:
\begin{align*}
\sigma_k^*(G) ~=~ \max_{Q \geq 0} \min_{Z \in \R^{n \times n}} &~~~ \tr\big( Z \cdot (\Pi - Q) \big) & 
\\
\st &~~~ Q(u,v) = 0 & & \forall uv \notin E
\\
&~~~ \sum_{v \in V} Q(u,v) = \pi(u) & & \forall u \in V
\\
&~~~ Q(u,v) = Q(v,u) & & \forall uv \in E
\\
&~~~ 0 \preceq \Pi^{\frac12} Z \Pi^{\frac12} \preceq I 
\\
&~~~ \tr\big( \Pi^{\frac12} Z \Pi^{\frac12} \big) = k. 
\end{align*}

Now we write the dual program by using von Neumann's minimax theorem in \autoref{thm:von-Neumann} to switch the max-min to min-max in the objective function.
Note that von Neumann's theorem can be applied because the objective function is multiliner in $Z$ and $Q$ (hence concave in $Q$ and convex in $Z$), the feasible region of $Q$ is compact and convex as it is bounded and defined by linear constraints, and the feasible region of $Z$ is compact and convex as it is bounded and defined by PSD and trace constraints.
Hence, we can switch the order of $\max_Q \min_Z$ to obtain the dual program by rewriting the objective function as
\[
\min_{Z \in \R^{n \times n}} \max_{Q \geq 0}~\tr\big( Z \cdot (\Pi - Q) \big).
\]

Next we rewrite this dual program into a vector program form.
As $Z \succeq 0$, we can write $Z = FF^T$ where $F$ is an $n \times n$ matrix.
We denote the $i$-th column of $F$ by $f_i \in \R^n$ for $1 \leq i \leq n$ and think of it as an eigenvector, denote the $v$-th row of $F$ by $f(v) \in \R^n$ and think of it as the spectral embedding of a vertex $v$, and denote the $(v,i)$-th entry of $F$ by $f_i(v)$ for $1 \leq i \leq n$ and $v \in V$.
As $\Pi-Q$ is the Laplacian matrix of $Q$, the quadratic form for a vector $x \in \R^n$ is $x^T (\Pi - Q) x = \sum_{uv \in E} (x(u)-x(v))^2 \cdot Q(u,v)$, and thus the objective function can be rewritten as
\[
\tr(F^T (\Pi - Q) F) 
= \sum_{i=1}^n f_i^T (\Pi - Q) f_i
= \sum_{i=1}^n \sum_{uv \in E} (f_i(u) - f_i(v))^2 \cdot Q(u,v)
= \sum_{uv \in E} \norm{f(u)-f(v)}^2 \cdot Q(u,v).
\]
Note that $\Pi^{1/2} Z \Pi^{1/2} = (\Pi^{1/2} F) (F^T \Pi^{1/2})$ and $F^T \Pi F = (F^T \Pi^{1/2}) (\Pi^{1/2} F)$ have the same spectrum by \autoref{fact:AB=BA}.
So the first constraint can be rewritten as
\[
0 \preceq F^T \Pi F = \sum_{v \in V} \pi(v) f(v) f(v)^T \preceq I,
\]
and the second constraint can be rewritten as
\[
\tr(F^T \Pi F) 
= \tr\Big( \sum_{v \in V} \pi(v) f(v) f(v)^T  \Big)
= \sum_{v \in V} \pi(v) \norm{f(v)}^2 = k.
\]

Therefore, the dual program for $\sigma_k^*(G)$ can be rewritten as follows:
\begin{align*}
\kappa(G) ~:=~ \min_{f:V \to \R^n} \max_{Q \geq 0}&~~~ \sum_{uv \in E} \norm{f(u)-f(v)}^2 \cdot Q(u,v) & 
\\
\st &~~~ Q(u,v) = 0 & & \forall uv \notin E
\\
&~~~ \sum_{v \in V} Q(u,v) = \pi(u) & & \forall u \in V
\\
&~~~ Q(u,v) = Q(v,u) & & \forall uv \in E
\\
&~~~ \sum_{v \in V} \pi(v) f(v) f(v)^T \preceq I 
\\
&~~~ \sum_{v \in V} \pi(v) \norm{f(v)}^2 = k. 
\end{align*}

Finally, as in~\cite{Roc05}, note that the inner maximization problem is just a linear program in $Q$.
For a fixed embedding $f: V \to \R^n$, we can use the linear programming duality theorem to rewrite the inner maximization problem into the following minimization problem:
\begin{align*}
\min_{g:V \to \R_{\geq 0}} &~~~ \sum_{v \in V} \pi(v) g(v) & 
\\
\st &~~~ g(u) + g(v) \geq \norm{f(u)-f(v)}^2 & & \forall uv \in E,
\end{align*}
where $g(u)$ is a dual variable for the constraint $\sum_{v \in V} Q(u,v) = \pi(u)$.
Recall that we assumed the graph has a self-loop $Q(v,v)$ at each vertex $v$ so that the primal program is always feasible, and the primal variable $Q(v,v)$ gives the dual constraint $g(v) \geq 0$.

To summarize, we rewrite the max-min optimization problem in the primal program into a min-min optimization problem using von Neumann's minimax theorem and linear programming duality.
The resulting program in the statement is a semidefinite program in the vector program form.
\end{proof}

\subsection{Gaussian Projection}

The second step is to project a solution to $\kappa(G)$ in \autoref{def:sigma-k-dual} into a low-dimensional solution and prove that several properties are approximately preserved.
The projection algorithm is a high dimensional version of the simple Guassian projection algorithm in \autoref{sec:projection}.

\begin{definition}[Gaussian Projection] \label{def:projection-higher}
Let $f : V \to \R^n$ be an embedding where each vertex $v$ is mapped to a vector $f(v) \in \R^n$.
Given an integer $h \leq n$,
let $\Gamma$ be an $h \times n$ matrix where each entry $\Gamma_{i,j}$ for $1 \leq i \leq h$ and $1 \leq j \leq n$ is an independent standard Gaussian random variable $N(0,1)$.
The Gaussian projection $\bar{f} : V \to \R^h$ of $f$ is an embedding of each vertex $v \in V$ to an $h$-dimensional vector defined as 
\[\bar{f}(v) = \frac{1}{\sqrt{h}} \cdot \Gamma \big(f(v)\big).\]
\end{definition}

As in \autoref{sec:projection}, we consider a related directed program $\directed{\kappa}(G)$ for the analysis of the Gaussian projection algorithm.
The proof of the following lemma is the same as in \autoref{lem:factor2} and \autoref{def:compact} and is omitted.

\begin{lemma}[Directed Dual Program Using Orientation for $\kappa(G)$] \label{lem:orientation-kappa}
Let $G=(V,E)$ be an undirected graph and $\pi$ be a probability distribution on $V$.
Let $\directed{E}$ be an orientation of the undirected edges in $E$.
Define 
\begin{eqnarray*}
\directed{\kappa}(G) := \min_{f: V \to \R^n} \min_{\directed{E}} & & \sum_{v \in V} \pi(v) \max_{u: uv \in \directed{E}} \norm{f(u)-f(v)}^2
\\
\st & & \sum_{v \in V} \pi(v) f(v) f(v)^T \preceq I_n
\\
& & \sum_{v \in V} \pi(v) \norm{f(v)}^2 = k.
\end{eqnarray*}
Then $\kappa(G) \leq \directed{\kappa}(G) \leq 2\kappa(G)$.
\end{lemma}

To prove that several properties of $f$ to $\directed{\kappa}(G)$ are preserved in $\bar{f}$,
we define the following quantities.
The first quantity is the objective value of $\directed{\kappa}(G)$, which is called the energy of the function $f$.

\begin{definition}[Energy] \label{def:energy}
Given a directed graph $G = (V, \directed{E})$ and a probability distribution $\pi$ on $V$, the energy of an embedding $f : V \to \R^h$ is defined as 
\[
{\cal E}(f) = \sum_{v \in V} \pi(v) \max_{u: u \to v} \norm{f(u)-f(v)}^2.
\]
\end{definition}

The second quantity is the LHS of the last constraint, which is called the mass of the function $f$.

\begin{definition}[Mass] \label{def:mass}
Given an embedding $f: V \to \R^h$, the mass of $f$ is defined as
\[
\mu(f) = \sum_{v \in V} \pi(v) \norm{f(v)}^2.
\]
\end{definition}

The final quantity is related to the constraint $\sum_{v \in V} \pi(v) f(v) f(v)^T \preceq I_n$, which is called the sub-isotropy condition for the vectors $\big\{\sqrt{\pi(v)} f(v)\big\}_{v \in V}$.
In~\cite{LOT12}, the sub-isotropy condition is used to establish the following spreading property, which is used crucially in the spectral partitioning algorithm for $k$-way edge conductance.

\begin{definition}[Spreading Property~\cite{LOT12}] \label{def:spreading}
Let $\pi$ be a probability distribution on $V$. 
For two parameters $\Delta \in [0, 1]$ and $\eta \in [0, 1]$, 
an embedding $f: V \rightarrow \R^h$ is called $(\Delta, \eta)$-spreading if for every subset $S \subseteq V$,
\[
{\rm diam}_{d_f}(S) \le \Delta
\quad \implies \quad
\sum_{v \in S} \pi(v) \norm{f(v)}^2 \le \eta \cdot \sum_{v \in V} \pi(v) \norm{f(v)}^2,
\]
where ${\rm diam}_{d_f}(S) := \max_{u,v \in S} d_f(u,v)$ is the diameter of the set $S$ under the radial projection distance function $d_f$ to be defined in \autoref{def:radial-distance}. 
\end{definition}

As we will state formally in \autoref{lem:isotropy-and-spreading} in the next subsection, any feasible solution $f$ to $\directed{\kappa}(G)$ is $(\Delta, \frac{1}{k(1-\Delta^2)})$-spreading, so we can regard $\eta = \frac{1}{k(1-\Delta^2)}$ in the following.
The precise parameters and also the definition of the radial projection distance are not important in this subsection.

The goal of this subsection is to prove that the energy, the mass, and the spreading property of $f : V \to \R^n$ are approximately preserved in the projection $\bar{f} : V \to \R^h$ for a small $h$.
In~\cite{LOT12}, it was already proved that the mass and the spreading property of $f$ are approximately preserved in $\bar{f}$.

\begin{lemma}[\cite{LOT12}, Lemma 4.3] \label{lem:LOT-dimension-reduction}
Let $\pi$ be a probability distribution on $V$.
Let $f : V \to \R^n$ be an embedding
 that is $(\Delta,\eta)$-spreading.
Let $\bar{f} : V \to \R^h$ be a Gaussian projection of $f$ as defined in \autoref{def:projection-higher}.
For some value 
\footnote{
Lemma 4.3 in~\cite{LOT12} was stated slightly differently.
Their assumptions are that $f : V \to \R^k$ and $\eta \geq 1/k$, and their conclusion is that $h \lesssim \frac{1}{\Delta^2} \log (\frac{k}{\Delta})$.
We note that the dependency on $k$ in their conclusion is based on the substitution $\eta = 1/k$ in the bound on $h$ we stated, which has no dependency on the ambient dimension $n$.
Their proof, without the substitution $\eta = 1/k$, gives the bound we stated.
}
\[
h \lesssim \frac{1}{\Delta^2} \Big( \log\Big( \frac{1}{\eta \Delta}\Big) \Big),
\] 
with probability at least $1/2$, the following two properties hold simultaneously:
\[
\mu(\bar{f}) \geq \frac{1}{2} \mu(f)
\quad {\rm and} \quad
\bar{f} {\rm~is~} \Big(\frac{\Delta}{4}, \big(1+\Delta\big)\eta \Big){\rm-spreading}.
\]
\end{lemma}

We prove that the energy is also approximately preserved.
We use the following proposition whose proof is deferred to \autoref{app:higher-order-vertex}.

\begin{proposition}[Expected Maximum of $\chi$-Squared Distribution] \label{prop:max-of-gaussians}
Let $(\Gamma_{ij})$ for $1 \le i \le d$ and $1 \le j \le m$ be Gaussian random variables with mean $0$ and variance at most $1$, and such that $\Gamma_{i1}, \Gamma_{i2}, \dots, \Gamma_{im}$ are mutually independent for each $i \in [d]$. Let
$Y_i := \frac{1}{m} \sum_{1 \leq j \le m} \Gamma_{ij}^2$ and let $Y := \max_{1 \leq i \leq d} Y_i$. Then,
  \[
    \E[Y] \le 4 \left(
      1 + \frac{1 + \log d}{m}
    \right).
  \]
\end{proposition}

The main result in this subsection is the following lemma which compares the energy, the mass, and the spreading property of $f$ and of its Gaussian projection $\bar{f}$.
Note that the second and the third items of the following lemma are directly from \autoref{lem:LOT-dimension-reduction}.

\begin{lemma}[Dimension Reduction] \label{lem:dimension-reduction}
Let $G = (V, \directed{E})$ be a directed graph with maximum indegree $d$ and $\pi$ be a probability distribution on $V$.
Let $f : V \to \R^n$ be an embedding that is $(\Delta,\eta)$-spreading.
Let $\bar{f} : V \to \R^h$ be a Gaussian projection of $f$ as defined in \autoref{def:projection-higher}.
By setting $h \lesssim \frac{1}{\Delta^2} \log (\frac{1}{\eta \Delta})$, with probability at least $1/4$, the following three properties hold simultaneously:
\[
{\cal E}(\bar{f}) \lesssim \Big( 1 + \frac{\log d}{h}\Big) \cdot {\cal E}(f)
\quad {\rm and} \quad
\mu(\bar{f}) \geq \frac{1}{2} \mu(f)
\quad {\rm and} \quad
\bar{f} {\rm~is~} \Big(\frac{\Delta}{4}, \big(1+\Delta\big)\eta \Big){\rm-spreading}.
\]
\end{lemma}

\begin{proof}
The second and the third item are from \autoref{lem:LOT-dimension-reduction}.
We will prove the first item holds with probability at least $3/4$ by using \autoref{prop:max-of-gaussians}, and this would imply the lemma by union bound.
Let $\Gamma_i$ be the $i$-th row of $\Gamma$ in \autoref{def:projection-higher}.
For $uv \in \directed{E}$,
\begin{eqnarray*}
\norm{\bar{f}(u) - \bar{f}(v)}^2
&=&
\sum_{i = 1}^h \bigg( \frac{1}{\sqrt{h}} \biginner{\Gamma_i}{f(u) - f(v)} \bigg)^2
= \frac{1}{h} \sum_{i=1}^h g_i^2,
\end{eqnarray*}
where $g_i = \inner{\Gamma_i}{f(u) - f(v)}$ is an independent Gaussian random variable with mean zero and variance $\norm{f(u)-f(v)}^2$.
Applying Proposition \ref{prop:max-of-gaussians} on each $v \in V$ with indegree at most $d$,
\[
  \E\Big[\max_{u:u \to v} \norm{\bar{f}(u) - \bar{f}(v)}^2 \Big]
  \le
  4 \left(
  1 + \frac{1 + \log d}{h}
  \right) \cdot \max_{u: u \to v} \norm{f(u) - f(v)}^2.
\]
By linearity of expectation and Markov's inequality,
\[
  \Pr \left[
    \sum_{u \in V} \pi(u) \max_{u: u \to v} \norm{\bar{f}(u) - \bar{f}(v)}^2
    \le
    16 \Big(1 + \frac{1 + \log d}{h} \Big) \cdot \sum_{u \in V} \pi(u) \max_{u: u \to v} \norm{f(u) - f(v)}^2
  \right] \ge \frac{3}{4},
\]
implying that ${\cal E}(\bar{f}) \lesssim (1 + \frac{\log d}{h}) \cdot {\cal E}(f)$ with probability at least $3/4$.
\end{proof}

\subsection{Spectral Partitioning} \label{sec:embedding}

The third step is to show that given $\bar{f} : V \to \R^h$ in \autoref{lem:dimension-reduction}, we can construct $l$ disjointly supported functions $\bar{f}_1, \ldots, \bar{f}_l : V \to \R^h$ with comparable energy and mass to that of $\bar{f}$.

\begin{lemma}[Spectral Partitioning] \label{lem:spectral-partitioning}
Let $G = (V,E)$ be an undirected graph and $\pi$ be a probability distribution on $V$.
Let $\directed{E}$ be an orientation of $E$ and $\bar{f} : V \to \R^h$ be an embedding.
Let $l$ be the targeted number of disjointly supported functions where $1 \leq l \leq k$.
Suppose that there exist $\Delta \in (0,1)$ and $\delta \in (0,1)$ such that
\[
\bar{f} {\rm~is~} \Big( \frac{\Delta}{4}, \frac{1}{k(1-\Delta)} \Big){\rm-spreading}
\qquad {\rm and} \qquad
\Delta \leq 1 - \frac{2(l-1)}{2(1-\delta)k - 1}.
\]
Then there exist embeddings $\bar{f}_1, \ldots, \bar{f}_l : V \to \R^h$ such that the supports of $\{\bar{f}_i\}_{i=1}^l$ are pairwise disjoint and 
\[
{\cal E}(\bar{f}_i) \lesssim \Big( 1 + \frac{h}{\Delta \delta} \Big)^2 \cdot {\cal E}(\bar{f}) 
\qquad {\rm and} \qquad 
\mu(\bar{f}_i) \gtrsim \frac{1}{k} \mu(\bar{f}).
\]
\end{lemma}

The proof of this step follows closely the proof in~\cite{LOT12}, 
so we will first review the ideas and results in~\cite{LOT12} before presenting the proof of \autoref{lem:spectral-partitioning}.

\subsubsection{Review of~\cite{LOT12}} \label{sec:review-LOT}

In~\cite{LOT12}, given the first $k$ eigenvectors $f_1,\ldots,f_k : V \to \R$ of the normalized Laplacian matrix, the spectral embedding $f: V \to \R^k$ is defined for each vertex $v \in V$ as $f(v) = \big(f_1(v), \ldots, f_k(v)\big)$. 
Since the eigenvectors are orthonormal, the spectral embedding satisfies the isotropy condition $\sum_{v \in V} f(v) f(v)^T = I_n$.
Lee, Oveis Gharan and Trevisan observed that the isotropy condition implies that not many points can be close in the radial projection distance defined below.

\begin{definition}[Radial Projection Distance~\cite{LOT12}] \label{def:radial-distance}
Let $G=(V,E)$ be a graph and $f: V \rightarrow \R^h$ be an embedding of the vertices. 
For each pair of vertices $u,v \in V$, 
the radial projection distance between $u$ and $v$ is defined as
\[
    d_f(u, v) = \norm{\frac{f(u)}{\norm{f(u)}} - \frac{f(v)}{\norm{f(v)}}}
\]
if $\norm{f(u)} > 0$ and $\norm{f(v)} > 0$.
Otherwise, if $f(u) = f(v) = 0$ then $d_f(u,v) := 0$, else $d_f(u,v) = \infty$.
\end{definition}

More precisely, they proved the following bound on the spreading property in \autoref{def:spreading} of the embedding $f$.
Their result is stated for an embedding $f$ satisfying the isotropy condition, but the same proof works for an embedding $f$ satisfying the sub-isotropy condition in \autoref{def:sigma-k-dual}.

\begin{lemma}[Sub-Isotropy Implies Spreading~\cite{LOT12}(Lemma 3.2)] \label{lem:isotropy-and-spreading}
Let $G = (V, E)$ be an undirected graph and $\pi$ be a probability distribution on $V$.
Suppose $f: V \rightarrow \R^h$ is an embedding with mass $\mu(f)$ in \autoref{def:mass}.
Then, for any $\Delta \in [0,1)$,
\[
\sum_{u \in V} \pi(u) f(u) f(u)^T \preceq I_h
\quad \implies \quad
f {\rm~is~} \Big( \Delta, \frac{1}{\mu(f) \cdot (1 - \Delta^2)} \Big){\rm-spreading}.
\]
\end{lemma}

As the embedding is spreading, any subset of points with small diameter in radial projection distance cannot have too much mass.
In order to construct many disjointly supported embeddings $f_1,\ldots,f_l:V \to \R^h$,
the points in $\R^h$ are partitioned into many groups of small diameter using the following definition and theorem from metric geometry.

\begin{definition}[Padded Decomposition~\cite{LOT12}] \label{def:padded-decomposition}
Let $(X, d_X)$ be a finite metric space. 
For $\Delta, \alpha, \delta > 0$, a random partition $\mathcal{P}$ of $X$ is called $(\Delta, \alpha, \delta)$-padded if
\vspace{-2mm}
\begin{itemize}
    \item each part in $\mathcal{P}$ has diameter at most $\Delta$ with respect to the distance function $d_X$;
    \item $\Pr\big[B\big(x, \frac{\Delta}{\alpha}\big) \subseteq \mathcal{P}(x)\big] \ge \delta$ for every $x \in X$, 
where $B\big(x, \frac{\Delta}{\alpha}\big)$ is the open ball of radius $\frac{\Delta}{\alpha}$ centered at $x$ and $\mathcal{P}(x)$ is the part in $\mathcal{P}$ that contains $x$.
  \end{itemize}
\end{definition}

\begin{theorem}[\cite{LOT12}~Theorem 2.3, \cite{GKL03}]
  \label{thm:padded-partition}
Let $(X,d)$ be a finite metric space.
If $X \subseteq \R^h$, 
then for every $\Delta > 0$ and $\delta \in (0, 1)$, 
$X$ admits a $\big(\Delta, O(\frac{h}{\delta}), 1 - \delta\big)$-padded random partition.
\end{theorem}

Let ${\cal P} = P_1 \sqcup P_2 \sqcup \ldots \sqcup P_m$ be a random partition sampled from \autoref{thm:padded-partition}.
By the second property in \autoref{def:padded-decomposition},
there is only a small fraction of points close to the boundary of the partition.
The points close to the boundary are removed to form $P_1' \sqcup P'_2 \sqcup \ldots \sqcup P_m'$, so that the distance between each pair $P_i'$ and $P_j'$ is lower bounded by say $2\eps$.
As each $P'_i$ does not have too much mass, they can be grouped into disjoint sets $S_1, \ldots, S_k$ where each has mass at least $\frac{1}{2}$.
Then the disjoint supported functions $f_1, \ldots, f_k$ are constructed on $S_1, \ldots, S_k$ by the following smooth localization procedure.

\begin{lemma}[Smooth Localization~\cite{LOT12}(Lemma 3.3)] \label{lem:smooth-localization}
Let $G=(V,E)$ be an undirected graph and $f : V \to \R^h$ be an embedding.
For any $S \subseteq V$ and any $\eps > 0$,
there is a mapping $f' : V \to \R^h$ which satsifies the following three properties: 
\begin{enumerate}
\item
$f'(v) = f(v)$ for all $v \in S$,
\item
$\supp(f') \subseteq N_{\eps}(S)$ where $N_{\eps}(S) := \{v \in V \mid \exists u \in S {\rm~with~} d_f(u,v) \leq \eps\}$ denotes the set of vertices with radial projection distance at most $\eps$ from $S$,
\item
for $uv \in E$, $\norm{f'(u) - f'(v)} \leq \big(1+\frac{2}{\eps}\big) \norm{f(u)-f(v)}$.
\end{enumerate}
\end{lemma}

To summarize, the energy of each $f_i$ is upper bounded by the third item in \autoref{lem:smooth-localization}, and the mass of each $f_i$ is lower bounded by $\frac{1}{2k}$ fraction of the total mass by the spreading property in \autoref{lem:isotropy-and-spreading}.

\subsubsection{Proof of \autoref{lem:spectral-partitioning}}

As our proof follows closely the steps in~\cite{LOT12}, the review in the previous subsubsection also serves well as an overview of our proof.

Given an embedding $\bar{f} : V \to \R^h$ 
and a target $l \leq k$,
we would like to find $l$ disjoint subsets $S_1, \ldots, S_l$ of $V$ such that 
\begin{enumerate}
\item 
for $1 \leq i \leq l$, the mass $\mu(S_i) := \sum_{u \in S_i} \pi(u) \norm{\bar{f}(u)}^2$ of each $S_i$ is at least $\frac{1}{2k} \cdot \mu(V)$, where $\mu(V)=\mu(\bar{f})$ in \autoref{def:mass}, and
\item 
for $1 \leq i \neq j \leq l$, the distance $d_{\bar{f}}(S_i,S_j) := \min_{u \in S_i, v \in S_j} d_{\bar{f}}(u,v)$ between $S_i$ and $S_j$ is at least $2\eps$ for some $\eps > 0$ to be determined later, where $d_{\bar{f}}$ is the radial projection distance in \autoref{def:radial-distance}.
\end{enumerate}
To this end, equip $V$ with the pseudo-metric $d_{\bar{f}}$ and consider the metric space $(V, d_{\bar{f}})$. 
Let $\mathcal{P} = P_1 \sqcup P_2 \sqcup \cdots \sqcup P_m$ be a $\big(\frac{\Delta}{4}, \frac{ch}{\delta}, 1 - \delta\big)$-padded random partition sampled from \autoref{thm:padded-partition}, where $c$ is a universal constant and $\Delta \in (0, 1)$ and $\delta \in (0, 1)$ are to be determined. 
By the assumption that $\bar{f}$ is $(\frac{\Delta}{4}, \frac{1}{k(1 - \Delta)})$-spreading, 
the first property in \autoref{def:padded-decomposition} implies that $\mu(P_i) \leq \frac{1}{k(1 - \Delta)} \cdot \mu(\bar{f})$ for $1 \leq i \leq m$.
Let $U := \{x \in V \mid B(x, \frac{\Delta \delta}{4 c h}) \not \subseteq P(x) \} $ be the set of points that are close to the boundaries of ${\cal P}$.
The second property in \autoref{def:padded-decomposition} implies that there exists a realization of $\mathcal{P}$ such that $\mu(V-U) \geq (1 - \delta) \cdot \mu(\bar{f})$.
We take such a realization $\mathcal{P} = P_1 \sqcup P_2 \sqcup \cdots \sqcup P_m$ and remove all points in $U$ to obtain $P_i' := P_i - U$ for $1 \leq i \leq m$.
By doing so, we end up with disjoint sets $P_1', P_2', \dots, P_m'$ with the following properties: 
\begin{enumerate}
\item $\mu(P_i') \leq \frac{1}{k(1 - \Delta)} \cdot \mu(\bar{f})$ for $1 \leq i \leq m$, 
\item $\sum_{i=1}^m \mu(P_i') \geq (1 - \delta) \cdot \mu(\bar{f})$,
\item $d_{\bar{f}}(P_i', P_j') \ge \frac{2\Delta \delta}{ch}$ for $i \neq j \in [m]$.
\end{enumerate}
Next, we will merge some of the sets $P_1', \ldots, P_m'$ to form disjoint sets $S_1, \ldots, S_l$ so that $\mu(S_i) \geq \frac{1}{2k} \cdot \mu(\bar{f})$ for $1 \leq i \leq l$. 
This can be done by a simple greedy process, where we sort the $P_i'$ by nonincreasing mass, and put consecutive sets into a group $S_j$ until $\mu(S_j) \geq \frac{1}{2k} \cdot \mu(\bar{f})$.
By the first property and the greedy process, each group $S_j$ produced has $\mu(S_j) \leq \frac{1}{k(1 - \Delta)} \cdot \mu(f)$.
Hence, by the second property, the greedy process will succeed to produce at least $l$ groups with mass at least $\frac{1}{2k} \cdot \mu(\bar{f})$ as long as
\[
(1-\delta) \cdot \mu(\bar{f}) - (l-1) \cdot \frac{\mu(\bar{f})}{k(1 - \Delta)} \geq \frac{\mu(\bar{f})}{2k} 
\quad \iff \quad 
\Delta \le 1 - \frac{2(l-1)}{2(1 - \delta)k - 1},
\]
which is exactly the assumption we made in the statement about $\Delta$.
Therefore, we can produce $S_1, \ldots, S_l$ satisfying the two requirements $\mu(S_i) \geq \frac{1}{2k} \cdot \mu(\bar{f})$ for $1 \leq i \leq l$ and $d_{\bar{f}}(S_i,S_j) \geq 2\eps := \frac{2\Delta \delta}{ch}$ for $i \neq j \in [l]$ by the third property of $P'_1, \ldots, P'_m$.

Now, we apply the smooth localization procedure in \autoref{lem:smooth-localization} on each $S_i$ with $\eps = \frac{\Delta \delta}{ch}$ to obtain an embedding $\bar{f}_i : V \to \R^h$ for $1 \leq i \leq l$.
First, we check that $\bar{f}_1, \ldots, \bar{f}_l$ are disjointly supported.
This follows from $d_{\bar{f}}(S_i, S_j) \geq 2\eps$ for $i \neq j$ and the second property in \autoref{lem:smooth-localization}.
Second, since $\mu(S_i) \geq \frac{1}{2k} \cdot \mu(\bar{f})$ and $\bar{f}_i(v) = \bar{f}(v)$ for $v \in S_i$ by the first property in \autoref{lem:smooth-localization}, it follows that $\mu(\bar{f}_i) = \mu(S_i) \geq \frac{1}{2k} \cdot \mu(\bar{f})$.
Finally, by the third property in \autoref{lem:smooth-localization}, it follows that
\[
{\cal E}(\bar{f}_i)
= \sum_{v \in V} \pi(v) \max_{u: u \to v} \norm{\bar{f}_i(u) - \bar{f}_i(v)}^2
\leq \Big(1 + \frac{2c h}{\Delta \delta} \Big)^2 
\sum_{v \in V} \pi(v) \max_{u: u \to v} \norm{\bar{f}(u) - \bar{f}(v)}^2
\lesssim \Big(1 + \frac{h}{\Delta \delta} \Big)^2 {\cal E}(\bar{f}).
\]
Therefore, we conclude that $\bar{f}_1, \ldots, \bar{f}_l$ satisfy all the properties stated in \autoref{lem:spectral-partitioning}.

\subsection{Cheeger Rounding} \label{sec:wrapup-higher}

The fourth step is to apply the results in \autoref{sec:Cheeger-vertex} on $\bar{f}_1, \ldots, \bar{f}_l$ from \autoref{lem:spectral-partitioning} to obtain disjoint subsets with small vertex expansion.

\begin{lemma}[Cheeger Rounding] \label{lem:higher-order-rounding}
Let $G = (V, E)$ be an undirected graph with maximum degree $d$ and $\pi$ be a probability distribution $\pi$ on $V$.
Given an orientation $\directed{E}$ and an embedding $\bar{f}: V \rightarrow \R^h$, there exists a set $S \subseteq \supp(\bar{f})$ with
\[
\psi(S)^2 \lesssim \min\{h, \log d\} \cdot \frac{{\cal E}(\bar{f})}{\mu(\bar{f})}.
\]
\end{lemma}

\begin{proof}
Given $\directed{E}$ and $\bar{f}$, we apply the Gaussian projection step in \autoref{prop:improved-projection} to obtain a $1$-dimensional embedding $x : V \to \R$ with ${\cal E}(x) / \mu(x) \lesssim \log d \cdot {\cal E}(\bar{f}) / \mu(\bar{f})$.
Alternatively, if $h \leq \log d$, we can choose the best coordinate from $\bar{f}$ to obtain a $1$-dimensional embedding $x : V \to \R$ with ${\cal E}(x) / \mu(x) \leq h \cdot {\cal E}(\bar{f}) / \mu(\bar{f})$ as was done in~\cite{OTZ22}.
So we have a $1$-dimensional embedding $x$ with ${\cal E}(x) / \mu(x) \lesssim \min\{h, \log d\} \cdot {\cal E}(\bar{f}) / \mu(\bar{f})$.

Then, we apply the threshold rounding step in \autoref{prop:threshold-rounding} on $x$ to obtain a set $S \subseteq \supp(x) \subseteq \supp(\bar{f})$ with $\directed{\psi}(S)^2 \lesssim {\cal E}(x) / \mu(x) \leq \min\{h,\log d\} \cdot {\cal E}(\bar{f}) / \mu(\bar{f})$.
Finally, we apply the postprocessing step in \autoref{lem:cleanup} to obtain a set $S' \subseteq S$ with $\psi(S') \leq 2\directed{\psi}(S)$ satisfying the requirements of this lemma.
\end{proof}

We are ready to put together the steps to prove the hard direction of the higher-order Cheeger inequality for vertex expansion.

\begin{theorem}[Hard Direction for Multiway Vertex Expansion] \label{thm:higher-order-hard}
Let $G = (V, E)$ be an undirected graph with maximum degree $d$ and $\pi$ be a probability distribution $\pi$ on $V$.
For any $2 \leq k \leq n$ and $0 \le \eps < 1$,
let $\xi := \max\{\eps,\frac{1}{2k}\}$, it holds that
\[
\psi_{(1-\eps)k}(G) \lesssim \frac{\log k}{\xi^4} \cdot \sqrt{\log d \cdot \sigma_k^*(G)}.
\]
\end{theorem}

\begin{proof}
The first step is to compute an optimal solution $(f^*,g^*)$ to the dual program in \autoref{def:sigma-k-dual} with objective value $\kappa(G) = \sigma_k^*(G)$.
Then, we use \autoref{lem:orientation-kappa} to obtain a solution $f : V \to \R^n$ to the directed program $\directed{\kappa}(G)$ with energy ${\cal E}(f) \leq 2\sigma_k^*(G)$ and $\mu(f) = k$.
As $f$ satisfies the sub-isotropy condition $\sum_u \pi(u) f(u) f(u)^T \preceq I_n$ in $\directed{\kappa}(G)$, we know from Proposition \ref{lem:isotropy-and-spreading} that $f$ is $\big(\Delta, 1/\big(k(1-\Delta^2)\big)\big)$-spreading for any $\Delta \in (0, 1)$ of our choice.

The second step is to apply the Gaussian projection algorithm in \autoref{lem:dimension-reduction} on $f$ with $\eta = 1/\big(k(1-\Delta^2)\big)$ to obtain $\bar{f}: V \rightarrow \R^h$ with
\[
h \lesssim \frac{1}{\Delta^2} \log \Big(\frac{1}{\Delta \eta} \Big) 
= \frac{1}{\Delta^2} \log \Big(\frac{k (1 - \Delta^2)}{\Delta} \Big)
\]
such that
\[
{\cal E}(\bar{f}) \lesssim \Big(1 + \frac{\log d}{h}\Big) \cdot {\cal E}(f)
\quad {\rm~and~} \quad 
\mu(\bar{f}) \gtrsim \mu(f)
\quad {\rm~and~} \quad
\bar{f} {\rm~is~} \Big(\frac{\Delta}{4}, \frac{1}{k(1-\Delta)}\Big){\rm-spreading}.
\]

The third step is to apply the spectral partitioning algorithm in \autoref{lem:spectral-partitioning} on $f$.
Let $l := (1-\eps)k$ be the target number of output sets.
By setting
\[
\Delta = \min \Big\{ \frac{1}{2}, \frac{2(k-l) + 1}{2(k+l) - 3} \Big\} 
\quad {\rm and} \quad 
\delta = \frac{2(k-l)+1}{4k},
\]
we can check that the conditions of \autoref{lem:spectral-partitioning} are satisfied, and so we can construct functions $\bar{f}_1, \dots, \bar{f}_l: V \rightarrow \R^h$ with disjoint support, such that for each $1 \leq i \leq l$ it holds that
\[
{\cal E}(\bar{f}_i) \lesssim \Big(1 + \frac{h}{\Delta \delta} \Big)^2 \cdot {\cal E}(\bar{f}) 
\quad {\rm and} \quad 
\mu(\bar{f}_i) \gtrsim \frac{1}{k} \mu(\bar{f}).
\]

The fourth step is to apply \autoref{lem:higher-order-rounding} to $\bar{f}_1, \ldots \bar{f}_l$ to obtain disjoint subsets $S_1, \dots, S_l$, such that for every $1 \leq i \leq l$,
\begin{eqnarray*}
\psi(S_i)^2
&\lesssim&
\min\{h, \log d\} \cdot \frac{{\cal E}(\bar{f}_i)}{\mu(\bar{f}_i)}
\\
&\lesssim&
\min\{h, \log d\} \cdot k \cdot \Big(1 + \frac{h}{\Delta \delta} \Big)^2 \cdot \frac{{\cal E}(\bar{f})}{\mu(\bar{f})}
\\
&\lesssim&
\min\{h, \log d\} \cdot \Big(1 + \frac{\log d}{h}\Big) \cdot \Big(1 + \frac{h}{\Delta \delta}\Big)^2 \cdot k \cdot \frac{{\cal E}(f)}{\mu(f)}
\\
&\lesssim&
\log d \cdot \Big( 1 + \frac{h}{\Delta \delta} \Big)^2 \cdot \sigma_k^*(G)
\\
&\lesssim&
\log d \cdot \frac{1}{\Delta^6 \delta^2} \cdot \log^2 \Big(\frac{k (1 - \Delta^2)}{\Delta} \Big) \cdot \sigma_k^*(G),
\end{eqnarray*}
where the fourth inequality uses that $\min\{h, \log d\} \cdot \big(1 + \frac{\log d}{h}\big) \le 2 \log d$ and $\mu(f) = k$, and the inequality uses the fact that $\frac{h}{\Delta \delta} \ge 1$.
This implies that
\[
\psi_l(G) \lesssim \frac{1}{\Delta^3 \delta} \log\Big( \frac{k}{\Delta} \Big) \cdot \sqrt{\log d \cdot \sigma_k^*(G)}.
\]
Finally, we plug in $l=(1-\eps)k$ and consider two cases.
In the case when $\eps \geq \frac{1}{2k}$,
we see that $\Delta = \Theta(\eps)$ and $\delta = \Theta(\eps)$, and so $\psi_l(G) \leq \frac{1}{\eps^4} \cdot \log k \cdot \sqrt{\log d \cdot \sigma_k^*(G)}$.
In the case when $\eps < \frac{1}{2k}$, we simply set $l = k$ and see that $\Delta = \Theta(1/k)$ and $\delta = \Theta(1/k)$ and so $\psi_l(G) \leq k^4 \cdot \log k \cdot \sqrt{\log d \cdot \sigma_k^*(G)}$. 
Combining the two cases proves the theorem.
\end{proof}

We prove the following easy direction in \autoref{app:higher-order-vertex}.
Note that it is about $\lambda_k^*(G)$ instead of $\sigma_k^*(G)$.

\begin{lemma}[Easy Direction for Multiway Vertex Expansion] \label{lem:higher-order-easy}
For any undirected graph $G = (V, E)$ and any probability distribution $\pi$ on $V$,
$\lambda_k^*(G) \leq 2\psi_k(G)$ for any $k \geq 2$.
\end{lemma}

Combining \autoref{lem:higher-order-easy} and \autoref{thm:higher-order-hard}, we conclude this section with the following higher-order Cheeger inequality for vertex expansion that implies \autoref{thm:higher-order-vertex}.

\begin{theorem}[Higher-Order Cheeger Inequality for Vertex Expansion] \label{thm:higher-order-vertex-refined}
For any undirected graph $G = (V,E)$ with maximum degree $d$ and any probability distribution $\pi$ on $V$,
\[
\frac{\sigma_k^*(G)}{k} 
\leq \lambda_k^*(G) 
\lesssim \psi_k(G) 
\lesssim k^4 \log k \sqrt{\log d \cdot \sigma_k^*(G)}
\leq k^{\frac{9}{2}} \log k \sqrt{\log d \cdot \lambda_k^*(G)}
\]
Furthermore, for any $1 > \eps \geq \frac{1}{2k}$, 
\[
\psi_{(1-\eps)k}(G) 
\lesssim \frac{1}{\eps^4} \log k \sqrt{\log d \cdot \sigma_{k}^*(G)} 
\leq\frac{1}{\eps^4} \log k \sqrt{k \log d \cdot \lambda_{k}^*(G)} .
\]
\end{theorem}
\section{Improved Cheeger Inequality for Vertex Expansion} \label{sec:improved-Cheeger-vertex}

The goal of this section is to prove \autoref{thm:improved-Cheeger-vertex}.
The proof in~\cite{KLLOT13} has two main steps.
The first step is to prove that if the second eigenfunction is close to a $k$-step function, then the approximation guarantee of threshold rounding is improved.
The second step is to prove that if $\lambda_k$ is large for a small $k$, then the second eigenfunction is close to a $k$-step function.

We follow the a similar plan to prove an improved version of the Cheeger inequality for $\gamma^{(1)}(G)$ in \autoref{thm:weighted-Cheeger}, where in the second step we replace $\lambda_k$ by $\sigma_k^*(G)$.
First we begin with the definition of a $k$-step function.

\begin{definition}[$k$-Step Function and Approximation] \label{def:k-step}
Let $G = (V, E)$ be an undirected graph and $\pi$ be a probability distribution on $V$.
Given $y: V \rightarrow \R$ and $1 \le k \le n$, we call $y$ a $k$-step function if the number of distinct values in $\{y(u)\}_{u \in V}$ is at most $k$.

Given $x: V \to \R$, we say $y$ is a $k$-step $\eps$-approximation to $x$ if $y$ is a $k$-step function and $\norm{x-y}_{\pi} \leq \eps$, where $\norm{z}_{\pi}^2 := \sum_{v \in V} \pi(v) z(v)^2$ for a vector $z : V \to \R$.
\end{definition}

The following is the precise statement of the first step for $\gamma^{(1)}(G)$, which informally says that if there is a good $k$-step approximation to an optimal solution to $\gamma^{(1)}(G)$ then the performance of threshold rounding is better than that in \autoref{thm:weighted-Cheeger}.

\begin{proposition}[Rounding $k$-Step Approximation] \label{prop:k-step-rounding}
Let $G = (V, E)$ be an undirected graph and $\pi$ be a probability distribution on $V$. 
For any feasible solution $(f, g)$ to the $\gamma^{(1)}(G)$ program with objective value $\gamma_f$ and any $k$-step function $y_f: V \rightarrow \R$ approximating $f$,
\[
  \psi(G) \lesssim k \cdot \gamma_f + k \norm{f - y_f}_{\pi} \sqrt{\gamma_f}.
\]
\end{proposition}

Our second step is to prove that if $\sigma_k^*(G)$ in \autoref{def:sigma-k-dual} is large for a small $k$, then there is a good $k$-step approximation to a good solution to $\gamma^{(1)}(G)$.

\begin{proposition}[Constructing $k$-Step Approximation] \label{prop:k-step-construction}
Let $G = (V, E)$ be an undirected graph and $\pi$ be a probability distribution on $V$. 
For any feasible solution $(f, g)$ to the $\gamma^{(1)}(G)$ program with objective value $\gamma_f$,
there exists a $k$-step function $y$ with
\[
\norm{f - y}_{\pi}^2 \lesssim \frac{k \cdot \gamma_f}{\sigma_k^*(G)}.
\] 
\end{proposition}

Assuming \autoref{prop:k-step-rounding} and \autoref{prop:k-step-construction},
we prove an exact analog of the improved Cheeger's inequality in~\cite{KLLOT13} for vertex expansion, with $\sigma_k^*(G) / k$ playing the role of $\lambda_k^*(G)$ .

\begin{theorem}[Improved Cheeger Inequality for Vertex Expansion] \label{thm:improved-Cheeger-vertex-expansion}
For any undirected graph $G = (V,E)$ and any probability distribution $\pi$ on $V$ and any $k \geq 2$,
\[
\gamma^{(1)}(G) \lesssim \psi(G) 
\lesssim k \cdot \gamma^{(1)}(G) \cdot \sqrt{ \frac{k}{\sigma_k^*(G)} }.
\] 
\end{theorem}

\begin{proof}
The easy direction is proved in \autoref{lem:Cheeger-easy}.
For the hard direction, let $(f^*, g^*)$ be an optimal solution to the $\gamma^{(1)}(G)$ program in \autoref{def:Roch-dual-1D} with objective value $\gamma^*$, and $\sigma^*$ be the optimal value of the $\sigma_k^*(G)$ program in \autoref{def:sigma-k-dual}.
By \autoref{prop:k-step-construction}, 
there exists a $k$-step function $y: V \rightarrow \R$ with
\[
\norm{f^* - y}_{\pi}^2 \lesssim \frac{k \cdot \gamma^*}{\sigma^*}.
\]
Applying \autoref{prop:k-step-rounding} with $y$, it follows that
\begin{eqnarray*}
\psi(G) \lesssim
k \cdot \big( \gamma^* + \norm{f^* - y}_{\pi} \cdot \sqrt{\gamma^*} \big)
\lesssim k \cdot \gamma^* \Big( 1 + \sqrt{\frac{k}{\sigma^*}} \Big)
\lesssim k \cdot \gamma^* \cdot \sqrt{\frac{k}{\sigma^*}}. 
\end{eqnarray*}
\end{proof}

Note that \autoref{thm:improved-Cheeger-vertex} follows immediately from \autoref{thm:improved-Cheeger-vertex-expansion}.

\begin{proofof}{\autoref{thm:improved-Cheeger-vertex}}
\[
\lambda_2^*(G) =
\gamma(G) \leq 
\gamma^{(1)}(G) \lesssim \psi(G) 
\lesssim k \cdot \gamma^{(1)}(G) \cdot \sqrt{ \frac{k}{\sigma_k^*(G)} }
\lesssim k \cdot \log d \cdot \gamma(G) \cdot \sqrt{ \frac{k}{\lambda_k^*(G)} },
\]
where we use 
$\gamma(G) = \lambda_2^*(G)$ in \autoref{prop:Roch-dual} 
and $\gamma(G) \leq \gamma^{(1)}(G) \lesssim \log d \cdot \gamma(G)$ in \autoref{prop:improved-projection} 
and $\sigma_k^*(G) \geq \lambda_k^*(G)$ in \autoref{def:sigma-k-dual}.
\end{proofof}

\begin{remark}[Tight Examples] \label{rem:tight}
We remark that \autoref{thm:improved-Cheeger-vertex-expansion} is tight.
The loss in \autoref{thm:improved-Cheeger-vertex} is because of the factor $\log d$ loss in the dimension reduction step and the factor $k$ loss in the transition from $\lambda_k^*(G)$ to $\sigma_k^*(G)$.

As an example, let $G$ be an $n$-cycle, where $n$ is odd, and $\pi$ the uniform distribution. Suppose $k \ll n$. Then $\sigma_k^*(G) = \sigma_k(G)$, because the only possible ``reweighting'' is the one with equal edge weight. Since $\psi(G) = \Theta(1/n), \gamma^{(1)}(G) = O(1/n^2)$ (choose $f: V \rightarrow \R$ that maps vertex $l \in [n]$ to point $C(1 - \frac{4}{n} \cdot \min(l, n - l))$, where $C = \Theta(1)$ is a normalizing factor, and $g(l) \equiv \frac{8C^2}{n^2}$), and
\[
  \sigma_k(G) =
  \sum_{l=0}^{k-1}
  \left(1 - \cos \left(\frac{2 \pi l}{n} \right) \right)
  = \Theta \left(\frac{k^3}{n^2} \right),
\]

one can verify that in this case the hard direction of \autoref{thm:improved-Cheeger-vertex-expansion} is tight.
\end{remark}

We prove \autoref{prop:k-step-rounding} and \autoref{prop:k-step-construction} in the following two subsections.

\subsection{Rounding $k$-Step Approximation}

We prove \autoref{prop:k-step-rounding} in this subsection.
First, we do some preprocessing on the solution $(f,g)$ as in \autoref{sec:Cheeger-vertex}.
Then, as in~\cite{KLLOT13}, the main step is to use a modified probability distribution on the thresholds based on the $k$-step approximation $y$ to analyze the threshold rounding algorithm.

Given a feasible solution $(f, g)$ to the $\gamma^{(1)}(G)$ program with objective value $\gamma_f$, 
we use \autoref{lem:factor2} and \autoref{def:compact} to obtain a solution $f$ and $\directed{E}$ to the directed program $\directed{\gamma}^{(1)}(G)$ in \autoref{def:directed-gamma} with objective value at most $2\gamma_f$.
Then we apply the truncation step in \autoref{lem:truncation} on $f$ and $\directed{E}$ to obtain a solution $x$ and $\directed{E}$ with $x \geq 0$ and $\pi(\supp(x)) \leq 1/2$ and 
\[
\sum_{v \in V} \pi(v) \max_{u:uv \in \directed{E}} (x(u)-x(v))^2 \leq 8 \gamma_f
\quad {\rm and} \quad 
\sum_{v \in V} \pi(v) x(v)^2 = 1.
\]
Let $y_f$ be the $k$-step approximation of $f$ in the assumption of \autoref{prop:k-step-rounding}.
Note that the way we construct $x$ from $f$ consists of shifting, truncating, and scaling by a factor of at most $2$.
So, applying the same transformations to $y_f$ will still give us a $k$-step function $y$ with $\norm{x-y}_{\pi} \leq 2 \norm{f - y_f}_{\pi}$.
Henceforth, we work with $x$ and its $k$-step approximation $y$.

The main step is to prove that applying the threshold rounding algorithm on $x$ will find a set $S \subseteq V$ with small directed vertex expansion $\directed{\psi}(S)$ as defined in \autoref{def:vertex-cover}.
For the analysis, we take the $k$-step approximation $y$ into consideration and give higher weight to a threshold $t$ if $t$ is far away from the function values in $y$.

The following weighting scheme is from~\cite{KLLOT13}.
Let $x_{\max} := \max_{v \in V} x(v)$.
Suppose the $k$-step function $y$ takes values $0 \leq y_1 \leq y_2 \leq \cdots \leq y_k$.
For $t \in [0,x_{\max}]$, define $\rho(t) := \min_{1 \leq i \leq k} |t - y_i|$.
In words, $\rho(t)$ is the distance from $t$ to the closest value of the $k$-step function $y$.
We sample $t \in [0,x_{\max}]$ with probability proportional to $\rho(t)$.
That is, for $0 \leq a < b \leq x_{\max}$,
\[
  \Pr\big[t \in [a, b]\big] = \frac{1}{Z} \int_a^b \rho(t) \, dt,
\]
where $Z := \int_0^{x_{\max}} \rho(t) dt$ is the normalizing factor.

We use the distribution described above to analyze the threshold rounding algorithm.
Let $S_t := \{v \in V \mid x(v) > t\}$ be a threshold set.
By a standard averaging argument,
\[
\min_{t \in [0,x_{\max}]} \directed{\psi}(S_t)
\leq \frac{\E_t [ \pi( \directed{\partial}S_t) ]}{\E_t [ \pi(S_t) ]}
\]
The denominator is
\begin{eqnarray*}
\E_t [\pi(S_t)]
~=~ \frac{1}{Z} \cdot \sum_{v \in V} \pi(v) \cdot \int_0^{x(v)} \rho(t) \,dt
~\gtrsim~ \frac{1}{kZ} \sum_{v \in V} \pi(v) x(v)^2
= \frac{1}{kZ},
\end{eqnarray*}
where the inequality can be seen as follows:
For any $v \in V$, let $k' \le k$ be the largest index so that 
$0 \le y_1 \le \dots \le y_{k'} \le x(v)$ and let $z_0 := 0$, $z_i := y_i$ for $1 \leq i \leq k'$ and $z_{k'+1} := x(v)$, then
\begin{eqnarray*}
\int_0^{x(v)} \rho(t) \, dt
= \sum_{i=0}^{k'} \int_{z_i}^{z_{i+1}} \rho(t) \, dt
\geq \frac{1}{4} \sum_{i=0}^{k'} (z_{i+1}-z_i)^2
\geq \frac{1}{4(k'+1)} \bigg( \sum_{i=0}^{k'} (z_{i+1}-z_i) \bigg)^2
\gtrsim \frac{1}{k} x(v)^2,
\end{eqnarray*}
where the first inequality is by simple calculus and the second inequality is by Cauchy-Schwarz.

The numerator is
\begin{eqnarray*}
  \E_t \Big[ \pi\big(\directed{\partial}S_t\big) \Big]
  &=&
  \sum_v \pi(v) \cdot \Pr_t[\text{ there exists an edge } uv \in \directed{E} \text{ with } uv \in \delta(S_t) ] 
  \\
  &=&
  \sum_v \pi(v) \cdot \max_{u:u \to v} \Pr_t\big[uv \in \delta(S_t)\big]
 \\
  &=&
  \sum_v \pi(v) \cdot \max_{u: u \to v} \frac{1}{Z} \cdot \bigg| \int_{x(v)}^{x(u)} \rho(t) \, dt \bigg|
  \\
  &\le&
  \frac{1}{Z} \sum_v \pi(v) \cdot \max_{u:u \to v}
  \bigg| \int_{x(v)}^{x(u)} \Big( \rho\big(x(v)\big) + \big|t - x(v)\big| \Big)\, dt \bigg|
  \\
  &=&
  \frac{1}{Z} \sum_v \pi(v) \cdot \max_{u: u \to v}
  \bigg[ \big|x(u) - x(v)\big| \cdot \rho\big(x(v)\big) + \frac{1}{2} \big(x(u) - x(v)\big)^2 \bigg]
  \\
  &\leq&
  \frac{1}{Z} \cdot \Big(
    4\gamma_f + \sum_v \pi(v) \cdot \rho\big(x(v)\big) \cdot \max_{u: u \to v} \big|x(u) - x(v)\big| \Big)
  \\
  &\le&
  \frac{1}{Z} \cdot \bigg( 4 \gamma_f +
  \sqrt{ \Big( \sum_v \pi(v) \cdot \rho\big(x(v)\big)^2 \Big)
  \cdot \Big( \sum_v \pi(v) \max_{u: u \to v} \big(x(u) - x(v)\big)^2 \Big) } \bigg)
  \\
  &\le&
  \frac{1}{Z} \cdot \Big(
    4 \gamma_f + \sqrt{\norm{x - y}_{\pi}^2 \cdot 8\gamma_f} \Big)
  \\
  &\lesssim&
  \frac{1}{Z} \cdot \big(\gamma_f + \norm{x - y}_\pi \cdot \sqrt{\gamma_f}\big),
\end{eqnarray*}
where the first inequality is because $\rho(t)$ is a $1$-Lipschitz function,
the third inequality is by Cauchy-Schwarz,
and the fourth inequality is because $\rho(x(v)) \leq |x(v) - y(v)|$ by the definition of $\rho$.

Combining the bounds on the denominator and the numerator,
\[
\min_{t \in [0,x_{\max}]} \directed{\psi}(S_t) ~\leq~
  \frac {\E_t \big[ \pi\big(\directed{\partial} S_t\big) \big]}
    {\E_t \big[ \pi(S_t)\big]}
  ~\lesssim~ k \gamma_f + k \norm{x - y}_{\pi} \cdot \sqrt{\gamma_f}
  ~\lesssim~ k \gamma_f + k \norm{f-y_f}_{\pi} \cdot \sqrt{\gamma_f}.
\]
Since $\pi(\supp(x)) \le \frac{1}{2}$, the output set $S$ has $\pi(S) \leq \frac{1}{2}$. 
Finally, we apply the postprocessing step in \autoref{lem:cleanup} to obtain a set $S' \subseteq S$ with $\psi(S') \leq 2\directed{\psi}(S)$.
We conclude that 
\[
\psi(G) ~\leq~ \psi(S') ~\lesssim~ \min_t \directed{\psi}(S_t) 
~\lesssim~ \gamma_f + k \norm{f-y_f}_{\pi} \cdot \sqrt{\gamma_f}. 
\]

\subsection{Constructing $k$-Step Approximation}

We prove \autoref{prop:k-step-construction} in this subsection.
The high level plan is similar to that in~\cite{KLLOT13}.
Given a feasible solution $(f,g)$ to the $\gamma^{(1)}(G)$ program, we aim to construct a good $k$-step approximation $y$ of $f$ using a simple procedure.
If we fail to do so, then we show that $f$ can be used to construct a good $k$-dimensional solution $\bar{f} = (\bar{f}_1, \bar{f}_2, \ldots, \bar{f}_k)$ to the $\sigma_k^*(G)$ program, contradicting the value of $\sigma_k^*(G)$ is large.
Therefore, the simple process must succeed to find a good $k$-step approximation $y$ of $f$.

Suppose our $k$-step function $y$ takes values $y_1 \leq y_2 \leq \cdots \leq y_k$.
For convenience $y_0 := -\infty$.
We use these values to define $k$ disjoint subsets $S_1, \ldots, S_k \subseteq V$ where $S_i := \{v \in V \mid y_{i-1} < f(v) \leq y_i\}$,
and define functions $\bar{f}_1, \ldots, \bar{f}_k$ supported on $S_1, \ldots, S_k$ respectively where
\[
  \bar{f}_i(v) := \begin{cases}
    \min\big\{|y_i - f(v)|, |f(v) - y_{i-1}|\big\}, & \text{ if } f(v) \in (y_{i-1}, y_i)
    \\
    0 & \text{ otherwise.}
  \end{cases}
\]
The role of $\bar{f}_i$ is to measure how well the two threshold values $y_{i-1}$ and $y_i$ approximate the values of the vertices in $S_i$ in $f$.
We would like to choose $y_1 \leq y_2 \leq \cdots \leq y_k$ so that each $\norm{\bar{f}_i}_{\pi}^2$ is small and $S_1 \cup \cdots \cup S_k = V$.
As we will show, this would imply that there exists a $k$-step function $y$ with threshold values $y_1 \leq y_2 \leq \cdots \leq y_k$ so that $\norm{f-y}_{\pi} = \sum_{i=1}^k \norm{\bar{f}_i}^2_{\pi}$ is small, and thus $y$ is a good $k$-step approximation to $f$.

Consider the following simple procedure to choose the threshold values $y_1 \leq y_2 \leq \ldots \leq y_k$ of the $k$-step function.
Let $\mu$ be a parameter to be determined later.
We would like to guarantee that $\norm{\bar{f}_i}_\pi^2 \leq \mu$ for $1 \leq i \leq k$ and $S_1 \cup S_2 \cup \cdots \cup S_k = V$ such that all vertices are covered by these functions $\bar{f}_1, \ldots, \bar{f}_k$.
We choose $y_i$ successively as follows:
given $y_0, \ldots, y_{i-1}$, set $y_i$ to be the smallest number such that $\norm{\bar{f}_i}_{\pi}^2 = \mu$.
If the smallest $y_i$ does not exist, then we set $y_i=\ldots=y_k := \max_{v \in V} f(v)$ and call the procedure a ``success''.
Otherwise, if we set all $y_1 \leq y_2 \leq \ldots \leq y_k$ and $S_1 \cup S_2 \cup \cdots S_k \subset V$ such that not all vertices are covered by $\bar{f}_1, \ldots, \bar{f}_k$, then we call the procedure a ``failure''.

Now we set 
\[
\mu = \frac{2 \gamma_f}{\sigma_k^*(G)}.
\]

When the procedure succeeds, then $\norm{\bar{f}_i}_{\pi}^2 \leq \mu$ for $1 \leq i \leq k$ and $S_1 \cup S_2 \cup \cdots \cup S_k = V$.
We would like to show that the there exists a $k$-step function $y$ with threshold values $y_1 \leq y_2 \leq \ldots \leq y_k$ such that $\norm{y-f}_{\pi}^2 \leq k \mu$.
Define $y : V \to \R$ to be
\[
  y(v) := \argmin_{\alpha \in \{y_1, \dots, y_{k}\}} |\alpha - f(v)|,
\]
such that each vertex $v$ is assigned to its closest threshold value.
Then 
\[
\norm{f-y}_{\pi}^2 
= \sum_{v \in V} \pi(v) \cdot \big| f(v) - y(v) \big|^2
= \sum_{i=1}^k \sum_{v \in S_i} \pi(v) \bar{f}_i(v)^2
= \sum_{i=1}^k \norm{\bar{f}_i}_{\pi}^2 
\leq k \mu
\lesssim \frac{k \gamma_f}{\sigma_k^*(G)},
\]
where the second equality uses that $S_1 \cup S_2 \cup \cdots \cup S_k = V$ and the first inequality uses that each $\norm{\bar{f}_i}_{\pi}^2 \leq \mu$.
Thus $y$ is a $k$-step function that satisfies the statement of the proposition.

To complete the proof, we would like to show that the procedure always succeeds.
When the procedure fails, then there exist $y_1 < y_2 < \cdots < y_k$ such that $\norm{\bar{f}_i}_{\pi}^2 = \mu$ for $1 \leq i \leq k$.
We will construct from $\bar{f}_1, \ldots, \bar{f}_k$ a solution $(\bar{f}, \bar{g})$ to the $\sigma_k^*(G)$ program with objective value less than $\sigma_k^*(G)$, thus arriving at a contradiction.
Define $\bar{f}: V \rightarrow \R^n$ and $\bar{g}: V \rightarrow \R$ as follows:
\[
\bar{f}(v) := \Big(\frac{\bar{f}_1(v)}{\sqrt{\mu}}, \dots, \frac{\bar{f}_k(v)}{\sqrt{\mu}}, 0, \dots, 0 \Big)^T
\quad {\rm and} \quad
\bar{g}(v) := \frac{1}{\mu} g(v).
\]

We will check that $(\bar{f}, \bar{g})$ is a feasible solution to the $\sigma_k^*(G)$ program in \autoref{def:sigma-k-dual}. 
For the sub-isotropy condition, note that each $\bar{f}(v)$ has at most one nonzero entry, and
\begin{eqnarray*}
  \sum_{v \in V} \pi(v) \bar{f}(v) \bar{f}(v)^T
  &=&
  \diag \Big(\frac{1}{\mu} \sum_{v \in S_1} \pi(u) \bar{f}_1(u)^2, \frac{1}{\mu} \sum_{v \in S_2} \pi(v) \bar{f}_2(u)^2, \dots, \frac{1}{\mu} \sum_{v \in S_k} \pi(v) \bar{f}_k(v)^2, 0, \dots, 0 \Big)
  \\
  &=&
  \frac{1}{\mu} \diag\Big( \norm{\bar{f}_1}_{\pi}^2, \norm{\bar{f}_2}_{\pi}^2, \dots, \norm{\bar{f}_k}_{\pi}^2, 0, \dots, 0\Big)
  ~=~ \diag(1, 1, \dots, 1, 0, \dots, 0) ~\preceq~ I_n.
\end{eqnarray*}
The mass constraint is satisfied as
\[
\sum_{v \in V} \pi(v) \norm{\bar{f}(v)}^2
= \tr \bigg( \sum_{u \in V} \pi(u) \bar{f}(u) \bar{f}(u)^T \bigg)
= \tr \Big( \diag\big(1, 1, \dots, 1, 0, \dots, 0\big) \Big)
= k.
\]
For the constraint on each edge $uv \in E$,
\[
\norm{\bar{f}(u) - \bar{f}(v)}^2 
= \frac{1}{\mu} \sum_{i=1}^k \big(\bar{f}_i(u) - \bar{f}_i(v)\big)^2
\leq \frac{1}{\mu} \big(f(u) - f(v)\big)^2  
\leq \frac{1}{\mu} \big(g(u) + g(v)\big)
= \bar{g}(u) + \bar{g}(v),
\]
where for the first inequality we consider two cases:
(i) suppose $u \in S_i$ and $v \in S_j$ for $i = j$,
then $\sum_{l=1}^k \big(\bar{f}_l(u) - \bar{f}_l(v) \big)^2 = \big(\bar{f}_i(u) - \bar{f}_i(v) \big)^2 \le \big(f(u) - f(v)\big)^2$,
and (ii) suppose $u \in S_i$ and $v \in S_j$ for $i \neq j$, then
$  \sum_{l=1}^k \big(\bar{f}_l(u) - \bar{f}_l(v)\big)^2 = \big(\bar{f}_i(u) - \bar{f}_i(v)\big)^2 + \big(\bar{f}_j(u) - \bar{f}_j(v)\big)^2 \le \big(f(u) - f(v)\big)^2$
since $\big|\bar{f}_i(u) - \bar{f}_i(v)\big| + \big|\bar{f}_j(u) - \bar{f}_j(v)\big| \le \big|f(u) - f(v)\big|$.

Therefore, $(\bar{f}, \bar{g})$ is a feasible solution to the $\sigma_k^*(G)$ program, and its objective value is
\[
  \sum_{v \in V} \pi(v) \bar{g}(v) = \frac{1}{\mu} \sum_{v \in V} \pi(v) g(v) = \frac{1}{\mu} \cdot \gamma_f = \frac{\sigma_k^*(G)}{2} < \sigma_k^*(G),
\]
a contradiction to the definition of $\sigma_k^*(G)$. 

To conclude, the procedure must succeed and return a $k$-step function $y$ with $\norm{f-y}_{\pi}^2 \lesssim k\gamma_f / \sigma_k^*(G)$.
\section{Vertex Expansion of $0/1$-Polytopes} \label{sec:polytopes}

The goal of this section is to present a construction of $0/1$-polytopes with poor vertex expansion as described in \autoref{thm:01-vertex-expansion}, which has implications about sampling from the uniform distribution as described in \autoref{sec:intro-polytope}.

A $0/1$-polytope is defined by a subset of vertices in the boolean hypercube $\{0,1\}^n$.
Our examples are based on the following simple probabilistic construction.

\begin{definition}[Probabilistic Construction] \label{def:construction}
Let $n$ be an even number and $k < n/2$.
For a binary string $x \in \{0,1\}^n$, denote its $1$-norm by $|x| := \sum_{i=1}^n |x_i|$.
The set of vertices of our constructed polytope is the union of three subsets:
\begin{enumerate}
\item A left part $L := \{ x \in \{0,1\}^n \mid |x| = k\}$ consists of all binary strings with $k$ ones.
\item A right part $R := \{ x \in \{0,1\}^n \mid |x| = n-k\}$ consists of all binary strings with $n-k$ ones.
\item A middle part $M \subset \{ x \in \{0,1\}^n \mid |x| = n/2\}$ consists of $\Theta(4^k n^2)$ number of uniformly random binary strings with $n/2$ ones.
\end{enumerate}
\end{definition}

The graph $G_Q = (V,E)$ of a polytope $Q$ is defined as the $1$-skeleton of the polytope $Q$.
Our plan is to prove that for a random polytope $Q$ constructed in \autoref{def:construction}, the middle part $M$ ``blocks'' all the edges between $L$ and $R$ in $G_Q$ with constant probability.
This would imply that $\partial L \subseteq M$ and $\partial R \subseteq M$, and thus $\psi(L), \psi(R) \lesssim 4^k n^2 / n^k$ are very small.

The organization of this section is as follows.
First, in \autoref{sec:blocking}, we provide a sufficient condition for two binary strings $x,y$ to have no edge in $G_Q$, using geometric arguments.
Then, in \autoref{sec:probabilistic}, we outline the main probabilistic argument to prove \autoref{thm:01-vertex-expansion}, by using a union bound over the set of linear threshold functions.
Finally, we show \autoref{cor:01-slow-uniform} in \autoref{sec:consequence}.
We defer all the proofs in \autoref{sec:blocking} and \autoref{sec:probabilistic} to \autoref{app:01-polytope}.

As mentioned in \autoref{rem:Gillmann} in the introduction, Gillmann \cite{Gil07} has constructed similar examples of 0/1-polytopes with poor vertex expansion. Following the same simple argument as in \autoref{sec:consequence}, one obtains analogous lower bounds on the fastest mixing time of these polytopes.
We remark that the construction in \cite{Gil07} is similar to our construction, but the proofs are different and so we present our proofs even though the results follow from \cite{Gil07}.

\subsection{A Sufficient Condition for Edge Blocking} \label{sec:blocking}

Let $Q$ be a $0/1$-polytope and $G_Q = (V,E)$ be its graph/$1$-skeleton.
For two binary strings $x,y \in \{0,1\}^n$,
if $xy$ is an edge in $G_Q$, then there is a separating hyperplane $l$ with $l(x),l(y) \geq 0$ while $l(z) < 0$ for all other binary strings $z$ in the $0/1$-polytope $Q$.

In the construction of $Q$ in \autoref{def:construction},
if $x \in L$ and $y \in R$ then $\frac{1}{2}(x+y)$ has $1$-norm equal to $n/2$.
If $xy$ is an edge in $G_Q$, then there is a separating hyperplane $l$ with $l\big(\frac12(x+y)\big) \geq 0$ while $l(z) < 0$ for all other binary strings $z$ in the middle part $M$.
So, if we could establish that $\frac12 (x+y)$ is in the convex hull ${\rm conv}(M)$ of $M$ for all $x \in L$ and $y \in R$, then there are no edges between $L$ and $R$ in the graph $G_Q$.
This is the sufficient condition that we will formalize.

In the analysis, we use the following definitions to group the pairs of vertices $x \in L$, $y \in R$ based on their common patterns.

\begin{definition}[Patterns]
For $n \in \N$, a pattern is an element $p \in \{0, 1, \lor \}^n$, where $0$, $1$, and $\lor$ are regarded as symbols.
    
The support of a pattern $p \in \{0, 1, \lor \}^n$ is defined as $\supp(p) := \{i \in [n] \mid p_i \neq \lor\}$. 
We also define $\supp_0(p) := \{i \in [n] \mid p_i = 0\}$ and $\supp_1(p) := \{i \in [n] \mid p_i = 1\}$.
    
Given two binary strings $x, y \in \{0, 1\}^n$, their common pattern $p^{(x, y)} \in \{0, 1, \lor\}^n$ is defined as
    \[
      p^{(x, y)}_i =
      \begin{cases}
        0, & \text{ if } x_i = y_i = 0 \\
        1, & \text{ if } x_i = y_i = 1 \\
        \lor, & \text{ if } x_i \neq y_i.
      \end{cases}
    \]
Given a pattern $p \in \{0, 1, \lor\}^n$ and a binary string $x \in \{0, 1\}^n$, $x$ is said to match $p$ if and only if $p_i \neq \lor$ implies $p_i = x_i$.
\end{definition}

For each pattern $p$, we consider a potential separating hyperplane of the following specific form, which will be convenient for the probabilistic analysis

\begin{definition}[Consistent Affine Function] \label{def:consistent}
Let $p \in \{0, 1, \lor \}^n$ be a pattern.
An affine function $l: (u_1, u_2, \dots, u_n) \in \R^n \mapsto \beta + \sum_i \alpha_i u_i$ is called $p$-consistent if
\[
\alpha_i = 0 {\rm~for~} i \in \supp(p)
\qquad {\rm and} \qquad
\beta + \frac{1}{2} \sum_{i: i \notin \supp(p)} \alpha_i = 0.
\]
If $p$ is the common pattern of $x$ and $y$,
then $l$ being $p$-consistent implies $l\big( \frac12 (x+y) \big) = 0$.
\end{definition}

We formulate the sufficient condition described above for the middle part $M$ blocking the edge $xy$ for $x \in L$ and $y \in R$ using the definitions that we have developed. The proof is deferred to \autoref{app:01-polytope}.

\begin{lemma}[Blocking One Edge] \label{lem:blocking-one-edge}
Let $Q = L \cup M \cup R$ be a $0/1$-polytope from \autoref{def:construction}. 
Let $x \in L$, $y \in R$ and $p = p^{(x, y)} \in \{0, 1, \lor\}^n$ be the common pattern of $x$ and $y$. 
If for any $p$-consistent affine function $l$ there exists a point $z \in M$ matching the pattern $p$ and satisfying $l(z) \ge 0$,
then there is no edge connecting $x$ and $y$ in the graph of $Q$.
\end{lemma}

The following is a sufficient condition for the middle part $M$ blocking all the edges between $L$ and $R$, by considering all possible common patterns of an $x \in L$ and a $y \in R$.

\begin{lemma}[Blocking All Edges] \label{lem:blocking-all-edges}
Let $Q = L \cup M \cup R$ be a $0/1$-polytope from \autoref{def:construction}. 
Suppose for every pattern $p \in \{0,1,\lor\}^n$ with $|\supp_0(p)|=|\supp_1(p)| \leq k$ and for any $p$-consistent affine function $l : \R^n \to \R$,
there is $z \in M$ matching the pattern $p$ with $l(z) \geq 0$.
Then there are no edges between $L$ and $R$ in the graph of $Q$.
\end{lemma}

\subsection{Probabilistic Analysis} \label{sec:probabilistic}

Our plan is to use the sufficient condition in \autoref{lem:blocking-all-edges} to prove that a random $M$ with not many points can block all the edges between $L$ and $R$.
To this end,
we prepare with two simple lemmas about the probability that $z \in M$ satisfying $l(z) \geq 0$ and matching a particular pattern $p$ with $|\supp_0(p)| = |\supp_1(p)| \leq k$.

The geometric intuition of the first lemma is simple:
when we restrict a $p$-consistent affine function $l$ on the coordinates in $[n] \setminus \supp(p)$, then $l$ is an ``unbiased'' hyperplane that goes through the point $\frac{1}{2}\cdot \vec{1}$ on $[n] \setminus \supp(p)$ (because of the second condition in \autoref{def:consistent}), and thus a random vertex in $M$ matching the pattern $p$ lies on the non-negative side of $l$ with probability at least $1/2$.

\begin{lemma} \label{lem:good}
Let $Q = L \cup M \cup R$ be a $0/1$-polytope from \autoref{def:construction}. 
Let $p$ be the common pattern of  $x \in L$ and $y \in R$, and $l$ be a $p$-consistent affine function.
Let $Z$ be the uniform distribution on $\{z \in \{0, 1\}^n: |z| = \frac{n}{2}\}$. 
Then,
  \[
    \Pr_{z \sim Z} \big[ \, 
      l(z) \ge 0 \mid z \text{ matches pattern } p \,
    \big] \ge \frac{1}{2}.
  \]
\end{lemma}

The second lemma gives a lower bound on the probability that a random point $z \sim Z$ matches a pattern $p$ with $|\supp_0(p)| = |\supp_1(p)| \leq k$. 

\begin{lemma} \label{lem:match-pattern}
Let $p \in \{0, 1, \lor\}^n$ be a pattern with $|supp_0(p)| = |supp_1(p)| = s \le k$, 
and let $Z$ be the uniform distribution on $\{z \in \{0, 1\}^n: |z| = n/2\}$. 
Then
  \[
    \Pr_{z \sim Z} [z \text{ matches pattern } p] \gtrsim 4^{-s}.
  \]
\end{lemma}

With the above two lemmas, we can show that for any pattern $p$ with $|\supp_0(p)| = |\supp_1(p)| = s \le k$ and any $p$-consistent affine function $l$, 
the probability that a random point in $\{z \in \{0, 1\}^n: |z| = n/2\}$ matches the pattern $p$ and satisfies $l(z) \ge 0$ with probability not too small. 
So, by adding enough number of random points in the middle part $M$, such a point $z$ exists in $M$ with high probability for a fixed $p$ and $l$.
Then, we would like to use a union bound over $p$ and $l$ to prove that there will be no edges between $L$ and $R$ in the graph of the polytope with constant probability.

One technical issue of this approach is that there are infinitely many affine functions $l : \R^n \rightarrow \R$. 
Note, however, that we only care about the values of $l$ on the hypercube vertices. 
This reduces the number of different functions to $2^{2^n}$.
Indeed, we only care about whether $l(z) \ge 0$ for $z \in \{0, 1\}^n$ for \autoref{lem:blocking-all-edges}.
Therefore, we only need to apply a union bound over the set of linear threshold functions over the boolean hypercube, which further reduces the number of different such functions to $2^{n^2}$.

\begin{proposition}[\cite{Cov65}] \label{prop:Cover}
The number of linear threshold functions on $\{0, 1\}^n$ is at most $2^{n^2}$. A linear threshold function on $\{0, 1\}^n$ is a function of the form $\tau: \{0, 1\}^n \rightarrow \{0, 1\}$, where
  \[
    \tau(u_1, \dots, u_n) =
    \begin{cases}
      1, & \text{ if } \beta + \sum_i \alpha_i u_i \ge 0;
      \\
      0, & \text{ if } \beta + \sum_i \alpha_i u_i < 0,
    \end{cases}
  \]
  for some $\alpha_1, \dots, \alpha_n, \beta \in \R$.
\end{proposition}

The complete proof of  \autoref{thm:01-vertex-expansion} can be found in \autoref{app:01-polytope}.

\subsection{Mixing Time} \label{sec:consequence}

\autoref{cor:01-slow-uniform} follows from \autoref{thm:01-vertex-expansion} and the easy direction of Cheeger's inequality for vertex expansion in \autoref{thm:OTZ22}.

\begin{proofof}{\autoref{cor:01-slow-uniform}}
Let $Q = L \cup M \cup R$ be a $0/1$-polytope from \autoref{def:construction}.
The number of vertices of $Q$ is $|L| + |M| + |R| \lesssim {n \choose k} \le (en/k)^k$.
By Theorem \ref{thm:01-vertex-expansion}, the vertex expansion of the graph $G_Q$ is $\psi(G_Q) \lesssim (4k)^k / n^{k-2}$.
Therefore, by the easy direction in \autoref{thm:OTZ22},
the mixing time of any reversible chain $P \in \R^{|V| \times |V|}$ on $G_Q$ with stationary distribution $\pi = \frac{1}{|V|} \vec{1}$ is at least  
\begin{eqnarray*}
\tau^*(G)
\gtrsim \frac{1}{\psi(G_Q)}
\gtrsim \frac{n^{k-2}}{(4k)^k}
\gtrsim \left( \frac{en}{k} \right)^{k-2}
\gtrsim |V|^{1 - \frac{2}{k}},
\end{eqnarray*}
where the second last inequality is by the assumption that $k$ is a constant and so only the exponent of $n$ matters, and the last inequality is by $|V| \leq (en/k)^k$ explained above.
\end{proofof}

The implication of \autoref{cor:01-slow-uniform} has been discussed in \autoref{sec:intro-polytope} and we won't repeat here.
\section{Tight Example to Cheeger's Inequality for Vertex Expansion}
\label{sec:tight-example}

The goal of this section is to construct a family of tight examples to \autoref{thm:Cheeger-vertex} when $\pi = \vec{1}/|V|$. The graphs constructed will have non-constant maximum degree $d$ and satisfy
\[
  \frac{\psi(G)^2}{\log d} \asymp \lambda_2^*(G).
\]
These examples are suggested to us by Shayan Oveis Gharan.

The graphs are realized as proximity graphs on $\S^{k-1}$.
We employ the following notations:
$\mu$ denotes the normalized Lebesgue measure on $\S^{k-1}$.
$\Delta(\cdot, \cdot)$ denotes the geodesic distance on $\S^{k-1}$.
For $\theta \in \S^{k-1}$ and $r > 0$, $B(\theta, r)$ denotes the set of points of distance less than $r$ from $\theta$.
$Cap(r)$ denotes a generic spherical cap on $\S^{k-1}$ of radius $r$.
For $S \subseteq \S^{k-1}$, we use $S_r$ or $S + Cap(r)$ to denote the set of points of distance less than $r$ from $S$.

\begin{definition}[Spherical Proximity Graph]
  Given $k$, a $(\gamma, \delta)$-spherical proximity graph in $k$ dimensions is a graph $G_k = (V, E)$ that can be constructed as follows:
  \begin{enumerate}
    \item Partition $\S^{k-1}$ into $n$ cells: $S_1, S_2, \dots, S_n$, such that each cell has diameter at most $\gamma$ and measure between $\eps := \mu(Cap(\gamma/4))$ and $2 \eps$. Take a point $x_i \in S_i$ for each $i \in [n]$.
    
    \item Set $V = \{x_i\}_{i \in [n]}$ and $E = \{(x_i, x_j) \in V \times V: \Delta(x_i, x_j) < \delta\}$.
  \end{enumerate}
\end{definition}

We shall show that, for suitable choices of $\gamma$ and $\delta$, the resulting graph will be a vertex expander with $\lambda_2^* = O(1/\log d)$.

\begin{theorem}[Tight Example of \autoref{thm:Cheeger-vertex}]
  \label{thm:sphere-example}
  Given any $k \in \N$, if we set $\gamma = c_1 / \sqrt{k}, \delta = c_2 / \sqrt{k}$ for constants $0 < c_1 < c_2$, then for any $(\gamma, \delta)$-spherical proximity graph $G_k$:
  \begin{itemize}
      \item the maximum degree $d$ of $G_k$ is $2^{O(k)}$;
      \item $\psi(G_k) = \Theta(1)$; and
      \item $\lambda_2^*(G_k) = O(1/\log d)$.
  \end{itemize}
\end{theorem}

The rest of the section is organized as follows.
First, we show that the construction of the spherical proximity graph is indeed possible.
Then, we prove the degree and expansion bounds, by relating them to volume ratios on the sphere.
Finally, we prove the bound on reweighted eigenvalue, using the spherical embedding of the graph.
All the proofs in this section are deferred to \autoref{app:tight-example}.

\subsection{The Graph Construction}

The following proposition is modified from {\cite[Lemma 8.3.22]{GM12}} and shows that we can always construct a spherical proximity graph.

\begin{proposition}[Constructing Spherical Proximity Graph]
  \label{prop:sphere-graph-construction}
  For every $k$ and $\gamma > 0$, there exists $n$ such that $\S^{k-1}$ can be partitioned into cells $S_1, \dots, S_n$ with ${\mathrm{diam}}(S_i) \le \gamma$ and $\mu(S_i) \in [\eps, 2 \eps]$ for all $i \in [n]$, where $\eps := \mu(Cap(\gamma / 4))$.
  Therefore, $n = \Theta(\mu(Cap(\gamma / 4))^{-1})$.
\end{proposition}

\subsection{Degree and Expansion Bounds}

Now, we establish bounds on the maximum degree and vertex expansion of $G_k$. We do so by connecting the quantities with the continuous notion of volume.

Recall our choice of parameters $\gamma = c_1 / \sqrt{k}$ and $ \delta = c_2 / \sqrt{k}$, where $0 < c_1 < c_2$ are constants.

\begin{lemma}[Degree, Expansion, and Volume]
  \label{lemma:degree-expansion-volume}
  Let $G = (V, E)$ be the graph constructed above.
  \begin{itemize}
      \item The maximum degree of $G$ is at most
      \[
        \frac{\mu(Cap(\delta + \gamma))}
                {\mu(Cap(\gamma/4))}.
      \]
      
      \item For any $T \subseteq V$ with $|T| \le |V|/2$, if we let $S_T := \cup_{i \in T} S_i$, then
      \[
        \psi(T) \ge \frac{\mu(S_T + Cap(\delta - \gamma)) - \mu(S_T)}{2 \mu(S_T)}.
      \]
  \end{itemize}
\end{lemma}

After relating the graph degree and vertex expansion to volumes of portions of sphere, we shall now use results from high-dimensional geometry to obtain bounds on the maximum degree and vertex expansion.

\begin{proposition}[Degree Bound]
  \label{prop:degree-bound-sphere-example}
  For constants $0 < c_1 < c_2$, it holds that
  \[
    \frac{\mu(Cap(\delta + \gamma))}
                {\mu(Cap(\gamma/4))}
    \le
    2^{O(k)}.
  \]
\end{proposition}

\begin{proposition}[Small-Volume Expansion of the Sphere]
  \label{prop:expansion-bound-sphere-example}
  For constants $0 < c_1 < c_2$ with $c_2 - c_1$ sufficiently large,
  it holds that for all $T \subseteq V$ with $|T| \le |V|/2$,
  \[
    \frac{\mu(S_T + Cap(\delta - \gamma)) - \mu(S_T)}{\mu(S_T)} \ge \Omega(1).
  \]
\end{proposition}

Combining \autoref{lemma:degree-expansion-volume}, \autoref{prop:degree-bound-sphere-example}, and \autoref{prop:expansion-bound-sphere-example} gives the first two parts of \autoref{thm:sphere-example}.

\subsection{Reweighted eigenvalue bound}

Finally, we bound $\lambda_2^*(G)$. On a high level, if a graph can be embedded in a unit sphere via $i \mapsto v_i$, such that the centre of mass is the origin and that each edge $ij$ satisfies $\norm{v_i-v_j} \leq \delta$, then the embedding certifies that $\lambda_2^*(G) \leq \delta^2$.
This is basically the proof.
The details can be found in \autoref{app:tight-example}.

\begin{proposition}[Reweighted Eigenvalue Bound]
  \label{prop:eigenvalue-bound-sphere-example}
  For constants $0 < c_1 < c_2$, it holds that
  \[
    \lambda_2^*(G_n) \lesssim \frac{1}{k} \lesssim \frac{1}{\log d}.
  \]
\end{proposition}

\autoref{prop:eigenvalue-bound-sphere-example} gives the last part of \autoref{thm:sphere-example}. Thus, the proof of \autoref{thm:sphere-example} is complete.

\section{Concluding Remarks}

We present a new spectral theory which relates (i) reweighted eignevalues, (ii) vertex expansion and (iii) fastest mixing time.
This is analogous to the classical spectral theory which relates (i) eigenvalues, (ii) edge conductance and (iii) mixing time.
This spectral approach for vertex expansion has the advantage that most existing results and proofs for edge conductances and eigenvalues have a close analog for vertex expansion and reweighted eigenvalues with almost tight bounds.
We do not intend to be exhaustive in this paper, and we fully expect that other results relating eigenvalues and edge conductances also have an analog for vertex expansion using reweighted eigenvalues.

To conclude, we believe that our work provides an interesting spectral theory for vertex expansion, as the formulations have the natural interpretation as reweighted eigenvalues and also have close connections to other important problems such as fastest mixing time and the reweighting conjectures in approximation algorithms.
We also believe that this approach can be extended further for hypergraph edge expansion.

\subsection*{Acknowledgements}

We thank Shayan Oveis Gharan for suggesting the tight example for \autoref{thm:Cheeger-vertex} in \autoref{sec:tight-example},
Robert Wang for suggesting the connection to the $0/1$-polytope expansion conjecture, and Sam Olesker-Taylor for providing insightful comments that improved the presentation of the paper.

\newpage

\begin{appendix}

\section{Deferred Proofs for Weighted Vertex Expansion} \label{app:Cheeger-vertex}

\begin{proofof}{\autoref{lem:Cheeger-easy}}
First we consider the case that there is an optimizer $S \subseteq V$ to $\psi(G)$,
with $0 < \pi(S) \leq 1/2$ and $\pi(\partial S)/\pi(S) = \psi(G)$.
For $a,b \in \R$ to be chosen below, we define a solution $f,g$ to $\gamma^{(1)}(G)$ in \autoref{def:Roch-dual-1D} as follows:
\[
  f(v) := \begin{cases}
    a, & \text{if } v \in S; \\
    b, & \text{if } v \not\in S,
  \end{cases}
  \qquad
  g(v) := \begin{cases}
    (a - b)^2, & \text{if } v \in \partial S; \\
    0, & \text{otherwise}.
  \end{cases}
\]
Then the constraints $g(u)+g(v) \geq (f(u)-f(v))^2$ for $uv \in E$ are satisfied by construction.
Note that we can always choose $a,b \neq 0$ that satisfy the two equations
$\pi(S) \cdot a + (1 - \pi(S)) \cdot b = 0$ and $\pi(S) \cdot a^2 + (1 - \pi(S)) \cdot b^2 = 1$ simultaneously.
This choice ensures that the two constraints $\sum_{v \in V} \pi(v) f(v) = 0$ and $\sum_{v \in V} \pi(v) f(v)^2 = 1$ are satisfied, and so $f,g$ as defined is a feasible solution to $\gamma^{(1)}(G)$.
The objective value is
\[
\sum_{v \in V} \pi(v) g(v) 
= \pi(\partial S) \cdot (a-b)^2
= \psi(G) \cdot \pi(S) \cdot (a-b)^2
\leq 2\psi(G) \cdot \pi(S) \cdot (a^2 + b^2)
\leq 2\psi(G),
\]
where the last inequality uses that $\pi(S) \leq 1/2$ and so $\pi(S) \cdot (a^2+b^2) \leq \pi(S) \cdot a^2 + (1-\pi(S)) \cdot b^2 = 1$.
This proves $\gamma^{(1)}(G) \leq 2\psi(G)$ in this case.

The other case is when $\psi(G)=1$.
We will show that $\gamma^{(1)}(G) \leq 2$ and this would imply that $\gamma^{(1)}(G) \leq 2\psi(G)$.
Let $v$ be a vertex with $0 < \pi(v) \leq 1/2$, which must exist as long as the graph has at least two vertices.
We define a solution $f,g$ to $\gamma^{(1)}(G)$ as above with $S=\{v\}$.
Following the same arguments, $f,g$ is a feasible solution to $\gamma^{(1)}(G)$ with objective value
\[
\sum_{v \in V} \pi(v) g(v) = \pi(v) (a-b)^2 \leq 2\pi(v)\big(a^2+b^2\big) 
\leq 2\big(\pi(v)a^2 + (1-\pi(v))b^2 \big) = 2.
\]
This proves $\gamma^{(1)} \leq 2 \psi(G)$ in the other case when $\psi(G)=1$.
\end{proofof}

\begin{proofof}{\autoref{lem:truncation}}
First we shift $y$ so that both $\pi(\{v \in V: y(v) > 0\})$ and $\pi(\{v \in V: y(v) < 0\}|)$ are at most $1/2$, which is always possible thanks to the constraint $\sum_{v \in V} \pi(v) y(v) = 0$.  
By shifting, the objective value of the solution does not change. 
Let $y^+$ and $y^-$ be the solutions with $y^+(v) := \max\{y(v),0\}$ and $y^-(v) := \min\{y(v),0\}$ for $v \in V$.
Note that the objective values for $y^+$ and $y^-$ are at most that for $y$, as
$|y^+(u) - y^+(v)| \leq |y(u)-y(v)|$ and $|y^-(u) - y^-(v)| \le |y(u) - y(v)|$.
By \autoref{lem:vector-rounding}, either $\sum_{v \in V} \pi(v) y^+(v)^2$ or $\sum_{v \in V} \pi(v) y^-(v)^2$ is at least $1/4$. 

We define $x$ to be $y^+$ if $\sum_{v \in V} \pi(v) y^+(v)^2 \geq \sum_{v \in V} \pi(v) y^-(v)^2$, or otherwise we define $x$ to be $-y^-$.
Then $x \geq 0$ and $\pi(\supp(x)) \leq 1/2$.
The numerator is at most the objective value of $y$ which is $\directed{\gamma}^{(1)}(G)$, and the denominator is at least $1/4$ by \autoref{lem:vector-rounding} and so the lemma follows. 
\end{proofof}

\begin{lemma} \label{lem:vector-rounding}
  Given a finite set $X$, a vector $f: X \rightarrow \R$, a distribution $\pi$ on $X$. Let
  \[
    \Var_{\pi}(f) :=
    \sum_{i \in X} \pi(i) f(i)^2
    -
    \bigg( \sum_{i \in X} \pi(i) f(i) \bigg)^2
  \]
  be the $\pi$-variance of $f$. Let $f^+ := \max(f, 0)$ and $f^- := \max(-f, 0)$ be the positive and negative parts of $f$ respectively (so that $f = f^+ - f^-$). Then,
  \[
    2\big(\Var_{\pi}(f^+) + \Var_{\pi}(f^-)\big) \ge \Var_{\pi}(f).
  \]
\end{lemma}
\begin{proof}
Let $S^+ \subseteq X := \{ i \in X: f(i) \ge 0\}$ and $S^- \subseteq X := \{ i \in X: f(i) < 0\}$. Suppose $\pi(S^+), \pi(S^-) > 0$; otherwise either $f^+ = f$ or $f^- = f$, and there is nothing to prove.
Then
\begin{eqnarray*}
  \Var_\pi(f^+)
  &=&
  \sum_{i \in S^+} \pi(i) f(i)^2 -
  \bigg( \sum_{i \in S^+} \pi(i) f(i) \bigg)^2
  \\
  &=&
  \pi(S^+) \sum_{i \in S^+} \frac{\pi(i)}{\pi(S^+)} f(i)^2 - \pi(S^+)^2 \bigg( \sum_{i \in S^+} \frac{\pi(i)}{\pi(S^+)} f(i) \bigg)^2
  \\
  &=&
  \pi(S^+) \pi(S^-) \sum_{i \in S^+} \frac{\pi(i)}{\pi(S^+)} f(i)^2
  +
  \pi(S^+)^2 \bigg[ \sum_{i \in S^+} \frac{\pi(i)}{\pi(S^+)} f(i)^2 - \bigg( \sum_{i \in S^+} \frac{\pi(i)}{\pi(S^+)} f(i) \bigg)^2
  \bigg]
\end{eqnarray*}
Similarly,
\begin{eqnarray*}
  \Var_\pi(f^-)
  &=&
  \pi(S^-) \pi(S^+) \sum_{i \in S^-} \frac{\pi(i)}{\pi(S^-)} f(i)^2
  +
  \pi(S^-)^2 \bigg[ \sum_{i \in S^-} \frac{\pi(i)}{\pi(S^-)} f(i)^2 - \bigg( \sum_{i \in S^-} \frac{\pi(i)}{\pi(S^-)} f(i) \bigg)^2
  \bigg]
\end{eqnarray*}
Therefore,
\begin{eqnarray*}
  & & 
  \Var_{\pi}(f)
  \\
  &=&
  \sum_{i \in X} \pi(i) f(i)^2 - 
  \left( \sum_{i \in X} \pi(i) f(i) \right)^2
  \\
  &=&
  \pi(S^+) \sum_{i \in S^+} \frac{\pi(i)}{\pi(S^+)} f(i)^2
  +
  \pi(S^-) \sum_{i \in S^-} \frac{\pi(i)}{\pi(S^-)} f(i)^2
  \\
  &&
  \quad -
  \bigg(
    \pi(S^+) \sum_{i \in S^+} \frac{\pi(i)}{\pi(S^+)} f(i)
    +
    \pi(S^-) \sum_{i \in S^-} \frac{\pi(i)}{\pi(S^-)} f(i)
  \bigg)^2
  \\
  &=&
  \Var_{\pi}(f^+) + \Var_{\pi}(f^-) - 2 \pi(S^+) \pi(S^-) 
  \bigg( \sum_{i \in S^+} \frac{\pi(i)}{\pi(S^+)} f(i) \bigg)
  \bigg( \sum_{i \in S^-} \frac{\pi(i)}{\pi(S^-)} f(i) \bigg)
  \\
  &\le&
  \Var_{\pi}(f^+) + \Var_{\pi}(f^-)
  + \pi(S^+) \pi(S^-) \bigg[
    \bigg( \sum_{i \in S^+} \frac{\pi(i)}{\pi(S^+)} f(i) \bigg)^2
    +
    \bigg( \sum_{i \in S^-} \frac{\pi(i)}{\pi(S^-)} f(i) \bigg)^2
  \bigg]
  \\
  &\le&
  \Var_{\pi}(f^+) + \Var_{\pi}(f^-)
  + \pi(S^+) \pi(S^-) \bigg[
    \bigg( \sum_{i \in S^+} \frac{\pi(i)}{\pi(S^+)} f(i)^2 \bigg)
    +
    \bigg( \sum_{i \in S^-} \frac{\pi(i)}{\pi(S^-)} f(i)^2 \bigg)
  \bigg]
  \\
  &\le&
  2\big(\Var_{\pi}(f^+) + \Var_{\pi}(f^-)\big),
\end{eqnarray*}
\end{proof}

\begin{proofof}{\autoref{thm:weighted-Cheeger}}
The easy direction $\gamma^{(1)}(G) \lesssim \psi(G)$ is proved in \autoref{lem:Cheeger-easy}.
For the hard direction, given a solution $y$ and $\directed{E}$ to $\directed{\gamma}^{(1)}(G)$, we apply the truncation step in \autoref{lem:truncation} to obtain $x \geq 0$ with $\pi(\supp(x)) \leq 1/2$ and 
$
\big( \sum_{v \in V} \pi(v) \max_{u:u \to v} (x(u)-x(v))^2 \big) / \big(\sum_{v \in V} \pi(v) x(v)^2 \big) \leq 4\directed{\gamma}^{(1)}(G).
$
Then, we apply the threshold rounding step in \autoref{prop:threshold-rounding} on $x$ to obtain a set $S$ with $S \subseteq \supp(x)$ and $\directed{\psi}(S) \lesssim \sqrt{\directed{\gamma}^{(1)}(G)}$.
If $\directed{\psi}(S) \geq 1/2$, then it implies that $\directed{\gamma}^{(1)}(G) = \Omega(1)$, and so the inequality $\psi(G)^2 \lesssim \directed{\gamma}^{(1)}(G)$ holds trivially as $\psi(G) \leq 1$ by definition.
Otherwise, if $\directed{\psi}(S) < 1/2$, we apply the postprocessing step in \autoref{lem:cleanup} on $S$ to obtain $S' \subseteq S$ with $\psi(S') \leq 2 \directed{\psi}(S)$.
Therefore, $S'$ is a set with $\pi(S') \leq \pi(S) \leq \pi(\supp(x)) \leq 1/2$ and $\psi(S') \leq 2 \directed{\psi}(S) \lesssim \sqrt{\directed{\gamma}^{(1)}(G)}$.
Thus we conclude the hard direction that $\psi(G) \leq \psi(S') \lesssim \sqrt{\directed{\gamma}^{(1)}(G)} \lesssim \sqrt{\gamma^{(1)}(G)}$. 
\end{proofof}

\section{Deferred Proofs for Bipartite Vertex Expansion} \label{app:Cheeger-bipartite}

\begin{proofof}{\autoref{lem:Cheeger-easy-bipartite}}
First we consider the case that there is an optimizer $S \subseteq V$ to $\psi_B(G)$, with bipartition $S=(S_1,S_2)$ and $\pi(\partial S)/\pi(S) = \psi_B(G)$.
Define $f: V \rightarrow \R$ and $g: V \rightarrow \R$ as follows:
\begin{itemize}
  \item $f(u) = 1$ if $u \in S_1$, and $f(u) = -1$ if $u \in S_2$, and $f(u) = 0$ if $u \not\in S$;
  \item $g(u) = 1$ if $u \in \partial S$, and $g(u) = 0$ otherwise.
\end{itemize}
For each edge $uv \in E$, we claim that $g(u) + g(v) \ge (f(u) + f(v))^2$. Note that $u$ and $v$ cannot both belong to $S_1$ or both belong to $S_2$.
One can check that the constraint is satisfied in all the remaining cases:
(1) $u \in S_1$ and $v \in S_2$ (or vice versa), (2) $u \in S$ and $v \not \in S$ (or vice versa), and (3) $u, v \not\in S$.
So, $(f, g)$ is a feasible solution to the $\nu^{(1)}(G)$ program, and the objective value is
\[
\frac{\sum_{v \in V} \pi(v) g(v)}{\sum_{v \in V} \pi(v) f(v)^2} 
= \frac{\sum_{v \in \partial S} \pi(v)}{\sum_{v \in S_1 \cup S_2} \pi(v)} 
= \frac{\pi(\partial S)}{\pi(S)} = \psi_B(S).
\]
This implies that $\nu^{(1)}(G) \leq 2 \psi_B(S)$ in this case.
The other case is when $\psi_B(G)=1$.
Choosing the feasible solution $f \equiv 1$ and $g \equiv 2$ to the $\nu^{(1)}(G)$ program shows that $\nu^{(1)}(G) \leq 2$.
This implies that $\nu^{(1)}(G) \leq 2\psi_B(G)$ in the other case.
\end{proofof}

\begin{proofof}{\autoref{thm:Cheeger-nu}}
The easy direction $\nu^{(1)}(G) \lesssim \psi_B(G)$ is proved in \autoref{lem:Cheeger-easy-bipartite}.
For the hard direction, given a solution $x$ and $\directed{E}$ to $\directed{\nu}^{(1)}(G)$, we apply the threshold rounding step in \autoref{prop:threshold-rounding-nu} on $x$ to obtain two disjoint sets $S_1,S_2 \subseteq V$ with $\directed{\psi}(S_1,S_2) \lesssim \sqrt{\directed{\nu}^{(1)}(G)}$.
If $\directed{\psi}(S_1,S_2) \geq 1/2$, then it implies that $\directed{\nu}^{(1)}(G) = \Omega(1)$, and so the inequality $\psi_B(G)^2 \lesssim \directed{\nu}^{(1)}(G)$ holds trivially as $\psi_B(G) \leq 1$ by definition.
Otherwise, if $\directed{\psi}(S_1,S_2) < 1/2$, we apply the postprocessing step in \autoref{lem:cleanup-bipartite} on $(S_1,S_2)$ to obtain $(S_1',S_2')$ so that $S_1' \cup S_2'$ is an induced bipartite graph in $G$ and $\psi(S_1' \cup S_2')\leq 2\directed{\psi}(S_1,S_2)$.
Thus we conclude the hard direction that $\psi_B(G) \leq \psi_B(S_1' \cup S_2') \leq 2\directed{\psi}(S_1,S_2) \lesssim \sqrt{\directed{\nu}^{(1)}(G)}$.  
\end{proofof}

\section{Deferred Proofs for Multiway Vertex Expansion} \label{app:higher-order-vertex}

\begin{proofof}{\autoref{prop:max-of-gaussians}}
By the Laurent-Massart bound of $\chi$-squared distribution (c.f. Lemma 1 of \cite{LM00}), for any $\delta > 0$ and $i \le d$,
\[
  \Pr\big[Y_i - 1 \ge 2 \sqrt{\delta / m} + 2 (\delta / m)\big] \le e^{- \delta}.
\]
Since
\[
  1 + 2 \sqrt{\delta / m} + 2 (\delta / m) \le 4 + 4 (\delta / m)
\]
it follows that
\[
  \Pr[Y_i \ge 4 + 4 (\delta / m)] \le e^{- \delta}.
\]
Recall that $Y = \max_{i \le d} Y_i$. 
Taking union bound over the $Y_i$'s,
\[
  \Pr[Y \ge 4 + 4 (\delta / m)] \le d \cdot e^{-\delta}.
\]
With this probability tail bound, the expectation of $Y$ can be bounded as follows:
\begin{eqnarray*}
  \E[Y]
  &=&
  \int_{0}^{\infty} \Pr[Y \ge t] \, dt
  ~\le~
  4 + \int_{0}^{\infty} \Pr[Y \ge 4 + t] \, dt
  ~\le~
  4 + \int_{0}^{\infty} \min\big\{1, d \cdot e^{-mt/4}\big\} \, dt \quad
  \\
  &=&
  4 + \frac{4\log d}{m} + \int_{\frac{4 \log d}{m}}^{\infty} d \cdot e^{-mt/4} \, dt
  ~=~
  4 + \frac{4\log d}{m} - \frac{4d}{m} \cdot e^{-mt/4} \Big|_{\frac{4 \log d}{m}}^{\infty}
  ~=~
  4 + \frac{4\log d}{m} + \frac{4}{m}.
\end{eqnarray*}
\end{proofof}

\begin{proofof}{\autoref{lem:higher-order-easy}}
If $\psi_k(G) \geq 1$, then the lemma holds trivially as $\lambda_k^*(G) \leq 2$.
Henceforth, we assume $\psi_k(G) < 1$,
and there are nonempty disjoint subsets $S_1, \dots, S_k \subseteq V$ with
$\max_{1 \leq i \le k} \psi(S_i) = \psi_k(G)$.

Using the notation in \autoref{def:sigma-k-dual},
the $\lambda_k^*$ program in \autoref{def:reweighted-lambda-k} can be written as
\begin{align*}
\lambda_k^*(G) ~:=~ \max_{Q \geq 0}  \min_{f_1, \dots, f_k: V \rightarrow \R} \max_{1 \leq i \leq k} &~~~ f_i^T (\Pi - Q) f_i  & 
\\
\st
&~~~ \sum_{v} \pi(v) f_i(v)^2 = 1 & & \forall 1 \leq i \leq k
\\
&~~~ \sum_{v} \pi(v) f_i(v) f_j(v) = 0 & & \forall 1 \leq i \neq j \leq k
\\  
&~~~ Q(u,v) = 0 & & \forall uv \notin E
\\
&~~~ \sum_{v \in V} Q(u,v) = \pi(u) & & \forall u \in V
\\
&~~~ Q(u,v) = Q(v,u) & & \forall uv \in E.
\end{align*}
Each $f_i^T (\Pi - Q) f_i$ can be written as
\[
  \sum_{uv \in E} Q(u,v) \big(f_i(u) - f_i(v)\big)^2 = 
  \sum_{uv \in E} \pi(u) P(u, v) \big(f_i(u) - f_i(v)\big)^2.
\]
We set
\[
  f_i(u) := \begin{cases}
    \frac{1}{\sqrt{\pi(S_i)}}, & \text{ if $u \in S_i$ }
    \\
    0, & \text{ otherwise.}
  \end{cases}
\]
We can check that the constraints on $f_i$ are satisfied. 
Moreover, for any $P$ satisfying the constraints,
\begin{eqnarray*}
  \sum_{uv \in E} \pi(u) P(u, v) \big(f_i(u) - f_i(v)\big)^2
  & = &
  \sum_{v \in S_i, u \in \partial(S_i)} \pi(u) P(u, v) \frac{1}{\pi(S_i)}
\\ 
  & \le &
  \frac{1}{\pi(S_i)} \sum_{v \in V, u \in \partial(S_i)} \pi(u) P(u, v)
  ~=~
  \frac{\pi(\partial(S_i))}{\pi(S_i)}
  ~=~
  \psi(S_i).
\end{eqnarray*}
So, $\max_{1 \leq i \le k} f_i^T (\Pi - \Pi P) f_i 
\leq \max_{1 \leq i \leq k} \psi(S_i)$. 
Taking maximum over $P$ gives $\lambda_k^*(G) \le \psi_k(G)$.
\end{proofof}

\section{Deferred Proofs for 0/1-Polytopes with Poor Vertex Expansion} \label{app:01-polytope}

\begin{proofof}{\autoref{lem:blocking-one-edge}}
The plan is to prove that the stated conditions imply $\frac12 (x+y) \in {\rm conv}(M)$, which would then immediately imply that there is no edge connecting $x$ and $y$ in $Q$.
We will prove the contrapositive: if $\frac12 (x+y) \notin {\rm conv}(M)$,
then there is a $p$-consistent affine function $l$ such that $l(z) < 0$ for all $z \in M$ matching the pattern $p$.

Denote $w := \frac12 (x+y)$.
As $w \notin {\rm conv}(M)$, by \autoref{prop:separation-thm},
there is an affine function $l': (u_1, \dots, u_n) \mapsto \beta' + \sum_i \alpha_i' u_i$ such that $l'(w) = 0$ and $l'(z) < 0$ for all $z \in M$.
We would like to modify $l'$ to obtain an affine function $l: (u_1, \dots, u_n) \mapsto \beta + \sum_i \alpha_i u_i$ such that (i) $l(w) = 0$, (ii) $\alpha_i=0$ for $i \in \supp(p)$, and (iii) $l(z) < 0$ for all $z \in M$ matching the pattern $p$.

Note that, by the definition of common pattern, 
for $i \in \supp(p)$, either $x_i = y_i = 1$ or $x_i = y_i = 0$,
and so $w_i \in \{0, 1\}$. 
Also, for any $z \in \{0, 1\}^n$ that matches the pattern $p$, 
we must have $z_i = w_i$ for $i \in \supp(p)$. 
So, for any such $z$,
\begin{eqnarray*}
    l'(z)
    =
    \beta' + \sum_{i=1}^n \alpha'_i z_i 
    =
    \bigg( \beta' + \sum_{i \in \supp(p)} \alpha'_i z_i \bigg) + \sum_{i \not\in \supp(p)} \alpha'_i z_i 
    =
    \bigg( \beta' + \sum_{i \in \supp(p)} \alpha'_i w_i \bigg)
    + \sum_{i \not\in \supp(p)} \alpha'_i z_i.
\end{eqnarray*}
Hence, 
if we set $\beta := \beta' + \sum_{i \in \supp(p)} \alpha_i' w_i$, 
and $\alpha_i = \alpha'_i$ for $i \notin \supp(p)$ and $\alpha_i = 0$ for $i \in \supp(p)$,
then the affine function $l: (u_1, u_2, \dots, u_n) \in \R^n \mapsto \beta + \sum_i \alpha_i u_i$ satisfies $l(z) = l'(z)$ for any $z$ that matches the pattern $p$.
Therefore, $l$ is an affine function that satisfies the three properties that (i) $l(w)=l'(w)=0$, (ii) $\alpha_i=0$ for $i \in \supp(p)$, and (iii) $l(z) = l'(z) < 0$ for $z \in M$ matching the pattern $p$. 
\end{proofof}

\begin{proofof}{\autoref{lem:blocking-all-edges}}
Note that for $x \in L$ and $y \in R$ in our construction in \autoref{def:construction}, 
their common pattern $p$ satisfies $|\supp_0(p)| = |\supp_1(p)| \le k$.
The lemma follows by applying \autoref{lem:blocking-one-edge} on all possible such patterns.
\end{proofof}

\begin{proofof}{\autoref{lem:good}}
Let $z \in \{0, 1\}^n$ be such that $|z| = \frac{n}{2}$ and $z$ matching the pattern $p$. 
Let $z^{\tau} \in \{0, 1\}^n$ be its ``opposite point'' formed by toggling the coordinates of $z_i$ for $i \not\in \supp(p)$ and leaving other coordinates unchanged. 
Note that the lemma follows from the following two facts:
(i) $|z^{\tau}| = \frac{n}{2}$ and $z^{\tau}$ matches the pattern $p$, and
(ii) $l(z) + l(z^{\tau}) = 0$.
  
For the first fact, $z^{\tau}$ matches the pattern $p$ because $z^{\tau}_i = z_i$ for $i \in \supp(p)$. 
And $|z^{\tau}| = \frac{n}{2}$ because $|z| = \frac{n}{2}$ and there are the same number of zeroes and ones in $p$, the latter being a consequence of $|x| + |y| = k+(n-k) = n$.
  
For the second fact, as $l$ is $p$-consistent,
  \begin{eqnarray*}
    l(z) + l(z^{\tau})
    =
    \bigg( \beta + \sum_{i \not\in \supp(p)} \alpha_i z_i \bigg)
    +
    \bigg( \beta + \sum_{i \not\in \supp(p)} \alpha_i z_i^{\tau} \bigg)
    =
    2 \bigg( \beta + \frac{1}{2} \sum_{i \not\in \supp(p)} \alpha_i \bigg)
    =
    0,
  \end{eqnarray*}
where the last equality is from the second condition in \autoref{def:consistent}.
\end{proofof}

\begin{proofof}{\autoref{lem:match-pattern}}
  The number of points $z \in \{0, 1\}^n$ with $|z| = \frac{n}{2}$ is ${n \choose n/2}$, whereas the number of such points that matches pattern $p$ is ${n - 2s \choose n/2 - s}$. Therefore,
  \begin{eqnarray*}
    \Pr_{z \sim Z} [z \text{ matches pattern } p]
    ~=~
    {n - 2s \choose n/2 - s} \Big/{n \choose n/2}
    ~\gtrsim~
    \bigg( \sqrt{\frac{2}{\pi(n-2s)}} \cdot 2^{n-2s} \bigg) \Big/
    \bigg( \sqrt{\frac{2}{\pi n}} \cdot 2^{n} \bigg)
    ~\ge~ 4^{-s},
  \end{eqnarray*}
where we used Stirling's approximation $n! \approx_n \sqrt{2 \pi n} (n/e)^n$.
\end{proofof}

\begin{proofof}{\autoref{thm:01-vertex-expansion}}
Let $Q$ be a $0/1$-polytope from \autoref{def:construction}.
We would like to apply \autoref{lem:blocking-all-edges} to prove the theorem.

Let $Z$ be the uniform distribution on $\{z \in \{0, 1\}^n: |z| = n/2\}$.
Combining \autoref{lem:good} and \autoref{lem:match-pattern}, 
it follows that for any pattern $p$ with $|\supp_0(p)| = |\supp_1(p)| = s \le k$ and any $p$-consistent affine function $l$,
\begin{eqnarray*}
  \Pr_{z \sim Z} \left[ \, 
      l(z) \ge 0 \text{ and } z \text{ matches pattern } p \,
    \right]
  \geq
  c \cdot 4^{-s}
\end{eqnarray*}
for some universal constant $c > 0$.
Therefore, if we take independent samples $z_1, \dots, z_m \sim Z$ and set $M := \{z_1, \dots, z_m\}$ where $m$ is a value to be determined later, 
then for any pattern $p$ with $|\supp_0(p)| = |\supp_1(p)| = s \le k$ and any $p$-consistent affine function $l$,
\begin{eqnarray*}
  \Pr[ \nexists z \in M \text{ with } l(z) \ge 0 \text{ and } z \text{ matching pattern } p]
  ~\leq~
  (1 - c \cdot 4^{-s})^m
  ~\leq~
  (1 - c \cdot 4^{-k})^m.
\end{eqnarray*}
To apply a union bound, we upper bound the numbers of different such $p$ and $l$.
The number of patterns $p$ with $|\supp_0(p)| = |\supp_1(p)| = s \le k$ is
\[
  \sum_{s = 0}^k {n \choose 2s} \cdot {2s \choose s} \le (k+1) \cdot n^{2k}.
\]
The number of $p$-consistent affine functions $l$ with different sign patterns on the boolean hypercube $\{0,1\}^n$ is upper bounded by the number of affine threshold functions on $\{0, 1\}^n$, which is at most $2^{n^2}$ by \autoref{prop:Cover}.
Combining the two estimates and the above probability bound,
the failure probability is
\[
\Pr[ \exists p~\exists l \text{ s.t. } \nexists z \in M \text{ with } l(z) \ge 0 \text{ and } z \text{ matching pattern } p]
\le (1 - c \cdot 4^{-k})^m \cdot (k+1) \cdot n^{2k} \cdot 2^{n^2}.
\]
Setting
\[
  m ~=~ \frac{4^k}{c} \cdot (1 + \log(k+1) + 2k \log n + n^2 \log 2) ~\lesssim~ 4^k n^2,
\]
the failure probability is at most
\begin{eqnarray*}
  & & (1 - c \cdot 4^{-k})^m \cdot (k+1) \cdot n^{2k} \cdot 2^{n^2}
  \\
  &\le&
  \exp(-c \cdot 4^{-k} \cdot m) \cdot \exp(\log (k+1) + 2k \log n + n^2 \log 2)
  \\
  &\le&
  \exp(-(1 + \log (k+1) + 2k \log n + n^2 \log 2)) \cdot \exp(\log (k+1) + 2k \log n + n^2 \log 2)
  \\
  &=&
  e^{-1}.
\end{eqnarray*}
Therefore, by \autoref{lem:blocking-all-edges}, we conclude that there exists $M \subseteq \{z \in \{0, 1\}^n : |z| = n/2\}$ with $|M| \lesssim 4^k n^2$ such that there are no edges between $L$ and $R$ in the graph of $Q$.

Finally, as $|L| = {n \choose k} \ge (n/k)^k$ and $\partial L \subseteq M$, it follows that
\[
\psi(L) \le \frac{|M|}{|L|} \lesssim \frac{4^k n^2}{(n/k)^k} = \frac{(4k)^k}{n^{k-2}}.
\]
\end{proofof}

\section{Deferred Proofs for Spherical Proximity Graph} \label{app:tight-example}

\begin{proofof}{\autoref{prop:sphere-graph-construction}}
The construction largely follows that of {\cite[Lemma 8.3.22]{GM12}}.
Given $n$ and $\gamma > 0$, we iteratively choose points $y_1, y_2, \dots, y_m \in \S^{k-1}$, such that each new point $y_{i+1}$ has distance at least $\gamma/2$ from $y_1, \dots, y_i$. We stop when it is no longer possible to choose a point that is $(\gamma/2)$-far from all existing points.

We now let $S_1', S_2', \dots, S_m'$ be the cells of the Voronoi diagram of $y_1, y_2, \dots, y_m$. That is, for any $x \in \S^{k-1}$, $x \in S_i'$ iff $\Delta(x, y_i) = \min_{j \in [m]} \Delta(x, y_j)$.
Note that cell $S_i'$ contains $B(y_i, \gamma/4)$ and is contained in $B(y_i, \gamma/2)$. Therefore, the measure of each $S_i'$ is at least $\eps = \mu(Cap(\gamma/4))$ and the diamater of each $S_i'$ is at most $\gamma$.
By further subdividing the cells (evenly) until each cell has measure $\le 2 \eps$, we obtain $S_1, S_2, \dots, S_n$, such that the measure of each $S_i$ is between $\eps$ and $2 \eps$, and ${\rm diam}(S_i) \le \gamma$ for each $i \in [n]$.
We can choose the points $x_i \in S_i$ arbitrarily.
\end{proofof}

\begin{proofof}{\autoref{lemma:degree-expansion-volume}}
  First, we prove the bound on the maximum degree. For any vertex $i \in V$, its degree is equal to the number of points $x_j$ that are $\delta$-close to $x_i$. We can count the number of such points using volume estimation.
  If a point $x_j$ is within distance $\delta$ from $x_i$, then the entire cell $S_j$ is within distance $\delta + \gamma$ from $x_i$. This contributes at least $\eps = \mu(Cap(\gamma/4))$ total measure to $B(x_i, \delta + \gamma)$.
  The total measure of these cells is at most $\mu(B(x_i, \delta + \gamma))$, so
  \[
    |\partial(i)| \cdot \mu(Cap(\gamma/4)) \le \mu(B(x_i, \delta + \gamma)).
  \]
  Rearranging gives the desired upper bound on $|\partial(i)|$, and thus on the maximum degree.
  
  Next, we prove the bound on the vertex expansion. Given any $T \subseteq V$, we wish to lower bound $|\partial(T)|$. For $j \in V$, if the cell $S_j$ is completely contained in $S_T + Cap(\delta)$, then $x_j \in T \cup \partial(T)$.
  If we take the union of all such cells, then the set will contain $S_T + Cap(\delta - \gamma)$:
  \[
    \cup \{ S_j : S_j \subseteq S_T + Cap(\delta) \}
    \supseteq S_T + Cap(\delta - \gamma).
  \]
  
  This is because, for any point $u \in S_T + Cap(\delta - \gamma)$, the cell that it is in will be contained in
  \[
    \{u\} + Cap(\gamma) \subseteq S_T + Cap(\delta - \gamma) + Cap(\gamma) = S_T + Cap(\delta).
  \]
  
  It follows that
  \[
    \cup \{ S_j : j \not\in T \text{ and } S_j \subseteq S_T + Cap(\delta) \} \supseteq (S_T + Cap(\delta - \gamma)) \setminus S_T.
  \]
  
  Combining previous observations, and since each cell has measure at most $2\eps$,
  \begin{eqnarray*}
    |\partial(T)|
    =
    |(T \cup \partial(T)) \setminus T|
    &\ge&
    |\{S_j : j \not\in T \text{ and } S_j \subseteq S_T + Cap(\delta)\}|
    \\
    &\ge&
    \frac{\mu(S_T + Cap(\delta - \gamma)) - \mu(S_T)}{2 \eps}.
  \end{eqnarray*}
  
  We are done after substituting this and $|T| \le \mu(S_T) / \eps$ into $\psi(T) = |\partial(T)| / |T|$.
\end{proofof}

\begin{proofof}{\autoref{prop:degree-bound-sphere-example}}
  We will use the following formula for spherical cap volume:
  \[
    \mu(Cap(x)) = \frac
    {\int_0^{x} \sin^{k-2} \theta \, d \theta}
    {\int_0^{\pi} \sin^{k-2} \theta \, d \theta}.
  \]
  
  This comes from the formula for unnormalized spherical cap volume \cite{Li11}:
  \[
    \mu_0(Cap(x)) = \frac{2 \pi^{(k-1)/2}}{\Gamma((k-1)/2)} \cdot \int_0^x \sin^{k-2} \theta \, d \theta.
  \]
  
  We shall use the approximation
  \[
    \sin^k \theta \approx_k \left( \theta - \frac{\theta^3}{3!} \right)^k
  \]
  
  for $0 \le \theta \le O(1/\sqrt{k})$.
  The third-degree approximation is sufficient because $(\theta - \theta^3 / 3! + O(\theta^5))^k \approx_k (\theta - \theta^3 / 3!)^k$ for $\theta$ in this range.
  
  Then, for $x = O(1 / \sqrt{k})$,
  \begin{eqnarray*}
    \int_{0}^x \sin^k \theta \, d \theta
    &\approx_k&
    \int_{0}^x \left( \theta - \frac{\theta^3}{3!} \right)^k \, d \theta
    \\
    &\approx_k&
    \int_{0}^x \left( \theta - \frac{\theta^3}{3!} \right)^k \cdot (1 - \theta^2 / 2) \, d \theta
    \quad (\because (1 - \theta^2 / 2) \text{ is close to } 1)
    \\
    &\overset{y := \theta - \theta^3/6}{=}&
    \int_{0}^{x - x^3 / 6} y^k \, dy
    \quad (\because y(\theta) := \theta - \theta^3 / 6 \text{ is increasing for $\theta \in [0, x]$})
    \\
    &=&
    \frac{(x-x^3 / 6)^{k+1}}{k+1}.
  \end{eqnarray*}
  
  Therefore,
  \begin{eqnarray*}
    \frac{\mu(Cap(\delta + \gamma))}
                {\mu(Cap(\gamma/4))}
    &=&
    \frac{\int_0^{\delta + \gamma} \sin^{k-2} \theta \, d \theta}
    {\int_0^{\gamma/4} \sin^{k-2} \theta \, d \theta}
    \\
    &\approx_k&
    \left( \frac{\delta + \gamma}{\gamma/4} \right)^{k-1} \cdot \left( \frac{1 - (\delta + \gamma)^2/6}{1 - (\gamma / 4)^2 / 6} \right)^{k-1}
    \\
    &=&
    \left( \frac{4(c_1 + c_2)}{c_1} \right)^{k-1} \cdot \left( \frac{1 - (c_1 + c_2)/6k}{1 - c_1 / 96k} \right)^{k-1}
    \\
    &\lesssim&
    \left( \frac{4(c_1 + c_2)}{c_1} \right)^k
    \\
    &\le&
    2^{O(k)}.
  \end{eqnarray*}
\end{proofof}

\begin{proofof}{\autoref{prop:expansion-bound-sphere-example}}
  By isoperimetric results on the sphere, the worst case is when the set $S_T$ is a spherical cap. Therefore, it suffices to prove that, for any $\tau \in (0, \pi / 2]$,
  \[
    \mu(Cap(\tau + (\delta - \gamma))) \ge (1 + \Omega(1)) \cdot \mu(Cap(\tau)).
  \]
  Let $c := c_2 - c_1 > 0$. This is equivalent to
  \[
    \int_0^{\tau + c / \sqrt{k}} \sin^{k-2} \theta \, d \theta \ge (1 + \Omega(1)) \int_0^{\tau} \sin^{k-2} \theta \, d \theta.
  \]
  
  For technical reasons, we first deal with the case where $\tau$ is close to $\pi/2$. For $\tau \ge \pi / 2 - c/\sqrt{k}$, the result follows from well-known upper bounds on spherical cap volume. For example, we may use the upper bound in \cite{Tko12}:
  \[
    \mu(Cap(\pi/2 - \theta)) \le e^{-k \sin^2 \theta /2}, \qquad \theta \in [0, \pi/2)
  \]
  
  and the fact that $\sin \theta \approx_k \theta$ for $\theta = O(1/\sqrt{k})$.
  Then, if $c$ is such that
  $\exp \left[-k \sin^2 \left(\frac{c}{2\sqrt{k}} \right) / 2 \right] \le 1/3$,
  \[
    \frac{\mu(Cap(\tau + c/\sqrt{k}))}{\mu(Cap(\tau))} \ge
    \min \left(
    \frac{\mu(Cap(\pi/2))}{\mu(Cap(\pi/2 - c/2\sqrt{k}))}
    ,
    \frac{\mu(Cap(\pi/2 + c/2\sqrt{k}))}{\mu(Cap(\pi/2))}
    \right)
    \ge
    \frac{4}{3}.
  \]
  Therefore, we may assume that $\tau < \pi/2 - c/\sqrt{k}$. \\
  
  We will actually prove the following relation: for all $x \in [0, \pi/2 - c/\sqrt{k}]$,
  \[
    \int_{x}^{x + c / \sqrt{k}} \sin^{k-2} \theta \, d \theta \ge (1 + \Omega(1)) \int_{x - c / \sqrt{k}}^x \sin^{k-2} \theta \, d \theta.
  \]
  
  If this relation is proven to be true, then by writing
  \[
    \int_0^{\tau + c / \sqrt{k}} \sin^{k-2} \theta \, d \theta
    =
    \left(
    \int_{\tau}^{\tau + c / \sqrt{k}} 
    + \int_{\tau - c / \sqrt{k}}^{\tau}
    + \cdots
    + \int_{0}^{\tau - t c / \sqrt{k}}
    \right)
    \sin^{k-2} \theta \, d \theta,
    \quad
    (
      t := \floor{\frac{\tau \sqrt{k}}{c}}
    ),
  \]
  and applying the above relation to each term on RHS, we obtain the desired result.
  
  In order to prove the relation, we show that, for all $x \in [0, \pi/2 - c/\sqrt{k}]$,
  \[
    \sin^{k-2} (x + c / \sqrt{k}) \ge
    (1 + \Omega(1)) \cdot \sin^{k-2} x.
  \]
  
  that is, we wish to show that the function
  \[
    f(x) := \frac{\sin^{k-2}(x + c/\sqrt{k})}{\sin^{k-2}(x)}
  \]
  
  is at least $1 + \Omega(1)$ for $x \in [0, \pi/2 - c/\sqrt{k}]$.
  
  Check that, for $k \ge 3$, $f(x)$ is decreasing by differentiating $\sin(x + c/\sqrt{k}) / \sin(x)$.
  It remains to compute
  \begin{eqnarray*}
  f(\pi/2 - c/\sqrt{k})
  &=&
  \frac{\sin^{k-2} (\pi/2)}{\sin^{k-2}(\pi/2 - c/\sqrt{k})}
  \\
  &=&
  [\cos(c/\sqrt{k})]^{-(k-2)}
  \\
  &=&
  \left( 1 - \frac{c^2}{2k} \right)^{-(k-2)} + o(1)
  \qquad (\text{by Taylor expansion})
  \\
  &=&
  \exp(c^2/2) + o_k(1).
  \end{eqnarray*}
  
  Therefore, for sufficiently large $k$, $f(x) \ge 1 + \Omega(1)$, where the constant depends on $c$. 
\end{proofof}

\begin{proofof}{\autoref{prop:eigenvalue-bound-sphere-example}}
We wish to construct test vectors $\{v_i\}_{i \in [n]}$, such that $\sum_{i \in [n]} v_i = \vec{0}$ and
\[
  \frac{\sum_{(i, j) \in E} P(i, j) \norm{v_i - v_j}^2}
  {\sum_{i \in [n]} \norm{v_i}^2}
\]

is small for any doubly stochastic reweighting $P$ of the graph $G_n$.

We claim that setting $v_i := x_i - z$ works, where $z := \frac{1}{n} \sum_{j \in [n]} x_j$.
By construction, $\sum_{i \in [n]} v_i = \vec{0}$. We next show that
\[
  \sum_{i \in [n]} \norm{v_i}^2 \ge \Omega(n).
\]
Since $\sum_{i \in [n]} v_i = \vec{0}$, note that
\[
  \sum_{i \in [n]} \norm{v_i}^2
  = \frac{1}{2n} \sum_{i, j \in [n]} \norm{v_i - v_j}^2
  = \frac{1}{2n} \sum_{i, j \in [n]} \norm{x_i - x_j}^2.
\]

It then suffices to show that
\[
  \sum_{j \in [n]} \norm{x_i - x_j}^2 \ge \Omega(n)
\]

for all $i \in [n]$. This follows from the fact that at least $\frac{n}{3}$ of the points $x_j$ are of geodesic distance at least $\frac{\pi}{2} - \gamma \ge \Omega(1)$ from $x_i$, and so are of Euclidean distance $\Omega(1)$ from $x_i$.

Therefore, for any doubly stochastic reweighting $P$,
\begin{eqnarray*}
  \lambda_2(P)
   \le 
  \frac{\sum_{(i, j) \in E} P(i, j) \norm{v_i - v_j}^2}
  {\sum_{i \in V} \norm{v_i}^2}  
  \le
  \frac{\sum_{i \in [n]} \sum_{j: (i, j) \in E} P(i, j) \cdot \delta^2}{\Omega(n)}
  = O(\delta^2).
\end{eqnarray*}

Note that, in the second inequality, we used the fact that straight-line distance in $\R^k$ is less than or equal to geodesic distance in $\S^{k-1}$.

We conclude that $\lambda_2^*(G) \lesssim \delta^2 \lesssim 1/k$. Combining with the degree bound in Proposition \ref{prop:degree-bound-sphere-example},
\[
 \lambda_2^*(G) \lesssim \frac{1}{k} \lesssim \frac{1}{\log d}.
\]
\end{proofof}

\end{appendix}

\newpage

\bibliographystyle{plain}

\end{document}